\documentclass[11pt,letterpaper]{amsart}
\usepackage{enumitem}
\usepackage{color}
\usepackage[numbers,sort&compress]{natbib}
\usepackage{pifont,latexsym,ifthen,amsthm,rotating,calc,textcase,booktabs}
\usepackage{amsfonts,amssymb,amsbsy,amsmath}
\usepackage{amsthm}
\numberwithin{equation}{section}
\usepackage[dvipsnames, svgnames, x11names]{xcolor}
\usepackage{graphicx}
\usepackage{tikz}
\usepackage{stmaryrd}
\usepackage{graphicx}
\usepackage[latin1]{inputenc}
\usepackage[T1]{fontenc}
\usepackage[english]{babel}
\usepackage{listings}
\usepackage{xcolor,mathrsfs,url}
\usepackage{ifthen}
\usepackage{fancyhdr}
\usepackage{verbatim}
\usepackage{todonotes}
\usepackage{float}
\usepackage{subfigure}
\usepackage{hyperref}
\usepackage{appendix}
\usepackage[percent]{overpic}

\newtheorem{theorem}{Theorem}[section]
\newtheorem{lemma}[theorem]{Lemma}
\newtheorem{corollary}[theorem]{Corollary}
\newtheorem{remark}[theorem]{Remark}
\newtheorem{proposition}[theorem]{Proposition}

\newtheorem{prob}{RH problem}[section]
\newtheorem{prob2}{$\bar{\partial}$-RH problem}[section]
\newtheorem{prob3}{$\bar{\partial}$-problem}[section]
\usepackage {caption}

\hypersetup{hypertex=true,
            colorlinks=true,
            linkcolor=blue,
            anchorcolor=blue,
            citecolor=red}

\DeclareMathOperator*{\im}{Im}
\DeclareMathOperator*{\re}{Re}

\begin{document}

\title[Painlev\'{e}  asymptotics for  defocusing NLS equation]{ The defocusing Schr\"odinger equation with a nonzero background:  Painlev\'{e}  asymptotics  in two transition regions}
\author{Zhaoyu Wang$^1$ \  and  \  Engui FAN$^{1}$  }
\footnotetext[1]{ \  School of Mathematical Sciences  and Key Laboratory   for Nonlinear Science, Fudan   University, Shanghai 200433, P.R. China.  Corresponding E-mal: faneg@fudan.edu.cn }
\date{  }

\maketitle
\begin{abstract}

 In this paper, we address the  Painlev\'e asymptotics  in the  transition region $|\xi|:=\big|\frac{x}{2t}\big| \approx 1$ to the Cauchy problem of  the defocusing
   Schr$\ddot{\text{o}}$dinger (NLS) equation with a nonzero background.
With the  $\bar\partial$-generation   of the nonlinear steepest descent approach and double scaling limit  to
 compute the long-time
 asymptotics of the solution in two transition regions defined as
$$ \mathcal{P}_{\pm 1}(x,t):=\{ (x,t) \in \mathbb{R}\times\mathbb{R}^+, \ \ 0<|\xi-(\pm 1)|t^{2/3}\leq C\}, $$
    we find that the long-time asymptotics  of the  NLS equation  in both  transition  regions $ \mathcal{P}_{\pm 1}(x,t)$
can be expressed in terms of the   Painlev\'{e} II equation. We are also able to express the
leading term explicitly in terms of the Airy function.

\end{abstract}

\tableofcontents

\quad

  \section{Introduction}
  \hspace*{\parindent}

 The present paper is concerned with  the  Painlev\'e asymptotics  in a transition region $|\xi|:=\big|\frac{x}{2t}\big| \approx 1$
for the Cauchy problem of  the defocusing   Schr$\ddot{\text{o}}$dinger  (NLS) equation with a nonzero background
	\begin{align}
&iq_t+q_{xx}-2(|q|^2-1)q=0,\label{q}\\
&	q(x,0)=q_0(x) \sim  \pm 1, \ \ x\to \pm \infty.\label{inq}
	\end{align}
It is an important model that has been discussed extensively in classical texts like the book by Faddeev and Takhtajan \cite{ft1987}.  Zakharov and Shabat first derived  the Lax pair  of the NLS equation \cite{ZS1}.  The  well-posedness of the NLS equation   with  the initial data in Sobolev spaces   was proved   \cite{Tsutsumi,Bourgain}.
Zakharov and Shabat developed  the inverse scattering transform (IST)   for the focusing NLS equation  with zero boundary conditions \cite{ZS2}.
  We mention the following works  on the long-time asymptotics of the defocusing NLS equation.
    For the initial data in  the Schwartz space $\mathcal{S}(\mathbb{R})$ and   using the IST method,  Zakharov and Manakov
  obtained the long-time asymptotics  of the  NLS equation  \cite{ZS3}.
   In 1981, Its presented a stationary phase method
   to analyze the long-time asymptotic behavior   for the NLS equation \cite{its1}.
    In 1993,  Deift and Zhou  developed a  nonlinear steepest descent method to rigorously obtain the long-time asymptotic behavior   for the mKdV equation
 \cite{sdmRHp}.   Later this method was extended  to get the leading  and high-order   asymptotic behavior for the solution of the  NLS equation (\ref{q})
 with the initial data  $q_0  \in \mathcal{S}(\mathbb{R})$  \cite{PX2,PX21}.
  Vartanian  obtained  the leading and first correlation terms in the asymptotic behavior of
 the NLS equation with finite density initial data \cite{VAH1,VAH2,VAH3}.
Under  much weaker  weighted Sobolev initial data $q_0 \in H^{1,1} (\mathbb{R})$,  Deift and Zhou gave the leading  asymptotics, where
the error is  $\mathcal{O}(t^{-1/2-\kappa}),\ \  0<\kappa<1/4$  \cite{PX3}.
 In 2008,  for the same space  $q_0\in H^{1,1}(\mathbb{R})$, Dieng and McLaughlin applied the  $\bar{\partial}$-steepest descent method to
 obtain  a sharp estimate, where
the error is  $\mathcal{O}(t^{-3/4}) $   \cite{MK}.
Jenkins  investigated the long-time/zero-dispersion limit of
the solutions to   the defocusing NLS equation  associated with the step-like  initial data   \cite{Jenkins}.
Fromm, Lenells, and Quirchmayr  studied the long-time asymptotics  for the defocusing NLS equation with  the step-like boundary condition \cite{Le}.

In the study of Painlev\'e equations, one of the important topics is the asymptotic behavior
of their solutions  \cite{Flaschka1,Jimbo1,Fokas2,Its1,Deift1}.
Especially some  Painlev\'e equations, for example, the homogeneous Painlev\'e II equation,
can be solved via a certain Riemann-Hilbert (RH) problem \cite{Fokas1}.
Although the study of the Painlev\'e equations originates from  pure mathematics,
they have found plenty of applications in many fields of mathematical physics, such
as random matrix theory, statistical physics, and integrable   systems.

The  transition asymptotic regions for  the  Korteweg-de Vries equation   were first  described in Painlev\'e transcendents
by Segur and Ablowitz in \cite{AblwzP1}.  The connection between different regions was first understood in the case of the
  modified Korteweg-de Vries equation  by Deift and Zhou \cite{DeiftP2}.
Boutet de Monvel et al. discussed the Painlev\'{e}-type asymptotics for the Camassa-Holm equation by the nonlinear steepest descent method \cite{Monvel}.
The connection  between    the  tau-function of the Sine-Gordon  reduction and the Painlev\'e III equation was given   by the RH  approach \cite{Its3}.
 Charlier and   Lenells carefully  considered the Airy and higher order Painlev\'{e} asymptotics for the mKdV equation  \cite{Charlier}.
  Recently,   Huang and  Zhang  obtained
  Painlev\'{e} asymptotics   for the whole  mKdV hierarchy \cite{Huanglin}.

In the  study of the fundamental rogue wave solutions in the limit of a large order of the focusing NLS
equation with nonzero background
   \begin{align}
&iq_t+  \frac{1}{2} q_{xx}+2q(|q|^2-1)=0,   \nonumber\\
&q(x,t) \sim   1, \ \ x\rightarrow\pm \infty,\nonumber
\end{align}
 Bilman, Ling, and Miller found that the limiting profile of
the rogue wave of  infinite order  satisfies ordinary differential equations concerning
space and time. The spatial differential equations were identified with certain members
of the Painlev\'e-III hierarchy \cite{Miller1}. Recently,
  considering  the focusing NLS equation with step-like oscillating background,
   Boutet de Monvel,   Lenells,  and Shepelsky  found
  the  asymptotics in a transition zone  between two genus-3 sectors.
   Especially   a local parametrix was constructed  by solving an RH  model problem  associated with the Painlev\'e IV equation
   \cite{Monvel5}.
In this paper, we are concerned with  the Painlev\'e transcendent to the defocusing NLS equation,
which  seems  still  unknown  to the best of our knowledge.

       \begin{figure}
	\begin{center}
		\begin{tikzpicture}[scale=1]
            \draw[blue!15, fill=blue!15] (-4.5,0)--(4.5,0)--(4.5,2.1)--(-4.5, 2.1)--(-4.5,0);
            \draw[yellow!20, fill=yellow!20] (0,0 )--(4.5,0)--(4.5,1.8)--(0,0);
           \draw[yellow!20, fill=yellow!20] (0,0 )--(-4.5,0)--(-4.5,1.8)--(0,0);
            \draw[green!20, fill=green!20] (0,0 )--(-4.2,2.1)--(4.2,2.1)--(0,0);
		\draw [-> ](-5,0)--(5,0);
		\draw [-> ](0,0)--(0,2.8);
		\node    at (0.1,-0.3)  {$0$};
		\node    at (5.26,0)  { $x$};
		\node    at (0,3.2)  { $t$};
		 \node  [below]  at (4.4,2.7) {\scriptsize $\xi=1$};
		 \node  [below]  at (-4.4,2.7) {\scriptsize $\xi=-1$};
  		\draw [red](0,0)--(-4.5,2);
		\draw [red](0,0)--(4.5,2);
	   \node  []  at (0.7,1.5) {\scriptsize $I$};
     \node  []  at (-0.7,1.5) {\scriptsize $I$};
     \node [] at (0,0.85) {\scriptsize Solitonic region};
		\node  []  at (-3.2,0.9) {\scriptsize $II$};
        \node [] at (-2.8,0.3) {\scriptsize Solitonless region};
	\node  []  at (3.2,0.9) {\scriptsize $II$};
    \node [] at (2.8,0.3) {\scriptsize Solitonless region};
      \node  []  at (4.9,2) {\tiny $III$ };
        \node  []  at (-4.9,2) {\tiny $III$};
              \node  []  at (5.1,1.7) {\tiny   Painlev\'e};
        \node  []  at (-5.1,1.7) {\tiny  Painlev\'e};
		\end{tikzpicture}
	\end{center}
	\caption{\footnotesize The space-time region of $x$ and $t$, where the green region I: $|\xi|<1$  which is solitonic region;
the yellow regions II: $|\xi|>1, $ which is solitonless region; the blue  regions III: $|\xi|\approx 1$ which is   transition  region.  }
	\label{spacetime}
\end{figure}
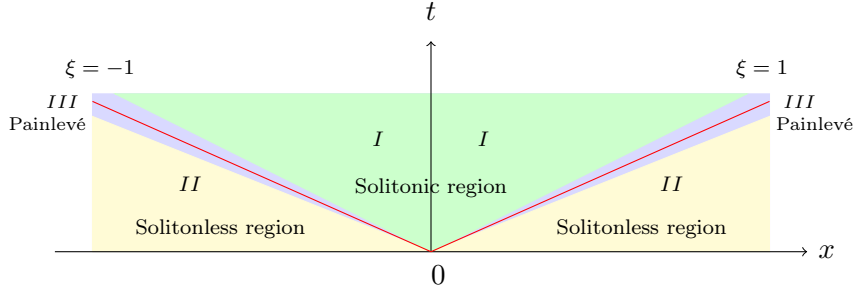

For this purpose,  we  consider  the Cauchy problem  (\ref{q})-(\ref{inq})  with finite density initial data $q_0  \in \tanh x +H^{4,4} (\mathbb{R})$
   that was  studied  by  Cuccagna and Jenkins \cite{CJ}.
They   derived the leading order approximation to  the Cauchy problem  (\ref{q})-(\ref{inq}) in the solitonic
  space-time region I:  $|\xi|<1$  (See  Figure \ref{spacetime})  by using the $\bar{\partial}$-steepest descent method  \cite{CJ}
\begin{align}
& q(x,t)=   T(\infty)^{-2} q^{sol,N}(x,t) + \mathcal{O}(t^{-1 }). \nonumber
\end{align}
  This  result  also  proves  the asymptotic stability of $N$-dark soliton solutions.  A series of  important  works on various  stability of solitons of
  the NLS equation can be found in     \cite{GS0,GS1,GS2}.
In recent years,    the $\bar{\partial}$-steepest descent method \cite{McL1,McL2}  has been successfully used to obtain  the long-time asymptotics and
 the soliton resolution  for some integrable systems  \cite{fNLS,Liu3,LJQ,YF1,YYLmch}.
For  the  solitonless  region $\mathrm{II}$:   $ |\xi| >1 $,
we  further    obtain the large-time asymptotic behavior to the problem   (\ref{q})-(\ref{inq}) in the form \cite{WF}
\begin{equation}
  q(x,t)= e^{-i\alpha(\infty)} \left( 1 +t^{-1/2} h(x,t) \right)+\mathcal{O}\left(t^{-3/4}\right). \nonumber
 \end{equation}

However,  how to   describe    the transition region $\mathrm{III}$:  $  |\xi|\approx 1  $,   proposed by  Cuccagna and Jenkins  (see page 923 in  \cite{CJ}),
seems to remain unknown to the best of our knowledge.    In order to
 solve this   problem,   in this paper,  we apply $\bar\partial$-techniques and double scaling limit method  to find the asymptotics of solution  $q(x,t)$ in two transition regions defined as
\begin{align}
\mathcal{P}_{\pm 1}(x,t):=\{ (x,t) \in \mathbb{R}\times\mathbb{R}^+, \ \ 0<|\xi -(\pm 1)|t^{2/3}\leq C\}. \label{regpm}
\end{align}

\subsection{Main results}
\hspace*{\parindent}

To state  the theorem precisely introduce the   normed spaces:
$L^{p,s}(\mathbb{R})$ defined with
$$\|q\|_{L^{p,s}(\mathbb{R})} := \|\langle x \rangle^s q\|_{L^{p}(\mathbb{R}) },  $$
where  $\langle x \rangle=\sqrt{1+x^2}$ and
$ H^{k,s} (\mathbb{R}) : = H^k(\mathbb{R}) \cap L^{2,s}(\mathbb{R}).$
The following theorem is  our  main results of this paper.

\begin{theorem} \label{th}
     Let $\{ r(z), R(z)\}$ and  $\{z_j\}_{j=0}^{N-1} $  be, respectively, the
 reflection coefficient and the discrete spectrum associated with the weighted Sobolev  initial data $q_0 \in \tanh (x)+H^{4,4}(\mathbb{R})$.
 Then the long-time asymptotics  of the   solution to the Cauchy problem (\ref{q})-(\ref{inq})  for the
 defocusing NLS  equation in two transition regions are given by the following formulas.
 \begin{itemize}
  \item[{\rm   Case I.}]  \label{th1} For $(x,t) \in \mathcal{P}_{-1}(x,t)$,
 \begin{equation}\label{q1}
  q(x,t)= e^{i\alpha(\infty)} \left(1 +\tau^{-1/3} \beta(-1) \right)+\mathcal{O}\left(t^{-1/2}\right),
 \end{equation}
where
\begin{align}
 &\alpha(\infty) = -2i \sum_{j=0}^{N-1} \log \bar{z}_j + \int_{0}^{\infty} \frac{\nu(\zeta)}{\zeta} \, \mathrm{d}\zeta, \ \ \nu(\zeta)= -\frac{1}{2\pi} \log(1-|r(\zeta)|^2), \label{ejs}\\
 &\beta(-1)= -\frac{i}{2}  \left(  u(s ) e^{i\varphi_0} + \int_s^\infty u^2(\zeta)  \mathrm{d}\zeta  \right),   \ \ s  = \frac{8}{3} (\xi +1) \tau^{\frac{2}{3}},\ \tau =\frac{3}{4}t.\nonumber
\end{align}
Here $\varphi_0=    \arg (r(-1)) $ and $u(s ) $ is a solution of the Painlev\'{e} \uppercase\expandafter{\romannumeral2} equation
\begin{equation}\label{pain2equ}
u''(s)  = 2u^3(s) + s u(s)
\end{equation}
with the asymptotics
\begin{align}\label{pain2equa}
 u(s ) \sim - |r(-1)| \rm{Ai}(s) \sim - |r(-1)| \frac{1}{2\sqrt{\pi}}s^{-\frac{1}{4}}  e^{-\frac{2}{3} s^{3/2} }, \ s \to +\infty,
\end{align}
 where $\rm{Ai}(s)$ is the classical Airy function.
  \item[{\rm Case II.}]  \label{th2}
For $(x,t) \in \mathcal{P}_{+1}(x,t)$,
  \begin{equation}\label{q2}
 q(x,t)= e^{i\alpha(\infty)} \left(1 + \tau^{-1/3} \beta(1) \right)+\mathcal{O}\left(t^{-1/2}\right),
\end{equation}
where
\begin{align}
&\alpha(\infty) = \int_{0}^{\infty} \frac{\nu(\zeta)}{\zeta} \, \mathrm{d}\zeta,\ \
 \beta(1)=  \frac{i}{2}\left( u(s )e^{i\varphi_0}+\int_s^\infty u^2(\zeta)  \mathrm{d}\zeta \right),\label{weeie}\\
& \varphi_0 =\arg ( R(1)), \ \ s  =-\frac{8}{3} (\xi -1) \tau^{\frac{2}{3}},  \ \tau =\frac{3}{4}t.\nonumber
\end{align}
Here, the function  $R(z)$  is defined by  (\ref{5oedk}) and $u(s ) $  is a solution of the  Painlev\'e equation  (\ref{pain2equ})  with the asymptotics
\begin{equation}
	u(s) \sim |R(1)|\frac{1}{2\sqrt{\pi}}s^{-\frac{1}{4}}  e^{-\frac{2}{3} s^{3/2} }, \ s \to +\infty.
\end{equation}

\end{itemize}

\end{theorem}
\begin{remark}
Comparing  Case I and Case II in Theorem \ref{th},
we find that  the formulae   (\ref{q1}) and (\ref{q2})  have a similar asymptotic form.
Their  difference   is that  the discrete spectrum contributes to the phase  $\alpha(\infty)$    in (\ref{ejs}),
while the reflection coefficient contributes to the phase $\alpha(\infty)$   in (\ref{weeie}).
\end{remark}

\begin{remark}  \label{rem3}  In the case of the  KdV and Camassa-Holm equation \cite{DVZ,Monvel, Monvel3},   a given RH problem can be reduced into  a  local RH problem
  at  a phase point  $z_0$ with  the  nonlinear steepest method.  For the non-generic case $|r(0)|\not= 1$,  the  local RH problem
can then be analyzed and controlled by the  norm
$ (1- |r(z_0)|^2)^{-1}$.  While  for the generic case $|r(0)|=1$,
 it  turns out that  the  norm  $ (1- |r(z_0)|^2)^{-1}$  blow up as $z_0 \to   0$. This indicates the   presence of a  new   collisionless shock region  for $|r(0)|=1$.

\end{remark}

\begin{remark}  By contrast to the results in  Remark \ref{rem3},
in the case of   the focusing NLS equation, Vartanian in  \cite{VAH1, VAH2,VAH3} obtained the long-time asymptotics for (\ref{q})-(\ref{inq})
 under the hypothesis  $\|r\|_{L^{\infty}(\mathbb{R})}<1$ including $|r(\pm 1)|<1$,  which is the non-generic case.
However, for the generic case $r(\pm 1)=1$,
Cuccagna and Jerkins proposed a new way in \cite{CJ}  to  remove the non-generic condition $\|r\|_{L^{\infty}(\mathbb{R})}<1$
 by specially  handling the singularity caused from   $ |r(\pm 1) | = 1$.
 This method  allows to proceed as in the similar  region, whose asymptotics matches the Painlev\'e asymptotics in self-similar region
 and thus  no new shock wave asymptotic  forms  appear.
 Because of this,  the  Painlev\'e asymptotics given   in Theorem \ref{th} are still effective     for the case    $ |r(\pm 1) | = 1$.

\end{remark}

\subsection{Plan of the proof}
 \hspace*{\parindent}

 We prove Theorem \ref{th} by applying $\bar\partial$-steepest descent approach and double scaling limit
   to the defocusing NLS equation (\ref{q})-(\ref{inq}).
   The organization of our paper is as follows.

     In Section \ref{sec1}, we quickly review some basic  results, especially
 the construction of  a basic   RH   formalism  $M(z)$  related to the Cauchy problem (\ref{q})-(\ref{inq})  for the
	defocusing NLS  equation.  For more details, see  \cite{CJ}.
 Following  a brief review of   Section \ref{sec1},  two main results of Theorem \ref{th}
are presented   in Section  \ref{sec3} and Section  \ref{sec4}, respectively.

  In Section \ref{sec3},  we focus on the long-time asymptotic  analysis for the defocusing NLS  equation in the Painlev\'e sector $\mathcal{P}_{-1}(x,t)$  with the following steps.
  We first   obtain a standard RH problem $ M^{(3)}(z)$   by  removing the poles   and  singularity of  the RH problem $M(z)$   in Subsection \ref{modi1}.
 Then in Subsection \ref{modi2},  after  a continuous extension of
   the  jump matrix with  the $\bar\partial$-steepest descent method,    the   RH problem $ M^{(3)}(z)$ is deformed into a  hybrid $\bar\partial$-RH problem $ M^{(4)}(z)$,
   which can be solved   by decomposing it into a pure RH problem $ M^{rhp}(z)$ and a pure $\bar\partial$-problem $ M^{(5)}(z)$.
The  RH problem $ M^{rhp}(z)$ can be  approximated by a solvable Painlev\'e model  via the local paramatrix near the critical point $z=-1$ in Subsection \ref{modi3}.
The residual error   comes  from a small RH problem $E(z)$ and the pure $\bar\partial$-problem $ M^{(5)}(z)$ in Subsection \ref{modi4}.
      Finally, summing up the estimates above yields  the asymptotic behavior of  the solutions of the defocusing NLS equation  in terms of  the real-valued solutions of the Painlev\'{e} $\mathrm{II}$ equation, and
      the proof of  Theorem \ref{th}---Case I is in Subsection \ref{thm-result1}.

      In Section \ref{sec4}, we investigate the asymptotics of the solution  in the Painlev\'{e} sector $\mathcal{P}_{+1}(x,t)$  using a similar way as Section \ref{sec3},
      and the proof of  Theorem \ref{th}---Case II is  in Subsection \ref{thm-result2}.

To clarify   a series of  deformations and approximations   to  the  Painlev\'e model described  above, we make the following chains as
\begin{align}
&M(z) \xrightarrow [\text{Poles, Contour}]{ \text{Conjugating }  }   M^{(1)}(z)
 \xrightarrow [\text{  Poles}]
{ \text{Removing} }    M^{(2)}(z)   \xrightarrow [\text{Singularity} ]
{ \text{Removing}  }    M^{(3)}(z)   \nonumber\\[6pt]
&\xrightarrow [\text{Contour}]
{\text{Opening} }    M^{(4)}(z)  \longrightarrow
\left\{\begin{matrix}
  \ M^{rhp}(z)
   \longrightarrow \left\{\begin{matrix}M^{loc}(z) \xrightarrow [\text{Painlev\'e } ]
{\text{Scaling} }   M^\infty(k)    \vspace{2mm}\cr
E(z)\qquad\qquad\qquad\qquad \
\end{matrix}\right. \\
  \ M^{(5)}(z):=M^{(4)} (z) ( M^{rhp}(z) )^{-1}.\qquad \qquad
\end{matrix}\right. \nonumber
\end{align}
The main approximation in matching the local model with  the  Painlev\'e RH  model $  M^\infty(k)$ is shown in  Proposition \ref{locpain}.
Two error functions $E(z)$ and $M^{( 5)}(z)$ are given by  Proposition \ref{error1} and  Proposition \ref{m51infty}, respectively.

\section{Inverse   scattering transform} \label{sec1}
\hspace*{\parindent}

In this section, we recall briefly    main results about inverse scattering transform
for the defocusing NLS equation that will be   used in this paper. The details can be found in
\cite{CJ}.

\subsection{Jost functions }
\hspace*{\parindent}

The NLS equation (\ref{q}) admits a Lax pair
  \begin{align}
  	  	&\psi_x=L \psi, \quad L=L(z;x,t)=i \sigma_3 \left(Q-\lambda\left(z\right)\right),\label{laxx}\\
  	  	&\psi_t=T \psi,\quad T=T(z;x,t)=-2\lambda(z) L+i(Q^2-I)\sigma_3+Q_x, \label{laxt}
  \end{align}
	where
	\begin{equation*}
		Q=Q(x,t)=\left(\begin{array}{cc}
			0 & \bar{q}(x,t)\\
			q(x,t) & 0
		\end{array}\right),\quad \lambda(z) = \frac{1}{2}\left(z+z^{-1}\right).
	\end{equation*}

We define the Jost solutions of Lax pairs (\ref{laxx})-(\ref{laxt}) with  asymptotics
	\begin{equation*}
		\psi^\pm(z) \sim Y_\pm e^{-it\theta(z)\sigma_3},  \quad x \to \pm \infty,
	\end{equation*}
	where we denote $\psi^\pm(z):=\psi^\pm(z;x,t)$  and
    \begin{align}
       & Y_\pm=I\pm \sigma_1z^{-1},\quad \text{det}Y_\pm=1-z^{-2},\\
       & \theta(z)=\zeta(z) \big(x t^{-1}- 2 \lambda(z) \big),\quad\zeta(z)=\frac{1}{2}(z-z^{-1}).\label{theta0}
     \end{align}

 Making  a transformation
	\begin{equation}
		m^\pm(z )=\psi^\pm(z )e^{it\theta(z)\sigma_3},
	\end{equation}
then $m^\pm(z )$ satisfy  the Volterra integral equations
\begin{align}
&m^\pm(z)=Y_\pm + \int_{\pm \infty}^{x} \left( Y_\pm e^{-i \zeta(z) (x-y)\hat\sigma_3}  Y^{-1}_\pm    \right) \left( \Delta L_\pm m^\pm(z;y)      \right)  \mathrm{d}y, \ z\not= \pm 1,\nonumber\\
&m^\pm(z)=Y_\pm + \int_{\pm \infty}^{x} \left( I + (x-y)L_\pm      \right) \Delta L_\pm (\pm 1;y) m^\pm(\pm 1;y)  \mathrm{d}y, \ z = \pm 1,\nonumber
\end{align}
where $\Delta L_\pm(z;y) = i \sigma_3 (Q \mp \sigma_1)$.

The existence, analyticity  and differentiability  of $m^\pm(z )$ can be proven directly. Here we list their  properties \cite{CJ}.

   \begin{lemma} \label{lemma3.1}
    	Denote $m^{\pm}(z) = \left( m^{\pm}_1(z), m^{\pm}_2(z) \right)$ and let $q_0  \in \tanh (x)+ L^{1,2} (\mathbb{R})$ and $q'_0 \in W^{1,1}(\mathbb{R})$,
    	then we have

\begin{itemize}

 \item   { Analyticity:} $m_1^+(z )$ and $m_2^-(z )$  can be analytically extended to $z \in \mathbb{C}^-$, while
  $m_1^-(z )$ and $m_2^+(z )$  can be analytically extended to $z \in \mathbb{C}^+$.

    \item {Symmetry:} $m^\pm_1(z )$ and $m^\pm_2(z )$ admit the symmetries
         	\begin{equation}
    	\psi_\pm(z )=\sigma_1 \overline{\psi_\pm(\bar{z} )}\sigma_1=\pm z^{-1} \psi^\pm(z^{-1} )\sigma_1.
    	\end{equation}

  \item    { Asymptotics:} $m^\pm_1(z )$ and $m^\pm_2(z )$ have the asymptotic properties
    	\begin{align*}
    	m^\pm_1(z ) &= e_1 +  \mathcal{O} \left(   z^{-1} \right); \    \  m^\pm_2(z ) = e_2 +   \mathcal{O} \left(   z^{-2} \right),  z \to \infty,\nonumber\\
    	m^\pm_1(z ) &= \pm \frac{1}{z}e_2 +\mathcal{O}(1);  \ \  m^\pm_2(z )  = \pm \frac{1}{z}e_1 +\mathcal{O}(1),  \  \; z\to 0.
    	\end{align*}

    \item For any $\delta >0$ sufficiently small, the maps
    $$ q \to \det[\psi_1^-, \psi_2^+], \  \ q \to \det[\psi_1^+, \psi_1^-],$$
    are locally Lipschiz maps
  $$ \{ q: q   \in \tanh (x)+ L^{1,2}, q' \in W^{1,\infty } \} \rightarrow  W^{1,\infty}( \mathbb{R}  \setminus  (-\delta, \delta)).$$
  \end{itemize}

    \end{lemma}

\subsection{Scattering data }
\hspace*{\parindent}

The Jost functions  $\psi_\pm(z)$ admit the scattering relation
 \begin{equation}
    \psi_-(z )=\psi_+(z )S(z),
    \end{equation}
where $S(z)$ is  the   spectral  matrix given by
    $$S(z)= \left(\begin{array}{cc}
    s_{11}(z)& 	\overline{s_{21}(\bar z)}\\
    s_{21}(z)& 	\overline{s_{22}(\bar z)}
    \end{array} \right),$$
   and $s_{11}(z) $ and $s_{21}(z) $ are the scattering data, by which   we define a reflection coefficient
	    \begin{equation}
	    	r(z) := \frac{s_{21}(z)}{s_{11}(z)}. \label{reflec}
	    \end{equation}
	  It is shown that the scattering data  and the reflection coefficient have  the following  properties \cite{CJ}.

 \begin{lemma}  \label{lemm3.2}
 Let   $q_0\in \tanh x +H^{2,2} (\mathbb{R})$, then
\begin{itemize}

\item     $s_{11}(z)$ can be analytically extended to $z \in \mathbb{C}^+$ while $s_{21}(z)$ and $r(z)$ are defined only for $z\in \mathbb{R}\setminus \{0, \pm1\}$.
 Zeros of $s_{11}(z)$ in  $\mathbb{C}^+ $ are simple, finite and distribute on the unitary circle.

   \item The scattering data can be described by the Jost functions
       	    \begin{equation}
       	    	s_{11}(z)=\frac{\det  \left[\psi_1^-(z ), \psi_2^+(z )\right]  }{1-z^{-2}},\quad    	s_{21}(z)=\frac{\det  \left[\psi_1^+(z ), \psi_1^-(z )\right]  }{1-z^{-2}}.\label{scatter}
       	    \end{equation}

\item For $z\in \mathbb{R}\setminus \{0, \pm1\}$, we have
\begin{align}
& |s_{11}(z) |^2=1+|s_{21}(z) |^2 \Longleftrightarrow  |r(z)|^2=1-|s_{11}(z)|^{-2}<1. \label{scatter2}
\end{align}

\item  $r(z)\in H^{1 }(\mathbb{R})$ and
   $ \parallel\log(1-|r|^2) \parallel_{L^p(\mathbb{R})}<\infty, \ \ p\geq 1. $

 \item {Symmetry:} $ {S ( {z})}  = \sigma_1\overline{S (\bar z) }\sigma_1 =  \pm z^{-1}  {S ( z^{-1} )}\sigma_1,$

 \item {Asymptotics:}
\begin{align*}
&|s_{21}(z)|=\mathcal{O}(|z|^{-2}),\ \ |z|\rightarrow \infty, \ \ |s_{21}(z)|=\mathcal{O}(|z|^2),\ \ |z|\rightarrow 0.\\\
&r(z)\sim z^{-2},\ \ z\rightarrow \infty, \ \ r(z)\sim 0,\ \ z\rightarrow 0.
\end{align*}

\item In the generic case, $s_{11}(z)$ and $s_{21} (z)$ have the same order singularities at $z=\pm 1$,
\begin{align*}
& s_{11}(z) =\frac{s_\pm}{z\mp 1} +\mathcal{O}(1), \ s_{21}(z) = \mp \frac{s_\pm}{z\mp 1} +\mathcal{O}(1),
\end{align*}
where $\ s_\pm=\det  \left[\psi_1^-(\pm 1 ), \psi_2^+(\pm1 )\right]$ and
\begin{align}
& \lim_{z\to \pm 1} r(z)=\mp 1. \label{weie}
\end{align}

\item In the non-generic case,   $s_{11}(z)$ and $s_{21} (z)$  are continuous  at $z=\pm 1$ and $|r(\pm 1)|<1 $.

\end{itemize}

\end{lemma}

\subsection{A basic Riemann-Hilbert problem}
\hspace*{\parindent}

Denote $H=\{ 0, 1,\cdots, N-1\}$ and
\begin{align*}
&\mathcal{Z}^+=\{ z_j| s_{11}(z_j)=0, \ z_j \in \mathbb{C}^+, \ |z_j|=1,  \  j\in H\}, \\
& \mathcal{Z}^- =\{ \bar z_j |  s_{22}(\bar z_j) =0, \ \bar z_j \in \mathbb{C}^-,\ |\bar z_j|=1, \  j\in H\}.
\end{align*}
Moreover, $s_{11}(z)$  satisfies  the trace formula
	  \begin{align}
	  	s_{11}(z)= \prod_{j=1}^{N} \frac{z-z_j}{z-\bar{z}_j} \exp\left(-i \int_{\mathbb{R}} \frac{\nu(\zeta)}{\zeta-z} \, \mathrm{d}\zeta\right),\label{trance}
	  \end{align}
	where $z_j \in \mathcal{Z}^+ $ and
	  \begin{align}
	  	\nu(\zeta)=-\frac{1}{2 \pi} \log \left( 1-|r(\zeta)|^2\right). \label{nu}
	  \end{align}

Based on the properties of the scattering data $s_{11}(z)$, we define
\begin{align}
&M(z)=M(z;x, t)=\left\{\begin{matrix} (m_1^{-}(z )/{s_{11}(z) },  m_2^{+}(z )), \  \ &z\in \mathbb{C}_+,\vspace{2mm} \cr
(m_1^{+}(z ),  m_2^{-}(z )/{\overline{s_{11}(\bar z)}} ), \  \ &z\in \mathbb{C}_-. \end{matrix}\right. \nonumber
 \end{align}
 It can be verified that $M(z)$ satisfies the following RH problem.

\begin{prob} \label{RHP0}
  Find  a matrix-valued function $M(z)$ which satisfies
\begin{itemize}
\item  Analyticity: $M(z)$ is meromorphic in $\mathbb{C}\backslash \mathbb{R}$.

\item  Symmetry: $M(z)=\sigma_1 \overline{M(\bar{z})} \sigma_1=z^{-1} M(z^{-1}) \sigma_1$.

\item Singularity:  $M (z)$ has the singularity point $z=0$ with $zM(z)\longrightarrow \sigma_1, \quad  z\rightarrow 0$.

\item  Asymptotic behavior: $M(z)\longrightarrow I, \quad  z\rightarrow \infty.  \ \    \qquad\qquad \label{fRHP3}$

\item  Jump condition: $M(z)$ satisfies the jump condition
	$$M_+(z)=M_-(z)V(z), \; z \in \mathbb{R},$$
      	where
      	\begin{equation}\label{V0}
      		V(z)=\left(\begin{array}{cc}
      			1-|r(z)|^2 & -\overline{r(z)}e^{-2it \theta(z)}\\
      		r(z)	e^{2it \theta(z)} & 1
      		\end{array}\right),
      	\end{equation}
      with
      $$ \theta(z) =\zeta(z)\frac{x}{t}-2\zeta(z)\lambda(z) = \frac{x}{2t}
      (z-z^{-1})-\frac{1}{2}  (z^2-z^{-2}).$$

\item  Residue conditions: $M(z)$  has simple poles at each points $z_j$ in $\mathcal{Z}^+ \cup \mathcal{Z}^-$ with the following residue conditions
\begin{align}
& \mathop{\mathrm{Res}}\limits_{z=z_j} M (z) =\lim_{z\rightarrow z_j}M(z)\left(\begin{array}{cc} 0&0\\ c_je^{2i t\theta(z_j)} &0\end{array}  \right), \label{fresm1}\\
&\mathop{\mathrm{Res}}\limits_{z=\bar z_j}M(z) =\lim_{z\rightarrow \bar z_j}M(z)\left(\begin{array}{cc} 0& -\bar c_je^{-2i t\theta(\bar z_j)} \\0&0\end{array}\right),\label{fresm2}
\end{align}
where
\begin{align}
& c_j=\frac{s_{21}(z_j)}{s_{11}'(z_j)}=\frac{4iz_j}{\int_\mathbb{R}|\psi_2^+(z_j,x)|^2dx}=iz_j|c_k|. \label{conect}
 \end{align}
\end{itemize}
This is an RH problem with jumps on the real axis and poles distributed on the unit circle. See Figure \ref{fmjump2}. The solution of the NLS equation can be given by the reconstruction formula
\begin{align}
 q(x, t)=\lim_{z\rightarrow \infty} (zM (z ))_{21}. \label{sol}
 \end{align}

\end{prob}

\subsection{Classification of asymptotic regions}
\hspace*{\parindent}

\begin{figure}
\begin{center}
\begin{tikzpicture}[scale=0.8]
 \draw [dotted] (0, 0) circle [radius=2];
\draw[   -latex ](-4.5, 0)--(4.5, 0);
\node    at (5.4, 0)  {\footnotesize $Re z$};
 \node    at (-0.9,  1.4)  {\footnotesize $z_j$};
 \node    at (-0.9, -1.4)  {\footnotesize $\bar z_j$};
\node    at (0, -0.3)  {\footnotesize $0$};
\node    at (2.3, -0.3)  {\footnotesize $1$};
\node    at (-2.4, -0.3)  {\footnotesize $-1$};
       \coordinate (A) at (1.73,  1);
		\coordinate (B) at (1.73,  -1);
		\coordinate (C) at (1,  1.73);
		\coordinate (D) at (1,  -1.73);
		\coordinate (E) at (-1.73,  1);
		\coordinate (F) at (-1.73,  -1);
		\coordinate (G) at (-1,  1.73);
		\coordinate (H) at (-1,  -1.73);
		\coordinate (I) at (0,  2);
		\coordinate (J) at (0,  -2);
		\coordinate (K) at (0,  0);
		\coordinate (L) at (2,  0);
		\coordinate (M) at (-2,  0);
\fill[red] (A) circle (1.5pt);
\fill[red] (B) circle (1.5pt);
\fill[red] (C) circle (1.5pt);
\fill[red] (D) circle (1.5pt);
\fill[red] (E) circle (1.5pt);
\fill[red] (F) circle (1.5pt);
\fill[red] (G) circle (1.5pt);
\fill[red] (H) circle (1.5pt);
\fill[red] (I) circle (1.5pt);
\fill[red] (J) circle (1.5pt);
\fill[blue] (K)  circle (1.5pt);
\fill[blue] (L) circle (1.5pt);
\fill[blue] (M) circle (1.5pt);
\end{tikzpicture}
\end{center}
\caption{\footnotesize The poles $z_j\in \mathcal{Z}^+, \bar z_j\in \mathcal{Z}^-$  and jump contour $\mathbb{R}$ for the  RH problem  $M(z)$.}
\label{fmjump2}
\end{figure}
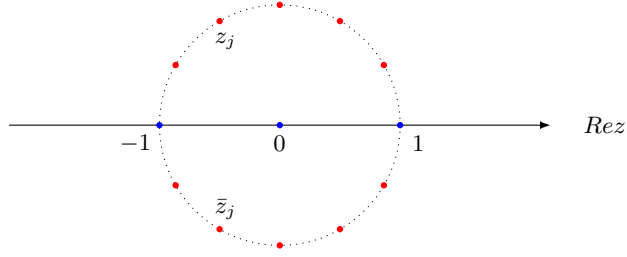

 The jump matrix $ 	V(z)$ admits  the following    factorizations
	   \begin{equation}\label{v}
	   	V(z)=        \begin{cases}
	   		\left(\begin{array}{cc} 1&-\bar{r}e^{-2it\theta(z) }\\ 0&1\end{array}  \right)
	   		\left(\begin{array}{cc} 1&0\\   re^{2it\theta (z)}&1\end{array}  \right),\\
	   		\left(\begin{array}{cc} 1 & 0 \\ \frac{r}{1-|r|^2}e^{2it\theta(z) } & 1\end{array}  \right)
	   		\left(\begin{array}{cc} 1-|r|^2&0\\ 0&\frac{1}{1-|r|^2} \end{array}  \right) \left(\begin{array}{cc} 1 &
	   			\frac{-\bar{r}}{1-|r|^2}e^{-2it\theta(z) } \\ 0 & 1\end{array}  \right).
	   	\end{cases}
	   \end{equation}
The long-time asymptotics of RH problem \ref{RHP0} is affected by the growth and decay of the oscillatory terms  $e^{\pm 2it\theta(z)}$ in the jump matrix $V(z)$.
Direct calculations show that
        \begin{align}
	     \re\left(2i\theta(z)\right) = 2\re z \im z \left( 1 + \frac{1}{\left(\re^2 z+ \im^2 z\right)^2} \right) -2\xi \im z \left( 1+\frac{1}{\re^2 z+ \im^2 z} \right),\label{theta01}
        \end{align}
where $\xi:=\frac{x}{2t}$.   The signature table   of $\re( 2 i\theta(z) )$ and distribution of phase points are   shown   in Figure \ref{proptheta}.
The signal of $\re (2i\theta(z))$ determines  the decay  regions  of the oscillating factor $e^{\pm 2it\theta(z)}$,
which  inspires  us to open the jump contour $\mathbb{R}$ with  different  factorizations  of  the jump matrix $V(z)$.

The stationary phase points are determined  by the equation
        \begin{align}
            & \theta'(z) =- z^{-1} ( \eta^2-\xi \eta-2) =0,\label{2theta}
        \end{align}
        where $ \eta=z+z^{-1}$.  The equation (\ref{2theta})  admits two solutions
  \begin{align}
& \eta_1=\frac{1}{2} (\xi -\sqrt{\xi^2+8}), \ \ \xi<-1, \label{2323}\\
& \eta_2=\frac{1}{2} (\xi +\sqrt{\xi^2+8}), \ \ \xi>1.\label{23232}
\end{align}

            \begin{figure}[htbp]
    	\centering
    	    	\subfigure[$ \xi<-1$]{\label{figurea}
    		\begin{minipage}[t]{0.32\linewidth}
    			\centering
    			\includegraphics[width=1.5in]{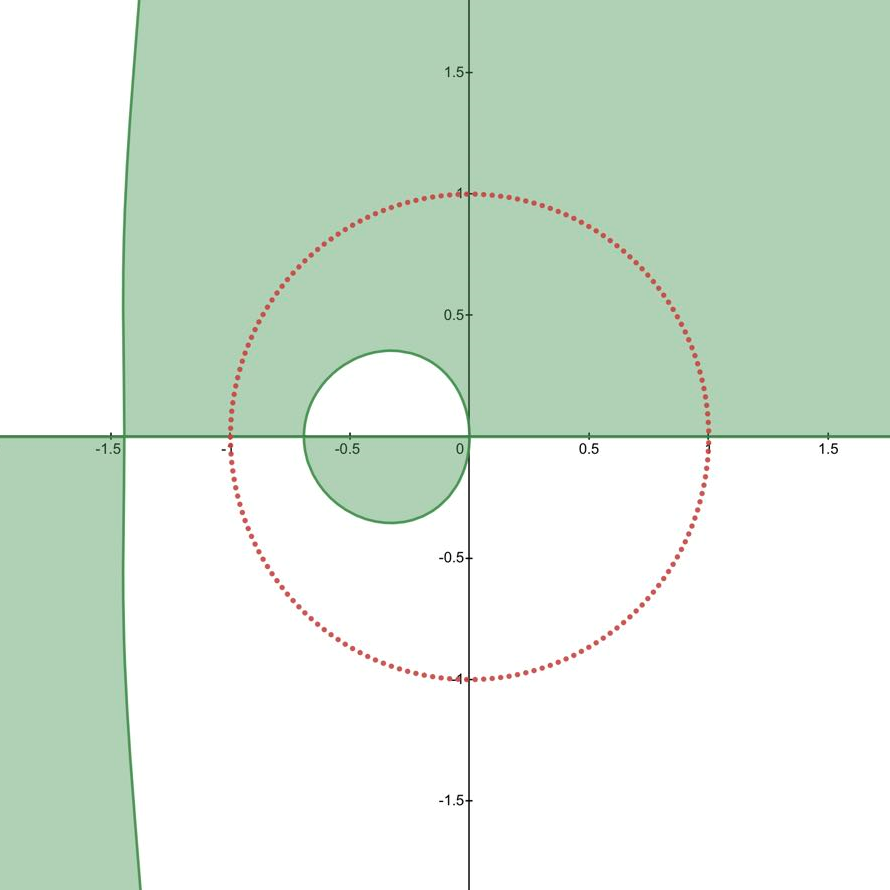}
    		\end{minipage}
    	}%
    	\subfigure[$\xi=-1$]{\label{figureb}
    		\begin{minipage}[t]{0.32\linewidth}
    			\centering
    			\includegraphics[width=1.5in]{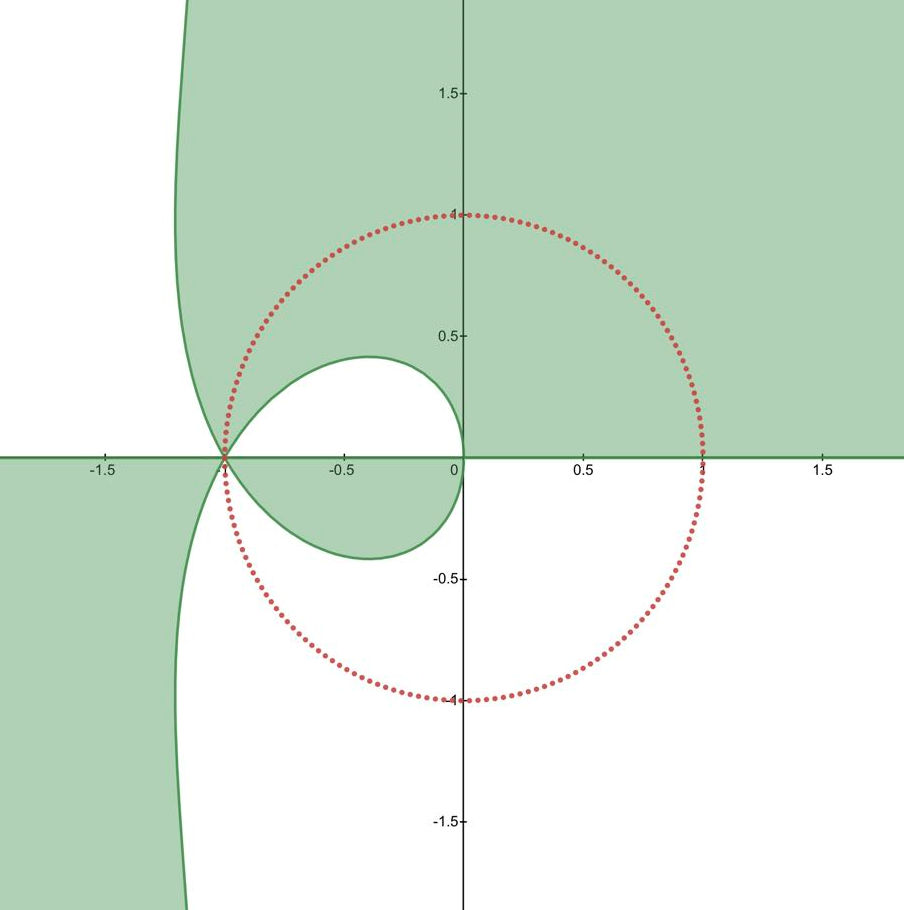}
    		\end{minipage}%
    	}%
    	\subfigure[$-1<\xi<0$]{\label{figurec}
	\begin{minipage}[t]{0.32\linewidth}
		\centering
		\includegraphics[width=1.5in]{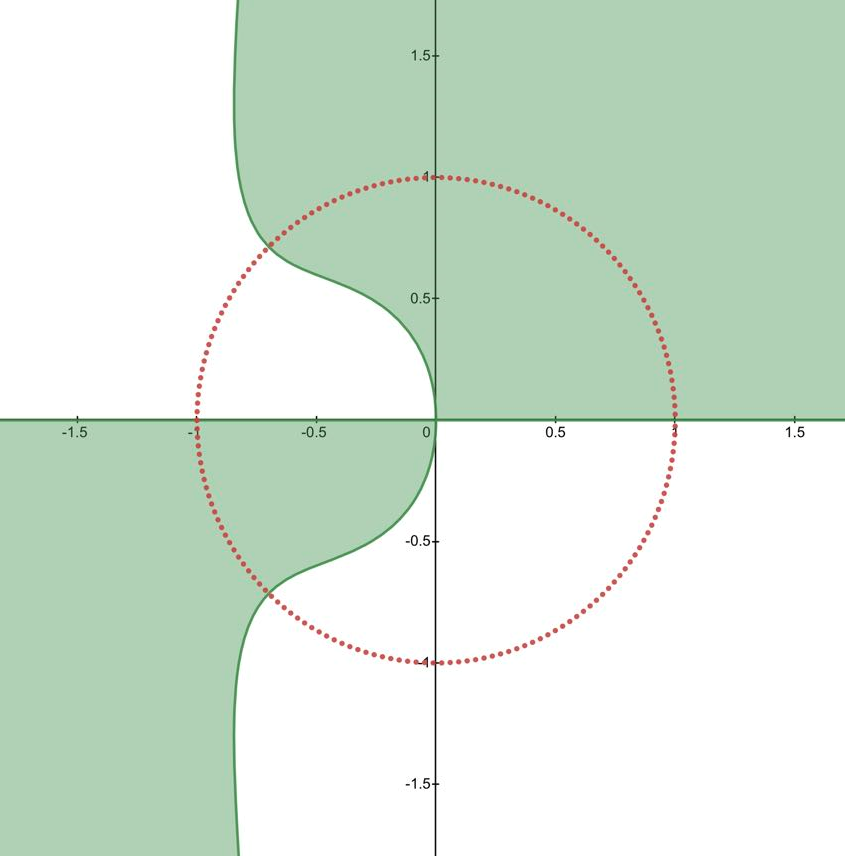}
	\end{minipage}%
}%

    \subfigure[$0<\xi<1$]{\label{figured}
    	\begin{minipage}[t]{0.32\linewidth}
    		\centering
    		\includegraphics[width=1.5in]{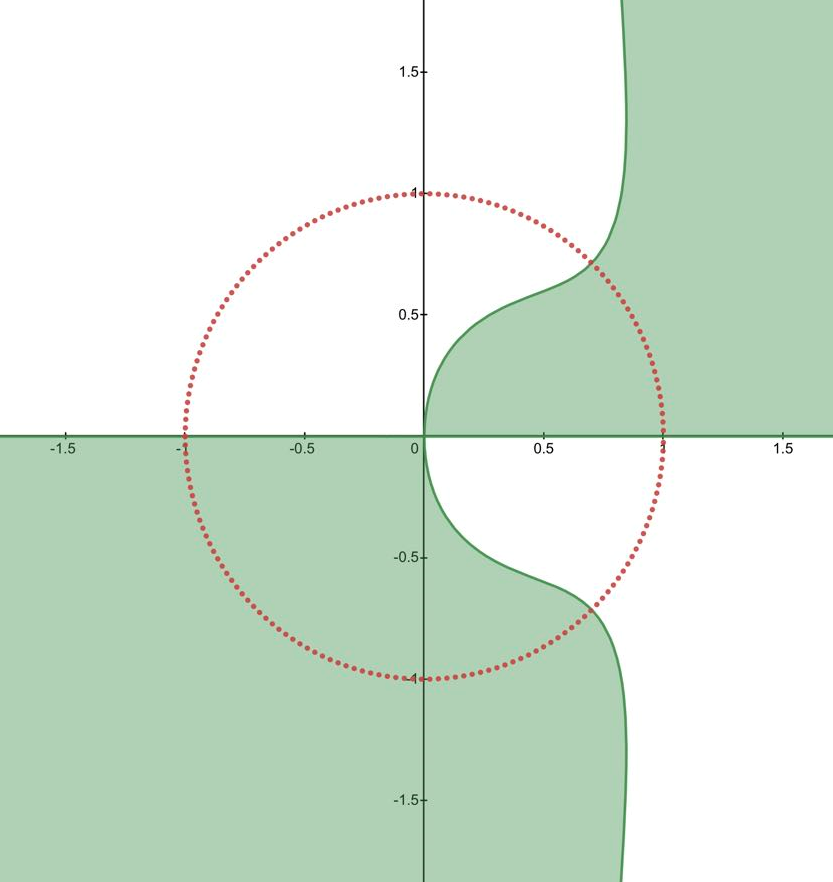}
    	\end{minipage}
    }%
    \subfigure[$\xi= 1$]{\label{figuree}
    	\begin{minipage}[t]{0.32\linewidth}
    		\centering
    		\includegraphics[width=1.5in]{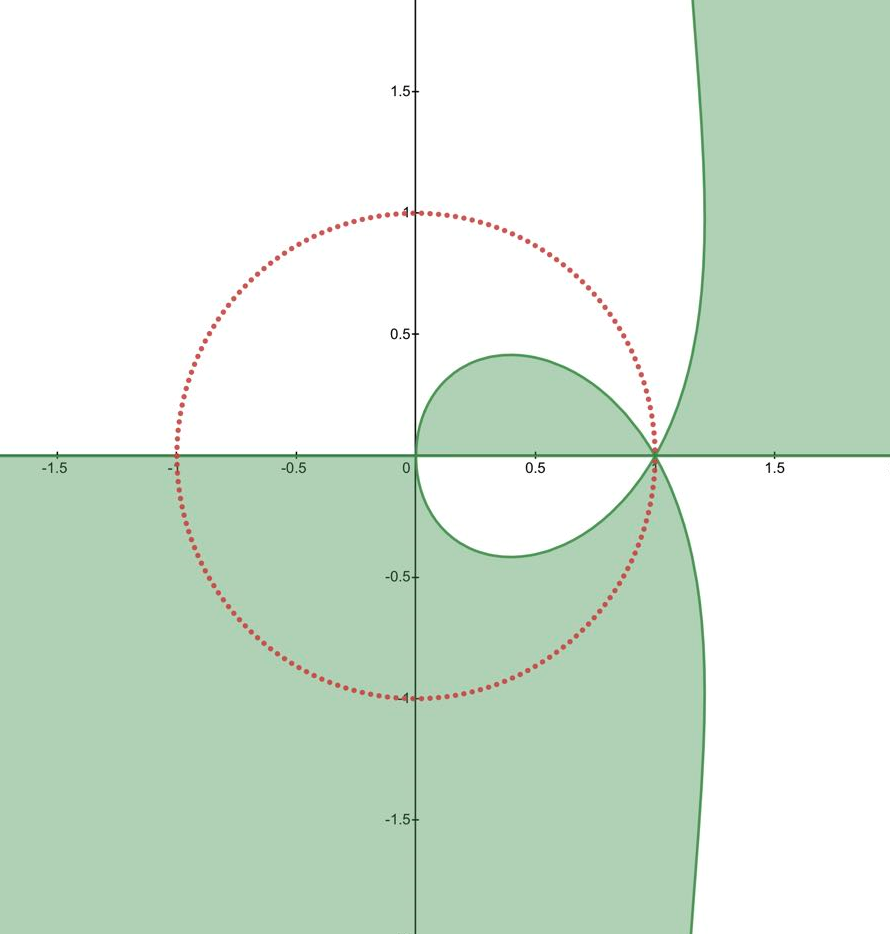}
    	\end{minipage}
    }%
    \subfigure[$\xi>1$]{\label{figuref}
	\begin{minipage}[t]{0.32\linewidth}
		\centering
		\includegraphics[width=1.5in]{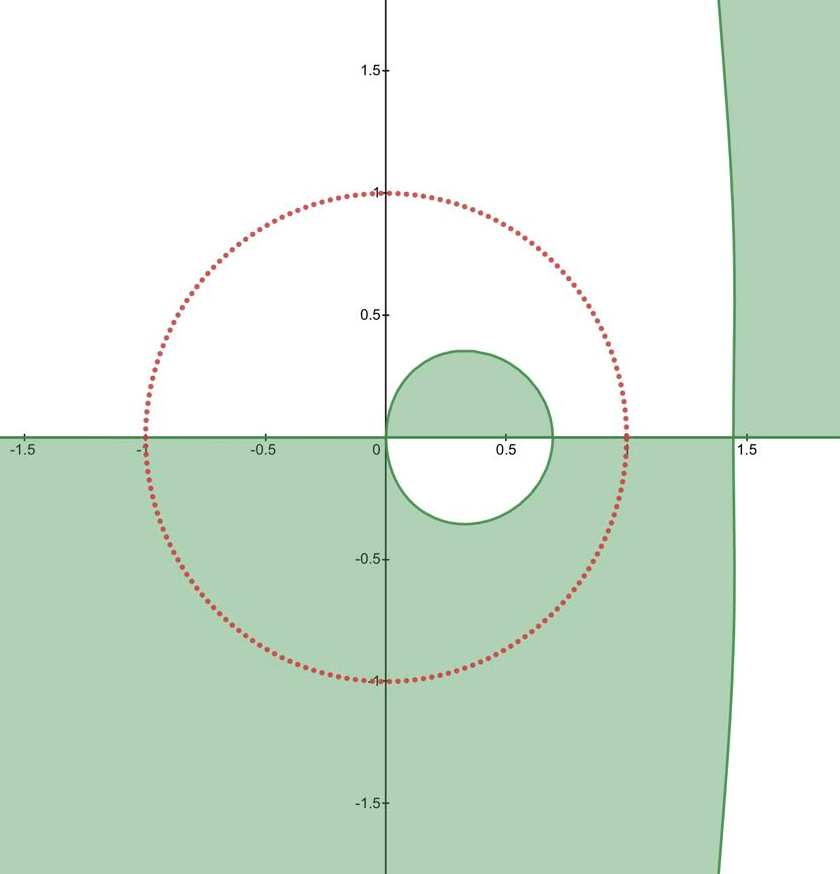}
	\end{minipage}
}%
    	\centering
    	\caption{\footnotesize The signature table   of $\re (2i\theta(z))$ and distribution of phase points.
    In the green region, we have $\re (2i\theta(z))>0$, which implies that $e^{-2it\theta(z)} \to 0$ as $t\to +\infty$; In the white region,  $\re (2i\theta(z))<0$,
    which implies that $e^{2it\theta(z)} \to 0$ as $t\to +\infty$. }
    	\label{proptheta}
    \end{figure}

For the case $\xi<-1$,  the two stationary phase points  satisfy the equation
  \begin{align}
&  z^2+1=\eta_1 z, \label{eq219}
\end{align}
which has two solutions
         \begin{align}
         &\xi_j = -\frac{1}{2}\left|  \eta_1 +( -1)^{j+1} \sqrt{ \eta_1^2  -4} \right|,\quad j=1, 2, \ \ {\rm for} \ \  \xi<-1,\label{xi2}
          \end{align}
with $\xi_2 <-1<\xi_1 <0$.

For the case $\xi>1$,  the two stationary phase points  satisfy the equation
  \begin{align}
&  z^2+1=\eta_2 z,
\end{align}
which has two solutions
         \begin{align}
         &\xi_j = \frac{1}{2}\left|  \eta_2 +( -1)^j \sqrt{ \eta_2^2  -4} \right|,\quad j=1, 2, \ \ {\rm for} \ \  \xi>1,\label{xi1}
          \end{align}
 with $0<\xi_1 <1<\xi_2 $.

 The number of phase points located on jump contour $\mathbb{R}$  allows  us to  divide the   half-plane  $(x,t)$   into
three  asymptotic regions.
\begin{itemize}
\item For  $|\xi|<1$,  there is no phase point on  $\mathbb{R}$. See Figure \ref{figurec}  and \ref{figured}.
This case is a   {solitonic  region}  studied by Cuccagana  and Jenkins \cite{CJ}.

\item    For  $ |\xi|>1$,  there are two phase points on  $\mathbb{R}$. See Figure  \ref{figurea}  and \ref{figuref}.
This case is a   solitonless  region  studied by us \cite{WF}.

\item   For   $|\xi|\approx 1$, there is one phase point on  $\mathbb{R}$. See Figure \ref{figureb}  and \ref{figuree}.
This case is a  transition region,  which is an  open question proposed by Cuccagana and Jenkins \cite{CJ},
and we will  solve it  in our present paper.

\end{itemize}

To describe the asymptotics  in  Figure  \ref{figureb}  and \ref{figuree},
 we aim to find the asymptotics of $q(x,t)$ in the transition region $\mathcal{P}_{\pm1}(x,t)$  defined
by (\ref{regpm})  in the next   Section \ref{sec3} and Section \ref{sec4}, respectively.

\section{Painlev\'e  asymptotics in transition   region   $\mathcal{P}_{-1}(x,t)$} \label{sec3}
\hspace*{\parindent}

In this section, we study  the Painlev\'e asymptotics    in the region  $(x,t)\in \mathcal{P}_{-1}(x,t) $.
Here we consider the region   $ -C< (\xi+1) t^{2/3}<0$ which corresponds to Figure \ref{figurea}. For brevity, we denote
$$   \mathcal{P}_{< -1}(x,t) =\{ (x,t): -C< (\xi+1) t^{2/3}<0\}.  $$
  In this case,
the two stationary points $\xi_1, \xi_2$ defined by (\ref{xi2}) are real and close to $z=-1$ at least the speed of $t^{-1/3}$ as $t\to +\infty$.

We make some modifications to the basic  RH problem \ref{RHP0}  to get a standard RH problem
  without poles and  singularities  by  performing   two essential operations.

\subsection{Modifications to the basic RH problem}\label{modi1}
\hspace*{\parindent}

Since the  poles $z_j \in \mathcal{Z}^+$ and $\bar z_j\in \mathcal{Z}^-$ are finite,
 distributed on the unitary circle and  far away from the jump contour $\mathbb{R}$ and critical line ${\rm Im}\theta(z)=0$,
they exponentially decay when we change their residues  into jumps on small circles.  This allows us to
first modify the basic  RH problem \ref{RHP0} by removing these poles.

\subsubsection{Removing poles }
\hspace*{\parindent}

    To remove poles $z_j \in \mathcal{Z}^+$ and $\bar z_j\in \mathcal{Z}^-$ and  open   the contour $(0,\infty)$  by  the  second  matrix   decomposition in  (\ref{v}),
    we  define  the  function
      \begin{equation}
        T(z)=
        \prod_{j=0}^{N-1} \Bigg(\frac{z-  z_j}{zz_j-1}\Bigg) \exp \left(  - i \int_0^{\infty}   \nu(\zeta) \left(\frac{1}{ \zeta-z}- \frac{1}{2\zeta} \right) \, \mathrm{d}\zeta \right),\label{tfcs}
        \end{equation}
where  $\nu(\zeta)$ is given by (\ref{nu}).  Then
  the following lemma holds.
\begin{lemma} \cite{CJ}
            The function $T(z)$ has the following properties
            \begin{itemize} \label{prop1}
                \item Analyticity: $T(z)$ is  meromorphic  in $\mathbb{C} \backslash   [0,\infty)   $ with  simple zeros at the points $z_j$ and simple poles at the points $\bar{z}_j$.
                \item Symmetry: $\overline{T(\bar{z})}=T(z)^{-1}=T(z^{-1})$.
                \item Jump condition:\begin{equation*}
                    T_+(z)=T_-(z)(1-|r(z)|^2), \quad z\in (0,\infty).
                \end{equation*}
                \item Asymptotic behavior: Let
                \begin{equation}
                    T(\infty) := \lim_{z \to \infty} T(z)= \Bigg(\prod_{j=0}^{N-1}\bar z_j\Bigg) \exp \left(i \int_{0}^{\infty} \frac{\nu(\zeta)}{2\zeta} \, \mathrm{d}\zeta \right).\label{Texpan1}
                \end{equation}
                  Then the asymptotic expansion at infinity is
                \begin{equation}
                   T(z)=T(\infty)\Bigg[ 1-\frac{i}{z}\Bigg(2\operatorname*{\sum}\limits_{j=0 }^{N-1} {\rm Im} z_{j}-\int_0^{\infty}\nu(\zeta)d\zeta\Bigg)+\mathcal{O}(z^{-2})\Bigg].\label{Texpan}
                \end{equation}
            \item Boundedness: The ratio $\frac{s_{11}(z)}{T(z)}$ is holomorphic in $\mathbb{C}^+$ and $\left|\frac{s_{11}(z)}{T(z)} \right|$ is bounded for $z\in \mathbb{C}^+$. Additionally,  $\frac{s_{11}(z)}{T(z)}$  extends as a continuous function on $\mathbb{R}^+$ with $\left|\frac{s_{11}(z)}{T(z)} \right|=1$ for $z
            \in (0,\infty)$.

            \end{itemize}
        \end{lemma}

For $z_j\in \mathcal{Z}^+$ on the circle  $|z|=1$, define
\begin{equation}\label{definerho}
 \rho <  \frac{1}{2} {\rm min}  \{   \operatorname*{min}\limits_{ z_j,  z_l\in  \mathcal{Z}^+ }  |z_j-z_l|,    \operatorname*{min}\limits_{z_j\in  \mathcal{Z}^+} |{\rm Im} z_j|,  \operatorname*{min}\limits_{z_j\in \mathcal{Z}^+,  { \rm Im} \theta(z) =0 }| z_j-z|  \},
 \end{equation}
and  make small circles  at  the center  $z_j $ and  $\bar z_j $   with the radius $\rho$, respectively.
The direction of each small circle in $\mathbb{C}^+$  is  counterclockwise, and  that of each small circle in $\mathbb{C}^-$ is clockwise. See Figure \ref{Djump65}.

  Notice that the exponential factors  in the residue conditions  (\ref{fresm1}) and (\ref{fresm2})
  increase with $t$.  To arrive at an RH problem with decreasing off-diagonal terms in the jump matrices,
we further  construct the interpolation function
            \begin{equation}
           G(z) = \begin{cases}
             \left(\begin{array}{cc} 1& - \displaystyle {\frac{z-z_j}{c_j e^{2it\theta(z_j)}} }\\ 0&1\end{array}  \right), \;   |z-z_j|<\rho,   \\[4pt]
           \left(\begin{array}{cc} 1&0 \\ -\displaystyle {\frac{z-\bar{z}_j}{\bar{c}_j e^{ -2it\theta(\bar{z}_j) }}}&1\end{array}  \right), \;   |z-\bar{z}_j|<\rho,\\
            \left(\begin{array}{cc} 1&0 \\ 0&1\end{array}  \right), \;  \;  \; \text{otherwise},
           \end{cases}
       \end{equation}
  where $z_j \in \mathcal{Z}^+$ and  $\bar{z}_j \in \mathcal{Z}^-$.
 Define a directed path
\begin{equation*}
	\Sigma^{(1)}=\mathbb{R} \cup\left[\bigcup_{j=0}^{N-1}  \{ z\in \mathbb{C}: |z-z_j|=\rho,  or \   |z-\bar z_j|=\rho\} \right],
\end{equation*}
where the direction on $\mathbb{R}$ goes from left to right, depicted in Figure \ref{Djump65}.

\begin{figure}
\begin{center}
\begin{tikzpicture}[scale=0.8]
 \draw [dotted ] (0, 0) circle [radius=2];
\draw[   -latex ](-5, 0)--(5, 0);
\node    at (5.5, 0)  {\footnotesize $Re z$};
  \node    at (-0.9,  1)  {\footnotesize $z_j$ };
 \node    at (-0.9, -1)  {\footnotesize $\bar z_j $};
\node    at (0, -0.3)  {\footnotesize $0$};
        \coordinate (A) at (1.73,  1);
		\coordinate (B) at (1.73,  -1);
		\coordinate (C) at (1,  1.73);
		\coordinate (D) at (1,  -1.73);
		\coordinate (E) at (-1.73,  1);
		\coordinate (F) at (-1.73,  -1);
		\coordinate (G) at (-1,  1.73);
		\coordinate (H) at (-1,  -1.73);
		\coordinate (I) at (0,  2);
		\coordinate (J) at (0,  -2);
		\coordinate (K) at (0,  0);
		\coordinate (L) at (2,  0);
		\coordinate (M) at (-2,  0);
\fill[red] (A) circle (1.5pt);
\fill[red] (B) circle (1.5pt);
\fill[red] (C) circle (1.5pt);
\fill[red] (D) circle (1.5pt);
\fill[red] (E) circle (1.5pt);
\fill[red] (F) circle (1.5pt);
\fill[red] (G) circle (1.5pt);
\fill[red] (H) circle (1.5pt);
\fill[red] (I) circle (1.5pt);
\fill[red] (J) circle (1.5pt);
\fill[blue] (K)  circle (1.5pt);
\fill[blue] (L) circle (1.5pt);
\fill[blue] (M) circle (1.5pt);
 \draw [] (A) circle [radius=0.3];
  \draw [  -> ]  (2.03,  1) to  [out=90,  in=0] (1.73,  1.3);	
  \draw [] (B) circle [radius=0.3];
  \draw [   -> ]  (1.43,  -1) to  [out=270,  in=180] (1.73,  -1.3);	
     \draw [] (C) circle [radius=0.3];
   \draw [   -> ]  (1.3,  1.73) to  [out=90,  in=0] (1,  2.03);	
        \draw [] (D) circle [radius=0.3];
   \draw [  -> ]  (0.7,  -1.73) to  [out=270,  in=180] (1,  -2.03);

   \draw [] (E) circle [radius=0.3];
   \draw [  -> ]  (-1.43,  1) to  [out=90,  in=0]  (-1.73,  1.3);	
    \draw [] (F) circle [radius=0.3];
     \draw [  -> ]  (-2.03, -1) to  [out=270,  in=180]  (-1.73,  -1.3);	
     \draw [] (G) circle [radius=0.3];
      \draw [-> ]  (-0.7,  1.73) to  [out=90,  in=0]   (-1,  2.03);
      \draw [] (H) circle [radius=0.3];
         \draw [ -> ]  (-1.3,  -1.73) to  [out=270,  in=180]   (-1,  -2.03);
       \draw [] (I) circle [radius=0.3];
          \draw [ -> ]  (0.3,  2) to  [out=90,  in=0]   (0,  2.3);
        \draw [] (J) circle [radius=0.3];
 \draw [  -> ]  (-0.3,  -2) to  [out=270,  in=180]   (0,  -2.3);
  \node    at (2.3, -0.3)  {\footnotesize $1$};
    \node    at (-2.4, -0.3)  {\footnotesize $-1$};
\end{tikzpicture}
\end{center}
\caption{\footnotesize The jump contour  $\Sigma^{(1)}$  of  $M^{(1)}(z)$.}
\label{Djump65}
\end{figure}
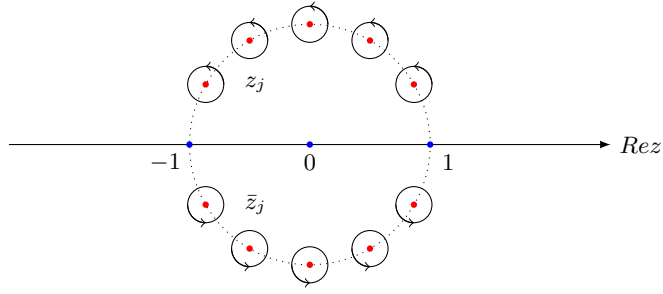

Denote  two  factorizations of  jump matrix  by
\begin{align}
& \left(\begin{array}{cc}
		1 &  \overline{r(z)} T(z)^{-2}e^{-2it\theta(z)} \\
		0 & 1
	\end{array}\right)
		\left(\begin{array}{cc}
			1 & 0\\
r(z)T(z)^2e^{ 2it\theta(z)} & 1
		\end{array}\right):=b_-^{-1} b_+,  \label{opep1}\\
&\left(\begin{array}{cc}
		1 & 0\\
		-\frac{ r(z) }{1-|r|^2}  T_-(z)^{2}  e^{ 2it\theta(z)} & 1
	\end{array}\right)\left(\begin{array}{cc}
	1 & -\frac{ \overline{r(z)}}{1-|r|^2}T_+(z)^{-2} e^{-2it\theta(z)}\\
	0 & 1
\end{array}\right):=B_-^{-1}B_+,  \label{opep2}
\end{align}
and  make the following transformation
       \begin{equation}
        M^{(1)}(z)=T(\infty)^{-\sigma_3} M(z) G(z)T(z)^{\sigma_3}, \label{trans1}
       \end{equation}
 then $M^{(1)}(z)$ satisfies the RH problem as follows.

 \begin{prob} \label{m1}
 Find  $M^{(1)}(z)=M^{(1)}(z;x,t)$ with properties
       \begin{itemize}
        \item   $M^{(1)}(z)$ is  analytical  in $ \mathbb{C} \setminus \Sigma^{(1)}$.
        \item   $M^{(1)}(z)=\sigma_1 \overline{M^{(1)}(\bar{z})}\sigma_1 =z^{-1}M^{(1)}(z^{-1})\sigma_1$.
        \item  $M^{(1)}(z)$ satisfies the jump condition
        \begin{equation*}
            M^{(1)}_+(z)=M^{(1)}_-(z)V^{(1)}(z),
        \end{equation*}
        where
\begin{equation*}
	V^{(1)}(z)=\left\{\begin{array}{ll}
B_-^{-1}B_+,  \  \ \  \;  z\in (0,  \infty),\\[8pt]
b_-^{-1}b_+,  \  \ \  \;  z\in (-\infty, \xi_2) \cup (\xi_1,0),\\[6pt]
T(z)^{-\sigma_3} V(z)T(z)^{\sigma_3}, \  \ \ \;  z\in (\xi_2,  \xi_1),\\[6pt]
		\left(\begin{array}{cc}
			1 & -\frac {  z-z_j} {c_j} T^{-2}(z) e^{ -2it\theta( z_j)} \\
			0 & 1
		\end{array}\right),  \ \  \ \ 	 |z-z_j|=\rho,\\[12pt]
		\left(\begin{array}{cc}
			1 & 0	\\
			 \frac {  z-\bar z_j} {\bar c_j} T^{ 2}(z) e^{  2it\theta(\bar z_j)} & 1
		\end{array}\right),   \ \  \ \	|z-\bar z_j|=\rho.
	\end{array}\right.
\end{equation*}
        \item $ M^{(1)}(z)$ admits the  asymptotic behaviors
        \begin{align*}
                &M^{(1)}(z)=I+\mathcal{O}(z^{-1}),	\quad  z \to  \infty,\\
                &zM^{(1)}(z)=\sigma_1+\mathcal{O}(z), \quad z \to 0.
        \end{align*}

    \end{itemize}
\end{prob}

  Since   the jump matrices on the circles $|z-z_j|=\rho$ or $|z-\bar z_j|=\rho$ exponentially decay to the
  identity matrix as $t \to  \infty$,  it can be shown that   the RH problem \ref{m1} is  asymptotically equivalent to the following  RH problem.

 \begin{prob} \label{m2}
 Find   $M^{(2)}(z)=M^{(2)}(z;x,t)$ with properties
       \begin{itemize}
        \item   $M^{(2)}(z)$ is  analytical  in $\mathbb{C}\backslash \mathbb{R}$.
        \item  $M^{(2)}(z)=\sigma_1 \overline{M^{(2)}(\bar{z})}\sigma_1 =z^{-1}M^{(2)}(z^{-1})\sigma_1$.
        \item $M^{(2)}(z)$ satisfies the jump condition
        \begin{equation*}
            M^{(2)}_+(z)=M^{(2)}_-(z)V^{(2)}(z),
        \end{equation*}
        where
\begin{equation}
	V^{(2)}(z)=\left\{\begin{array}{ll}
B_-^{-1}B_+,    \; z\in (0, \infty),\\[6pt]
b_-^{-1}b_+,    \;  z\in (-\infty, \xi_2) \cup (\xi_1,0),\\[6pt]
T(z)^{-\sigma_3} V(z)T(z)^{\sigma_3},  \;  z\in (\xi_2,  \xi_1).
	\end{array}\right. \label{jumpv2}
\end{equation}

        \item $M^{(2)}(z)$ admits the asymptotic behaviors
        \begin{align*}
                &M^{(2)}(z)=I+\mathcal{O}(z^{-1}),	\quad  z \to  \infty,\\
                &zM^{(2)}(z)=\sigma_1+\mathcal{O}(z), \quad z \to 0.
        \end{align*}

    \end{itemize}
\end{prob}

\begin{proposition}  The solution of   RH problem \ref{m1} can be  approximated by the solution of  RH problem \ref{m2}
 \begin{equation}
     M^{(1)}(z ) =   M^{(2)}(z ) \left( I + \mathcal{O}\left(e^{-c t}\right) \right),\label{trans2}
 \end{equation}
 where $c$ is a constant.
\end{proposition}
\begin{proof}
The result is derived from the theorem of Beals-Coifman and the corresponding   norm estimates.
\end{proof}

Next  we   remove the spectral singularity $z=0$   by  an appropriate  transformation.

\subsubsection{Removing singularity}\label{remove2}
\hspace*{\parindent}

In order to remove the   singularity   $z=0$, we  make a transformation
\begin{align}
M^{(2)}(z)=\left( I+ \frac{1}{z} \sigma_1 M^{(3)}(0)^{-1} \right ) M^{(3)}(z),\label{trans3}
\end{align}
then $M^{(3)}(z)$ satisfies the RH problem without spectral singularity.

 \begin{prob} \label{ms3}
 Find  $M^{(3)}(z)=M^{(3)}(z;x,t)$ with properties
       \begin{itemize}
        \item  $M^{(3)}(z)$ is  analytical  in $\mathbb{C}\backslash \mathbb{R}$.
        \item  $M^{(3)}(z)=\sigma_1 \overline{M^{(3)}(\bar{z})}\sigma_1 =\sigma_1 M^{(3)}(0)^{-1}M^{(3)}(z^{-1})\sigma_1$.
        \item $M^{(3)}(z) $ satisfies the jump condition
        \begin{equation*}
            M^{(3)}_+(z)=M^{(3)}_-(z)V^{(2)}(z),
        \end{equation*}
        where
        $V^{(2)}(z)$ is given by   (\ref{jumpv2}).
        \item $M^{(3)}(z)$ admits the  asymptotics    $  M^{(3)}(z)=I+\mathcal{O}(z^{-1}),	\quad  z \to  \infty.$

    \end{itemize}
\end{prob}

\begin{proof}
We show that if $M^{(3)}(z)$ satisfies the RH problem \ref{ms3}, then $M^{(2)}(z)$ satisfies the RH problem \ref{m2}. Firstly, we verify the jump condition
\begin{align}
M^{(2)}_+(z)=\left( I+ \frac{1}{z} \sigma_1 M^{(3)}(0)^{-1} \right )  M^{(3)}_+(z)= M^{(2)}_-(z) V^{(2)}(z).   \label{udh}
\end{align}
To show the singularity of $z=0$, substituting the expansion
$$M^{(3)}(z) = M^{(3)}(0)+ z \widetilde{M}^{(3)}(z),$$
into (\ref{udh}) yields
\begin{align}
M^{(2)}(z)& =  \frac{1}{z} \sigma_1 + M^{(3)}(0) + z \widetilde{M}^{(3)}(z)+\sigma_1 M^{(3)}(0)^{-1} \widetilde{M}^{(3)}(z)\nonumber\\
&= \frac{1}{z} \sigma_1 +\mathcal{O}(1), \ z\to 0. \nonumber
\end{align}
\end{proof}

\subsection{Transformation to a hybrid $\bar{\partial}$-RH problem} \label{modi2}

\hspace*{\parindent}


In this section, we open the jump contour $\mathbb{R}\setminus (\xi_2, \xi_1)$ by the $\bar{\partial}$ extension.
Denote $\xi_{0}=0$, $\gamma=(\xi_{0}+\xi_{1})/2$, and
    \begin{align*}
    l_1\in \left(0,  |\gamma|\sec \phi\right), \ \  l_2\in \left(0,  |\gamma|\tan \phi\right),
      \end{align*}
where $\phi=\phi(\xi)$.  We define the following rays  passing through $\xi_0, \xi_1$ and $ \xi_2$
 \begin{align*}
     & \Sigma_{0}=e^{i\left(\pi - \phi\right)} l_1,\quad \Sigma_{1}=\xi_1+e^{i\phi} l_1,\; \quad \Sigma_{2}=\xi_2+e^{i\left(\pi - \phi\right)} \mathbb{R}^+, \\
        & \Sigma_{3}=e^{i\phi} \mathbb{R}^+,  \quad L =\gamma+ e^{i\pi/2}l_2,
     \end{align*}
  $\overline{\Sigma}_{j}$ and  $\overline{L} $  denote their  conjugate  rays. These  rays are opened at  a sufficiently small angle $0<\phi <\pi/4$
 such that they all fall into their decaying   regions,
  which correspond to  the signature table of $\re \left(2i\theta(z)\right)$.
  The  opened sectors with  the above  jump lines are denoted by $\Omega_{j}$ and $\overline{\Omega}_{j} $.  See   Figure \ref{signdbar},
       which corresponds to  Figure \ref{figurea}.

 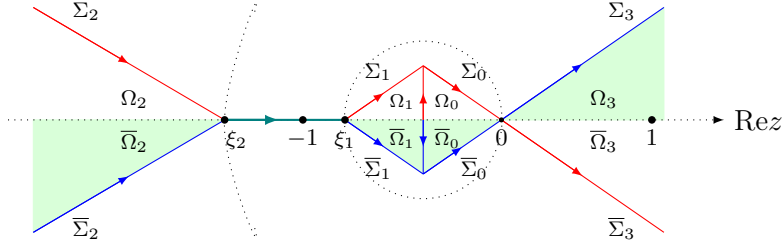
\begin{figure}
        \begin{center}
  \begin{tikzpicture}[scale=0.8]
     \draw[green!15, fill=green!15] (-4.3, 0)--(-7.479, -1.87)--(-7.479,0)--(-4.1, 0);
  \draw[green!15, fill=green!15] (0.3, 0)--(3, 1.87)--(3,0)--(0.3, 0);
  \draw[green!15, fill=green!15] (0.3, 0)--(-2.3,0)--(-1,-0.9)--(0.3, 0);
 \draw [dotted  ] (-3.8,1.8) to [out=75,in=90] (-4.3,0);
  \draw [dotted ] (-3.8,-1.8) to [out=-75,in=-95] (-4.3,0);
                \draw[dotted,-latex](-7.9,0)--(4,0)node[black,right]{Re$z$};

                \node[below] at (0.3,0) {\footnotesize $0$};
                 \node[below] at (-3,0) {\footnotesize $-1$};
                    \node[below] at (2.8,0) {\footnotesize $ 1$};
                                 \draw [ red] (0.3, 0)--(-1,0.9);
                                 \draw[red](-2.3,0)--(-1,0.9);
                               \draw [ red] (0.3,0 )--(3,-1.87);
                             \draw [-latex,red] (-2.3,0)--(-1.65,0.45);
                                 \draw [ red] (-4.3, 0)--(-7.479, 1.87);
                           \draw [-latex, red] (-7.479, 1.87)--(-5.8895, 0.935);
   \draw [-latex,red] (0.3,0)--(1.65,-0.935 );
   \draw[-latex,red](-1,0.9 )--(-0.35,0.45);

   \draw[red](-1,0)--(-1,0.9);
   \draw[red,-latex](-1,0)--(-1,0.45);
   \draw[blue](-1,-0.9)--(-1,0);
   \draw[blue,-latex](-1,0)--(-1,-0.45);

                             \draw [ blue] (-2.3,0)--(-1, -0.9);
                                \draw [-latex,blue] (-2.3,0)--(-1.65,-0.45);

                                \draw [blue] (-4.3, 0)--(-7.479, -1.87);
                                 \draw [thick,teal] (-4.3, 0)--(-2.3, 0);
                                  \draw [thick,teal,-latex](-4.3, 0)--(-3.4,0);

                                \draw [blue] (0.3,0 )--(3,1.87);
                                   \draw [blue] (0.3, 0)--(-1,-0.9);
                                 \draw [-latex,blue]  (-7.479, -1.87)--(-5.8895, -0.935);

                                   \draw [-latex,blue] (0.3,0)--(1.65,0.935 );
                                   \draw[-latex,blue](-1,-0.9 )--(-0.35,-0.45);


                \node[shape=circle,fill=black, scale=0.15]  at (-3,0){0} ;
                \node[shape=circle,fill=black, scale=0.15]  at (2.8,0){0} ;
            \node[shape=circle,fill=black,scale=0.15] at (-2.3,0) {0};
             \node[shape=circle,fill=black,scale=0.15] at (-4.3,0) {0};
                \node[shape=circle,fill=black,scale=0.1] at (0.3,0) {0};

                                                 \node[below] at (-4.1,0) {\scriptsize $\xi_2$};
                                         \node[below] at (-2.3,0) {\scriptsize $\xi_1$};
                                      \node  at (-5.8,0.35) {\scriptsize $\Omega_{2}$};
                                         \node  at (-5.8,-0.35) {\scriptsize $\overline{\Omega}_{2}$};
                                           \node  at (2,0.35) {\scriptsize $\Omega_{3}$};
                                         \node  at (2,-0.35) {\scriptsize $\overline{\Omega}_{3}$};
                                             \node  at (-1.35,0.3) {\tiny $\Omega_{1}$};
                                             \node  at (-0.6,0.3) {\tiny $\Omega_{0}$};
                                             \node  at (-1.35,-0.3) {\tiny $\overline{\Omega}_{1}$};
                                             \node  at (-0.6,-0.3) {\tiny $\overline{\Omega}_{0}$};
                                         \node  at (-6.6,1.8) {\scriptsize  $\Sigma_{2}$};
                                         \node  at (-6.6,-1.8) {\scriptsize $\overline{\Sigma}_{2}$};

                                            \node  at (2.3,1.8) {\scriptsize $\Sigma_{3}$};
                                         \node  at (2.3,-1.8) {\scriptsize $\overline{\Sigma}_{3}$};

                                           \node  at (-1.75,0.8) {\scriptsize $\Sigma_{1}$};
                                             \node  at (-0.15,0.8) {\scriptsize $\Sigma_{0}$};

                                             \node  at (-1.75,-0.8) {\scriptsize  $\overline{\Sigma}_{1}$};
                                             \node  at (-0.15,-0.8) {\scriptsize $\overline{\Sigma}_{0}$};
\draw [dotted ](-1,0) circle (1.31);
  \end{tikzpicture}
            \caption{\footnotesize{Open the jump contour $\mathbb{R}\setminus (\xi_2, \xi_1)$ along  red rays  and blue rays.
          The green regions
             are continuous extension sectors  with $\re \left(2i\theta(z)\right)>0$, while the white regions are  continuous extension sectors with
             $\re \left(2i\theta(z)\right) <0$. }}
      \label{signdbar}
        \end{center}
    \end{figure}

  To determine the decaying properties  of the oscillating factors   $ e^{\pm 2it\theta(z)}$,  we  especially  estimate
  $ \re(2i\theta(z))$ in regions  $\Omega_{j}, \ j=0, 1, 2, 3$.

\begin{proposition}\label{reprop1}  Let $\xi<-1$ and  $(x,t)\in \mathcal{P}_{<-1}(x,t)$.
  Then    the following estimates hold.
 \begin{itemize}
  \item[{\rm   Case I.}] (corresponding to $z=0$)
    \begin{align}
    & \re(2i\theta(z))\geq   |\sin 2\varphi|| v| ,\quad z\in  \overline{\Omega}_{0} \cup \Omega_{3},\label{estm1}\\
 & \re(2i\theta(z))\leq -   |\sin 2 \varphi|| v| ,\quad z\in \Omega_{0}\cup \overline{\Omega}_{3},
    \end{align}
    where $z =u+iv$ and $ \varphi =\arg z  $.
   \item[{\rm   Case II.}] (corresponding to $z=\xi_1$)
    \begin{align}
       & \re(2i\theta(z))\leq -\frac{4}{|\xi_1|}u^2 |v|,\quad z\in \Omega_{1},\\
       & \re(2i\theta(z))\geq \frac{4}{|\xi_1|}u^2 | v|,\quad z\in  \overline{\Omega}_{1},
    \end{align}
    where $z=\xi_1+u+iv$.
     \item[{\rm   Case III.}](corresponding to $z=\xi_2$)
                 \begin{align}\label{omega2est1}
       &	\re(2i\theta(z))\leq  \begin{cases}
        -\frac{1}{8|\xi_2|^3}u^2 |v|,\quad z\in \Omega_{2} \cap\{ |z| \leq 2\},\\[4pt]
       	  -2 \sqrt{2}|v|,\quad z\in \Omega_{2} \cap  \{  |z|>2 \}.	
       	\end{cases}\\
      & 	\re(2i\theta(z))\geq  \begin{cases}
        \frac{1}{8|\xi_2|^3}u^2 |v|,\quad z\in  \overline{\Omega}_{2} \cap \{  |z| \leq 2 \},\\[4pt]
       	  2 \sqrt{2}|v|,\quad z\in  \overline{\Omega}_{2} \cap  \{  |z|>2 \},	
       	\end{cases}
       \end{align}
    where $z=\xi_2+ u+iv$.
 \end{itemize}
  \end{proposition}
  \begin{proof}
   For the case I,  we take   $\Omega_3$   as an  example  to prove the estimate (\ref{estm1}).
  Other cases can be proven similarly.

 For $z\in \Omega_3$, denote the ray $z=|z|e^{i\varphi}=u+iv$ where $0<\varphi<\phi$ and $u>v>0$,  and the function $F(s)=s+s^{-1}$ with $s>0$. Then  (\ref{theta01}) becomes
    \begin{align}
        \re\left(2i\theta(z)\right)=   \sin2\varphi  \left( F(|z|)^2 - \xi \sec \varphi F(|z|)-2\right). \label{eoue}
    \end{align}
Since  $\xi<-1$, we have $- \xi \sec \varphi F(|z|)>0$, and so
      \begin{align}
 &  F(|z|)^2 - \xi \sec \varphi F(|z|)-2 \geq      F(|z|)^2  -F(|z|)   \geq   v.\nonumber
    \end{align}
Substituting the above estimate    into   (\ref{eoue})   gives
    \begin{align*}
 &  \re\left(2i\theta(z)\right)   \geq      v \sin2\varphi,
    \end{align*}
which yields (\ref{estm1}) for  $z\in \Omega_3$.

For cases II and III,  we take   $\Omega_{2}$ as an example.
In this case, for $z=\xi_2 +u +iv = \xi_2 +|z-\xi_2| e^{i\varphi} \in \Omega_{2}$,  we have
 $$u<0, \ v  = -u \tan \varphi>0, \  |z|^2 = (u+\xi_2)^2+u^2\tan^2 \varphi,$$
with $ 0<\varphi<\phi$.  We    prove the estimate (\ref{omega2est1})  to the cases $|z|\leq 2$ and $|z|>2$ respectively.

 For the   region $  |z| \leq 2  $,    with  (\ref{2323}) and  (\ref{xi2}),  we have
     \begin{align*}
     \xi = \frac{\xi_2^4+1}{\xi_2^3+\xi_2}.
     \end{align*}
Substituting  the above formula   into    (\ref{theta01}) gives
     \begin{align*}
     &\re \left(2i\theta(z)\right) = \frac{2v}{|z|^4 \left(\xi_2^3 +\xi_2\right)} \Big\{ \left[ \left(\xi_2^3+\xi_2\right)u +\xi_2^2-1\right]\\
     & \left[ \left(u+\xi_2\right)^2 +u^2\tan^2 \varphi \right]^2  - \left( \xi_2^4+1\right) \left(1+\tan^2 \varphi\right)u^2\\
     &+ \left(-2\xi_2^5+\xi_2^3-\xi_2\right)u +\xi_2^4 -\xi_2^6 \Big\}.
     \end{align*}
With  the properties
 $$\xi_2<-1, \ u<0,\ \left(\xi_2^3+\xi_2\right)u +\xi_2^2-1 \geq  0,      \ \left( \xi_2^3+\xi_2\right) u^3 \geq 0,$$
after  removing  the terms in order  $u^4$ and $u^3$, we find that
 \begin{align}
 \re \left(2i\theta(z)\right) \leq  \frac{2v}{|z|^4 \left(\xi_2^3 +\xi_2\right)} \left( f(\xi_2) u^2 +g(\xi_2) u \right), \label{refg}
 \end{align}
 where
 \begin{align*}
  f(\xi_2) &=4 \xi_2^6 +(9+ \tan^2 \varphi) \xi_2^4 -2(3+\tan^2 \varphi) \xi_2^2 - (1+\tan^2 \varphi),\\
  g(\xi_2) &= \xi_2^7 + 3 \xi_2^5 -3 \xi_2^3 -\xi_2.
 \end{align*}
 Direct calculation  shows  that the function $f(\xi_2)$ is strictly decreasing in  $\xi_2 \in (-\infty,-1)$, and hence
  $$f(\xi_2) > f(-1)=6-2 \tan^2 \varphi.$$
  For $\varphi \in (0, \frac{\pi}{4})$, we have  $\tan^2 \varphi \in (0,1)$, $f(-1) \in (4,6)$, and so  $f(\xi_2)>4$.

 Similarly, $g(\xi_2)$ is strictly increasing in  $\xi_2 \in (-\infty,-1)$  and   $g(-1)=0$,   then $g(\xi_2)u>0$.
 Therefore, (\ref{refg}) becomes
  \begin{align*}
 \re \left(2i\theta(z)\right) \leq  \frac{8u^2v}{|z|^4 \left(\xi_2^3 +\xi_2\right)},
 \end{align*}
 together with $|z| \leq 2$, gives us  the  estimate  (\ref{omega2est1}).

 Next  we prove  (\ref{omega2est1}) in  the region  $|z|>2$.
        In this case, we notice that
        $$|z|=|u+\xi_2| \sec w >2,$$
          with $ w=\arg z$  and  $0<w<\varphi$,  which implies that
          $$u+\xi_2 <-2 \cos w.$$
      Further with  (\ref{theta01}), we have
          \begin{align}
         & \re \left(2i\theta(z)\right)= 2v \left(1+\frac{1}{|z|^4} \right) \left( u+\xi_2 - \xi   \frac{|z|^4+|z|^2}{|z|^4+1} \right) \leq 2v h(\xi),\label{re20}
          \end{align}
where
  \begin{align*}
   h(\xi) = -2 \cos w -(\xi+1) \left(1+\frac{1}{4\cos^2 w} \right) -1+ \frac{1}{4\cos^2 w} .
  \end{align*}
Let $t \to  +\infty$, we have  $\xi \to -1^-$ and
  \begin{align*}
  h(\xi) \to h(-1) = -2 \cos w-1+ \frac{1}{4\cos^2w}.
  \end{align*}
  Since $0<\varphi<\frac{\pi}{4}$, we have $\frac{\sqrt{2}}{2}<\cos w<1$ and then $h(-1 ) \in (-\frac{11}{4},-\frac{2\sqrt{2}+1}{2})$.
  Therefore, there exists a large   $T$  such that when $t>T$,
   \begin{equation}\label{estfxi}
h(\xi) \le h(-1 )+ \frac{1}{2} <-\sqrt{2}.
  \end{equation}
  Substituting (\ref{estfxi}) into  (\ref{re20}) gives the second estimate when $|z|>2$  of  (\ref{omega2est1}). We finish the proof.
  \end{proof}

Next we open  the contour $\mathbb{R}\setminus [\xi_2, \xi_1]$  via  continuous extensions of the jump matrix $V^{(2)}(z)$
by defining   appropriate functions.

 \begin{proposition} \label{prop3}
     Let  $q_0 \in \tanh (x)+H^{4,4}(\mathbb{R})$ and  define   functions $R_{j}(z)(j=0,1,2,3) $ with  boundary values
       \begin{align}
        &R_{j}(z)= \begin{cases}
             r(z)T(z)^2, \quad  z \in (-\infty,\xi_2) \cup (\xi_1,0),\\
              r\left( \xi_j\right)T(\xi_j)^2,\quad z \in \Sigma_{j},\ j=0,1,2,
            \end{cases}\\
        &\overline{R}_{j}(z)= \begin{cases}
         \overline{r(z)}T(z)^{-2}, \quad z  \in (-\infty,\xi_2) \cup (\xi_1,0),\\
                 \overline{r\left( \xi_j\right)}T(\xi_j)^{-2},\quad z \in \overline{\Sigma}_{j},\ j=0,1,2,
               \end{cases}\\
       &R_{3}(z)= \begin{cases}
            \frac{\overline{r(z)}T_+(z)^{-2}}{1-|r(z)|^2}, \quad z \in (0,\infty),\\
             0,\quad z \in \Sigma_{3},
           \end{cases} \label{r31}\\
        &\overline{R}_{3}(z)= \begin{cases}
                 \frac{r(z)}{1-|r(z)|^2} T_-(z)^{2}, \quad  z \in (0,\infty),\\
                 0,\quad z \in \overline{\Sigma}_3,
               \end{cases}\label{r32}
       \end{align}
     where $r(\xi_0)=r(0)=0$.   Then there exists a  constant  $c =c(q_0)$ which only depends on $q_0$ such that

     For $ j=1,2$,
 \begin{equation}
  |\bar{\partial}R_{j}| \le c\left( |\varphi(\re(z))|+  | r'\left( \re\left( z \right)\right)|+    | z-\xi_j |^{-1/2} \right),   \;z\in \Omega_j\cup \overline{\Omega}_j.\label{437}
\end{equation}
For $ j=0,3$,
       \begin{equation}
	|\bar{\partial}R_{ j}| \le \begin{cases}
c \left( |\varphi(\re z)|+  | r'\left( \re z \right)|+    |z|^{-1/2} \right),  \;z\in \Omega_j \cup \overline{\Omega}_j,\\
c |z-1|, \quad \text{near}\; z= 1, \label{438}
\end{cases}
\end{equation}
where $\varphi \in C_0^\infty \left( \mathbb{R}, \left[0,1\right] \right)$ with small support near $z=1$.

        \end{proposition}

\begin{proof}   The extension functions  $R_j(z),  j=0,1,2$ and their $\bar\partial$ estimates
are easy obtained by defining
\begin{equation}
R_j(z)=\cos (k\arg z)  r(|z|)T( z)^2+  [1-\cos( k\arg z )] r ( \xi_j )T(\xi_j)^2, \label{429}
\end{equation}
where $ z\in \Omega_j, j=0,1,2$ and  $k=\frac{\pi}{2\phi}$.

As for $R_3(z)$, in the non-generic case $|r(z)|<1, \ z\in \mathbb{R}$, we define
\begin{equation}
R_3(z)=   \frac{\overline{r(|z|)} }{1-|r(|z|)|^2}  T ( z)^{-2}\cos(k\arg z), \ \ z\in \Omega_3. \label{430}
\end{equation}
In the  generic case  $|r(1)|=1$, as observed in (\ref{r31})-(\ref{r32}),
  $R_3(z)$  is  singular at $z= 1$,
 however,  the   singularity  can be balanced by the factor $T(z)^{-2}$  following   Cuccagna and Jenkins's method \cite{CJ}.

It follows from Lemmas  \ref{lemma3.2} that
\begin{equation}
\frac{\overline{r  (z )} }{1-|r( z )|^2} T_+^{-2}(z)=\frac{\overline{s_{21}(z )} }{s_{11}(z)}\left(\frac{ s_{11}(z) }{T_+(z)}\right)^2 =\frac{ \overline{S_{21}(z)} }{S_{11}(z)} \left(\frac{ s_{11}(z) }{T_+(z)}\right)^2,
\label{oeeedk}
\end{equation}
where
\begin{equation}
S_{21}(z)=\det (\psi_1^+ (z),  \psi_1^-(z)), \ \ \ S_{11}(z)=\det ({\psi_1^-}(z),  {\psi_2^+}(z)).
\label{oeeedk2}
\end{equation}
Then  the denominator of each factor in the r.h.s. of (\ref{oeeedk}) is non-zero and analytical in $ \Omega_3$,
 with well defined nonzero limit on $\partial \Omega_3$.

Introduce  a  cutoff function $\chi_0(z), \chi_1(z)\in C_0^\infty(\mathbb{R},[0,1])$ with small support near $z=0$ and  $z=1$ respectively.
Define  the extension function   $R_3(z)=R_{11}(z)+R_{12}(z) $ with
\begin{align}
&R_{11}(z) =(1-\chi (z)) \frac{\overline{r(|z|)} }{1-|r(|z|)|^2}  T ( z)^{-2}\cos(k\arg z),\label{oedk} \\
& R_{12}(z)=f(|z|) g(z) \cos(k\arg z) + \frac{i|z|}{k} \chi_0(\frac{ \arg z}{\delta_0}) f'(|z|) g(z) \sin (k\arg z).\nonumber
\end{align}
where $\delta_0$ is a fixed small constant, $f'(s)$ is the derivative of $f(s)$,  and
\begin{align}
&g(s)= \frac{ s_{11}(s)^2}{T_+(s)^2}, \ \  f(s) =\chi_1(s) \frac{ \overline{S_{21}(s)} }{S_{11}(s)}.\nonumber
\end{align}

In this way, the effect of the singularity at $z=1$ can be neutralized.  The details of the proof can be found  in \cite{CJ}.
We omit it here.

\end{proof}

Further,  with these functions $R_j$, we  define a matrix function
 \begin{align}\label{R3}
    	R^{(3)}(z)=\begin{cases}
          \left( \begin{array}{cc}
    	1& 0\\-R_je^{2it\theta(z)}&1
    \end{array} \right),\quad z \in \Omega_j,\;  j = 0,1,2,\\
    \left(\begin{array}{cc}
    	1&  -\overline{R}_je^{-2it\theta(z)}\\0&1
    \end{array} \right), \quad z \in \overline{\Omega}_{j},\;  j = 0,1,2,\\
     		\left(\begin{array}{cc}
    	1&  R_3 e^{-2it\theta(z)}\\0&1
    \end{array} \right), \quad z \in \Omega_{3},\\
          		\left(\begin{array}{cc}
    	1&  0\\ \overline{R}_3e^{2it\theta(z)}&1
    \end{array} \right), \quad z \in \overline{\Omega}_{3},\\
    I, \quad {\rm others},
    	\end{cases}
    \end{align}
and a contour
$$\Sigma^{(4)}=\Gamma\cup \overline{\Gamma}\cup [\xi_2, \xi_1], \ \Gamma= \cup_{j=0}^3 \Sigma_j \cup L.$$
  Then the new matrix function
    \begin{equation}\label{trans4}
        M^{(4)}(z)=M^{(3)}(z)R^{(3)}(z)
    \end{equation}
 satisfies the   following hybrid  $\bar{\partial}$-RH problem.

\begin{prob2}
    Find    $M^{(4)}(z)=M^{(4)}(z;x,t)$  with properties
    \begin{itemize}
        \item  $M^{(4)}(z)$ is continuous in $\mathbb{C}\setminus  \Sigma^{(4)} $. See Figure \ref{jumpm4}.
        \item $M^{(4)}(z)$ satisfies the  jump condition
        \begin{equation*}
            M^{(4)}_+(z)=M^{(4)}_-(z)V^{(4)}(z),
        \end{equation*}
     \end{itemize}
        where
       \begin{align}\label{V3}
       	V^{(4)}(z)=\begin{cases}
           \left( \begin{array}{cc}
    	1& 0\\ r\left( \xi_j\right)T(\xi_j)^2 e^{2it\theta(z)}&1
    \end{array} \right), \  z \in \Sigma_j, \ j=1,2;\vspace{2mm}\\
    \left(\begin{array}{cc}
    	1&  -\overline{r\left( \xi_j\right)}T(\xi_j)^{-2} e^{-2it\theta(z)}\\0&1
    \end{array} \right), \  z \in \overline{\Sigma}_j,\ j=1,2;\vspace{2mm}\\
     	\left(	\begin{array}{cc}
     	1& 0\\
     	-r\left( \xi_1\right)T(\xi_1)^2 e^{2it\theta(z) } & 1
     \end{array}\right),\ z\in L;\vspace{2mm}\\
            \left(		\begin{array}{cc}
       	1& \overline{r\left( \xi_1\right)}T(\xi_1)^{-2} e^{-2it\theta(z)}\\
       	0 & 1
       \end{array}\right),\  z\in \overline{L};\vspace{2mm}\\
       T(z)^{-\sigma_3} V(z)T(z)^{\sigma_3},\quad z\in (\xi_2,\xi_1).
       	\end{cases}
       \end{align}
    \begin{itemize}
        \item $ M^{(4)}(z)=I+\mathcal{O}(z^{-1}),	\quad  z \to  \infty.$

        \item For $z\in \mathbb{C}\setminus \Sigma^{(4)} $, we have
        \begin{equation*}
            \bar{\partial}M^{(4)}(z)= M^{(4)}(z) \bar{\partial}R^{(3)}(z),
        \end{equation*}
        where
        \begin{equation}\label{parR2}
            \bar{\partial}R^{(3)}(z)= \begin{cases}
                \left( \begin{array}{cc}
    	1& 0\\-\bar{\partial}R_je^{2it\theta(z)}&1
    \end{array} \right), \quad  z \in \Omega_j,\; j = 0,1,2,\\
    \left(\begin{array}{cc}
    	1&  -\bar{\partial} \overline{R}_je^{-2it\theta(z)}\\0&1
    \end{array} \right), \quad z \in \overline{\Omega}_j,\; j= 0,1,2,\\
            		\left(\begin{array}{cc}
    	1&  \bar{\partial}R_3e^{-2it\theta(z)}\\0&1
    \end{array} \right), \quad z \in \Omega_3,\\
          		\left(\begin{array}{cc}
    	1&  0\\\bar{\partial} \overline{R}_3e^{2it\theta(z)}&1
    \end{array} \right), \quad z \in \overline{\Omega}_3.
            	\end{cases}
           \end{equation}

\end{itemize}
    \end{prob2}

	We decompose $M^{(4)}(z)$ into a pure RH problem  $M^{rhp}(z)$ with $\bar{\partial} R^{(3)}(z)=0$ and a pure $\bar{\partial}$-problem $M^{(5)}(z)$ with $\bar{\partial} R^{(3)} (z)\neq 0$
in the form
	\begin{align}
M^{(4)}(z)=M^{(5)}(z) M^{rhp}(z). \label{trans5}
\end{align}

   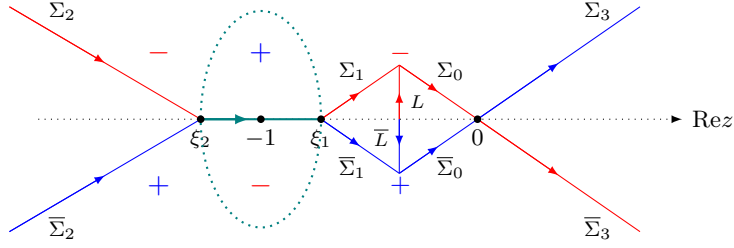
\begin{figure}
        \begin{center}
  \begin{tikzpicture}[scale=0.8]
                \draw[dotted, -latex](-7,0)--(3.7,0)node[black,right]{\footnotesize Re$z$};

                \node[below] at (0.3,0) {\footnotesize $0$};
                 \node[below] at (-3.3,0) {\footnotesize $-1$};
                                 \draw [ red] (0.3, 0)--(-1,0.9);
                                 \draw[red](-2.3,0)--(-1,0.9);
                               \draw [ red] (0.3,0 )--(3,-1.87);
                             \draw [-latex,red] (-2.3,0)--(-1.65,0.45);
                                 \draw [ red] (-4.3, 0)--(-7.479, 1.87);
                           \draw [-latex, red] (-7.479, 1.87)--(-5.8895, 0.935);
   \draw [-latex,red] (0.3,0)--(1.65,-0.935 );
   \draw[-latex,red](-1,0.9 )--(-0.35,0.45);
  \draw[red](-1,0)--(-1,0.9);
   \draw[red,-latex](-1,0)--(-1,0.45);
   \draw[blue](-1,-0.9)--(-1,0);
   \draw[blue,-latex](-1,0)--(-1,-0.45);
    \draw [ blue] (-2.3,0)--(-1, -0.9);
          \draw [-latex,blue] (-2.3,0)--(-1.65,-0.45);
   \draw [blue] (-4.3, 0)--(-7.479, -1.87);
                                 \draw [thick,teal] (-4.3, 0)--(-2.3, 0);
                                  \draw [thick,teal,-latex](-4.3, 0)--(-3.5,0);
        \draw [blue] (0.3,0 )--(3,1.87);
        \draw [blue] (0.3, 0)--(-1,-0.9);
                \draw [-latex,blue]  (-7.479, -1.87)--(-5.8895, -0.935);

             \draw [-latex,blue] (0.3,0)--(1.65,0.935 );
                      \draw[-latex,blue](-1,-0.9 )--(-0.35,-0.45);
          \node[red,thick]  at (-5,1.1 ) {\bf  $-$};
                  \node[blue,thick]  at (-5,-1.1 ) {\bf  $+$};
                                    \node[red,thick]  at (-0.98,1.1) {\small\bf $-$};
                                     \node[blue,thick]  at (-0.98,-1.1) {\small\bf $+$};
                           \node[red,thick]  at (-3.3,-1.1) {\bf  $-$};
                                \node[blue,thick]  at (-3.3,1.1) {\bf  $+$};

                                                 \node[below] at (-4.3,0) {\scriptsize $\xi_2$};
                                         \node[below] at (-2.3,0) {\scriptsize $\xi_1$};
                                         \node  at (-6.6,1.8) {\scriptsize  $\Sigma_2$};
                                         \node  at (-6.6,-1.8) {\scriptsize $\overline{\Sigma}_2$};

                                            \node  at (2.3,1.8) {\scriptsize $\Sigma_3$};
                                         \node  at (2.3,-1.8) {\scriptsize $\overline{\Sigma}_3$};

                                           \node  at (-1.78,0.8) {\scriptsize $\Sigma_1$};
                                             \node  at (-0.15,0.8) {\scriptsize $\Sigma_0$};

                                             \node  at (-1.78,-0.8) {\scriptsize  $\overline{\Sigma}_1$};
                                             \node  at (-0.15,-0.8) {\scriptsize $\overline{\Sigma}_0$};
                                                      \node  at (-0.7,0.3) {\tiny $L$};
                                                \node  at (-1.3,-0.3) {\tiny $\overline{L}$};
  \draw[dotted,teal,thick] (-3.3,0) ellipse (1 and 1.8);
                    \node[shape=circle,fill=black, scale=0.15]  at (-3.3,0){0} ;
            \node[shape=circle,fill=black,scale=0.15] at (-2.3,0) {0};
             \node[shape=circle,fill=black,scale=0.15] at (-4.3,0) {0};
                \node[shape=circle,fill=black,scale=0.15] at (0.3,0) {0};
  \end{tikzpicture}
            \caption{ \footnotesize{ The jump contours of  $M^{(4)}(z) $ and  $M^{rhp}(z) $  }}
      \label{jumpm4}
        \end{center}
    \end{figure}

    \subsection{Long-time analysis on  a pure RH problem} \label{modi3}
  \hspace*{\parindent}

In this subsection,
we find a local solution $M^{loc}(z) $,  which approximates to $M^{rhp}(z)$ near $ z=-1$.
The pure RH problem  is given as follows.
	
\begin{prob}\label{mrhp}
    Find  $M^{rhp}(z)=M^{rhp}(z;x,t)$ which satisfies
	  \begin{itemize}
	  	\item  $M^{rhp}(z)$ is analytical in $\mathbb{C}\backslash \Sigma^{(4)}$. See   Figure \ref{jumpm4}.
	  	\item $M^{rhp}(z)$  satisfies the jump condition
\begin{equation*}
	  		M^{rhp}_+(z)=M^{rhp}_-(z)V^{(4)}(z),
	  	\end{equation*}
	  	where $V^{(4)}(z)$ is given by (\ref{V3}).
	  	\item   $M^{rhp}(z)$  has the same asymptotics with  $M^{(4)}(z)$.

	  \end{itemize}
\end{prob}

\subsubsection{Local paramatrix near $z=-1$}
\hspace*{\parindent}

In the region $ -C< (\xi+1) t^{2/3}<0$,  we first notice that
 $ \xi \to -1^- \   as \  t \to  \infty, $
 further  from
(\ref{xi2}), it is found that  two  phase points  $\xi_1$ and $\xi_2$ also will merge to $z=-1$.
Making  the asymptotic expansion of the phase function $t\theta(z)$   near   $z=-1 $,  we find that
\begin{align}
 t \theta(z) &= t \left(\textcolor{red} { (z+1)^3}   +\frac{1}{2}\sum_{n=4}^\infty  (n-1)(z+1)^n  \right)\nonumber\\
 & +(x+2t) \left(\textcolor{red} {  (z+1)}  + \frac{1}{2} \sum_{n=2}^\infty (z+1)^n  \right) \nonumber\\
 &: =    \frac{4}{3} k^3 + s k + S(t;k),\label{asymn1}
\end{align}
where the scaled parameters are given by
\begin{align}
 & k  = \tau^{\frac{1}{3}} (z+1), \quad
  s   = \frac{8}{3}  (\xi +1) \tau^{\frac{2}{3}}, \label{scaled}
\end{align}
with $\tau=  \frac{3}{4} t$.  The first two terms  $ \frac{4}{3} k^3 + s k $   play  a key  role in
matching  the   Painlev\'e  model  in the local region, and   the  remainder term  is given by
\begin{align}
 &S(t;k ) =  \frac{2}{3}   \sum_{n=4}^\infty (n-1)  \tau^{\frac{3-n}{3}} k^n+ \frac{4}{3}  (\xi+1) \sum_{n=2}^\infty \tau^{\frac{3-n}{3}} k^n. \label{sers}
\end{align}

Next we show that  two scaled phase points $k_j= \tau^{1/3} (  \xi_j+ 1  ), \ j=1,2$ are always in a fixed interval in the $k$-plane.
\begin{proposition} \label{opow}
 In the transition region  $\mathcal{P}_{< -1}(x,t)$ and under scaling transformation (\ref{scaled}),   we have
  \begin{align}
 k_j\in ( -(3/4)^{1/3} \sqrt{2C}, (3/4)^{1/3} \sqrt{2C}), \ j=1, 2.
        \end{align}

\end{proposition}

\begin{proof} From (\ref{eq219}), the  phase point $\xi_1$ satisfies the equation
\begin{align}
&\xi_1^2 +1 =\eta_1 \xi_1  \Longrightarrow    ( \xi_1 +1 )^2=(\eta_1+2)\xi_1. \label{por40}
\end{align}
Noting that $\eta_1 +2 <0, \ \xi_1 >-1$, and  using (\ref{2323}), we have
\begin{align}
& (\eta_1+2) \xi_1 <  -(\eta_1 +2)  =    -\frac{1}{2} \xi +\frac{1}{2}\sqrt{\xi^2+8} -2
 \leq  -2(\xi+1).\label{23236}
\end{align}
Substituting (\ref{23236}) into (\ref{por40}) gives
\begin{align*}
( \xi_1 +1 )^2 < -2( \xi +1) <  2Ct^{-2/3},
\end{align*}
 which implies that
\begin{align*}
&     k_1 = \tau^{1/3} (  \xi_1+ 1  )  <   (3/4)^{1/3} \sqrt{2C}.
\end{align*}

Further, by  the symmetry $ \xi_1\xi_2=1$,    we obtain
 \begin{align*}
&   \xi_2 +1 =  \frac{1}{\xi_1} +1    =  \frac{ \xi_1+1}{\xi_1}.
\end{align*}
Therefore,  for  $-1<\xi_1<0$, $k_2$ can be controlled  by a constant
 \begin{align*}
&   k_2 = \tau^{1/3} (  \xi_2 +1   ) =  \tau^{1/3}  \frac{ \xi_1+1}{\xi_1} > -(3/4)^{1/3} \sqrt{2C}.
\end{align*}

\end{proof}

Let  $t$  be  large enough so that $\sqrt{2C} \tau^{-1/3}<\rho$ where $\rho$ has been defined in (\ref{definerho}).
For a fix constant  $\varepsilon \leq \sqrt{2C}$, define  two open disks
  \begin{align}
&   \mathcal{U}_{z}(-1) = \{z \in \mathbb{C}: |z+1|< \varepsilon \tau^{-1/3}\},\nonumber\\
&  \mathcal{U}_{k }(0) = \{k \in \mathbb{C}: |k|< \varepsilon \},\nonumber
\end{align}
whose   boundaries are  oriented counterclockwise.
The transformation  defined by (\ref{scaled})
defines a map $z \mapsto k$  maps   $ \mathcal{U}_{z}(-1)$  onto the disk $ \mathcal{U}_{k }(0)$ of the radius $\varepsilon$ in the $k$-plane.
Proposition \ref{opow}  implies that  for  large $t$,   we have   $\xi_1, \xi_2 \in \mathcal{U}_{z} (-1)$,
and also  $k_1, k_2 \in\mathcal{U}_{k} (0)$. See Figure \ref{scalingzk}.

\begin{figure}
        \begin{center}
  \begin{tikzpicture}[scale=0.8]
  \draw[dotted,thick ](-6,0)--(-4.3,0);
                \draw[dotted,thick,-latex](-2.3,0)--(0,0) node[black,right]{\footnotesize Re$z$};
                 \node[below] at (-3.3,0) {\footnotesize $-1$};
                  \draw [dotted,thick,teal](-3.3,0) circle (1.9);
                                 \draw[red](-2.3,0)--(-1,0.9);
                             \draw [-latex,red] (-2.3,0)--(-1.65,0.45);
                                 \draw [ red] (-4.3, 0)--(-6.3, 1.2);
                           \draw [-latex, red] (-6.3, 1.2)--(-5.3, 0.6);
    \draw [ blue] (-2.3,0)--(-1, -0.9);
          \draw [-latex,blue] (-2.3,0)--(-1.65,-0.45);
   \draw [blue] (-4.3, 0)--(-6.3, -1.2);
   \draw [-latex,blue]  (-6.3,- 1.2)--(-5.3, -0.6);
                                 \draw [thick,teal] (-4.3, 0)--(-2.3, 0);
                                  \draw [thick,teal,-latex](-4.3, 0)--(-3.6,0);
\node[]  at (-3.3,2.2) {\scriptsize $\partial \mathcal{U}_{z}(-1)$};

 \node[below] at (-4.3,0) {\footnotesize $\xi_2$};
                                         \node[below] at (-2.3,0) {\scriptsize $\xi_1$};
                                         \node  at (-6,1.5) {\scriptsize $\Sigma_{2}$};
                                         \node  at (-6,-1.6) {\scriptsize $\overline{\Sigma}_{2}$};
                                           \node  at (-1,1.2) {\scriptsize $\Sigma_{1}$};
                                             \node  at (-1,-1.2) {\scriptsize $\overline{\Sigma}_{1}$};

                                             \node  at (-5.5,0.3) {\scriptsize $\Omega_{2}$};
                                         \node  at (-5.5,-0.3) {\scriptsize $\overline{\Omega}_{2}$};
                                           \node  at (-1,0.3) {\scriptsize $\Omega_{1}$};
                                             \node  at (-1,-0.3) {\scriptsize $\overline{\Omega}_{1}$};
    \draw [dotted,  ->  ]  (-1.4, 0)  to  [out=90,  in=0] (-3.3,  1.9);
                \node[shape=circle,fill=black, scale=0.15]  at (-3.3,0){0} ;
            \node[shape=circle,fill=black,scale=0.15] at (-2.3,0) {0};
             \node[shape=circle,fill=black,scale=0.15] at (-4.3,0) {0};

 \draw[dotted,thick ](1.6,0)--(3.2,0);
 \draw[dotted,thick,-latex](4.8,0)--(7,0)node[black,right]{\footnotesize Re$k$};
 \node[below] at (4,0) {\footnotesize $0$};
   \draw [dotted,thick,teal](4,0) circle (1.6);
\draw [thick,teal] (3.2, 0)--(4.8, 0);
  \draw [thick,teal,-latex](-4.3, 0)--(-3.6,0);
\node[shape=circle,fill=black, scale=0.15]  at (4,0){0} ;
            \node[shape=circle,fill=black,scale=0.15] at (3.2,0) {0};
             \node[shape=circle,fill=black,scale=0.15] at (4.8,0) {0};
     \draw[red](4.8,0)--(6.2,1);
    \draw [-latex,red] (4.8,0)--(5.5,0.5);
   \draw [ red] (3.2, 0)--(1.4, 1.5);
     \draw [-latex, red] (3.2, 0)--(2.3, 0.75);

\draw[blue](4.8,0)--(6.2,-1);
    \draw [-latex,blue] (4.8,0)--(5.5,-0.5);
   \draw [blue] (3.2, 0)--(1.4,- 1.5);
     \draw [-latex, blue] (3.2, 0)--(2.3, -0.75);
   \node[below] at (3.3,0) {\footnotesize $k_2$};
 \node[below] at (4.8,0) {\footnotesize $k_1$};
  \node[]  at (4.3, 2) {\scriptsize $\partial \mathcal{U}_{k}(0)$};
      \node  at (6,1.4) {\footnotesize $\Sigma_1'$};
                                         \node  at (6,-1.4) {\scriptsize $\overline{\Sigma}_1'$};
                                           \node  at (2,1.4) {\scriptsize $\Sigma_2'$};
                                             \node  at (2,-1.4) {\scriptsize $\overline{\Sigma}_2'$};

                                                   \node  at (6,0.35) {\scriptsize $\Omega_1'$};
                                         \node  at (6,-0.35) {\scriptsize $\overline{\Omega}_1'$};
                                           \node  at (2,0.35) {\scriptsize $\Omega_2'$};
                                             \node  at (2,-0.35) {\scriptsize $\overline{\Omega}_2'$};
                                             \draw [dotted,  ->  ]  (5.6, 0)  to  [out=90,  in=0] (4,  1.6);
  \node    at (-3.3, -2.5 )  {  $ (a)$};
    \node    at (4, -2.5 )  {  $ (b)$};
  \end{tikzpicture}
            \caption{\footnotesize{ The map relation  between two  disks  $\mathcal{U}_z(-1)$  and   of  $\mathcal{U}_k(0)$.   }}
      \label{scalingzk}
        \end{center}
    \end{figure}
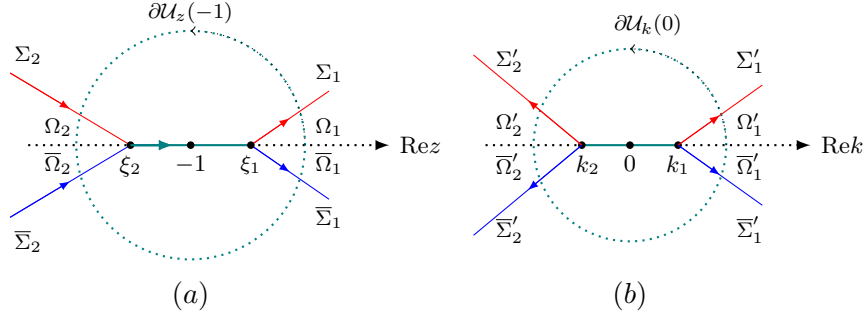

We show   that when $t$ is sufficiently large, $\xi$ is close  to $z=-1$ and $z$ is close to  $z=-1$, the phase function $ t \theta(z)$ can be  approximated by $ \frac{4}{3} k^3 + s k $.
For this purpose,  we first prove that the series  $S(t;k)$ converges uniformly in  $\mathcal{U}_{k}(0)$
 and  decays  with respect to  $t$.

\begin{proposition}\label{order1}
  Let  $(x,t)\in  \mathcal{P}_{< -1}(x,t)$,   then for  $k\in \mathcal{U}_k(0)$,
we have the uniform convergence estimate
$$| S(t;k))| \leq c  t^{-1/3},  \ t\to +\infty, $$
where $c$ is only dependent on the radius of   $\mathcal{U}_k(0)$.
\end{proposition}

\begin{proof}
We decompose  the series  $S(t;k)$ into
\begin{align*}
S (t;k)= S_1 +S_2,
\end{align*}
where
\begin{align*}
S_1 =  \frac{2}{3}   \tau^{-1/3}   k^4  \sum_{n=4}^\infty (n-1) \left(\tau^{-\frac{1}{3}} k\right)^{n-4}, \ \ S_2 = \frac{4}{3}(\xi+1)\tau^{1/3}   k^2      \sum_{n=2}^\infty \left(\tau^{-\frac{1}{3}}k\right)^{n-2}.
\end{align*}

For the first  series $S_1$,
by  $\lim_{n \to \infty}\frac{n-1}{n} = 1$ and $|k|<\varepsilon$,
we have
\begin{align}
 & |S_1|\leq \frac{2}{3}    \tau^{-1/3}  \varepsilon^4   \sum_{n=4}^\infty (n-1) (\tau^{-1/3}\varepsilon)^{n-4}\leq c t^{-1/3},  \nonumber
\end{align}
which implies that $S_1$  converges uniformly and  decays  with respect to  $t$.

 Similarly, by $|(\xi+1)t^{2/3} |<C$ and $|k|<\varepsilon $,
 we have
 \begin{align}
 & |S_2|\leq    \left(\frac{3}{4}\right)^{-\frac{1}{3}}  C \tau^{-1/3} \varepsilon^2    \sum_{n=2}^\infty  (\tau^{-1/3}\varepsilon)^{n-2}\leq c t^{-1/3}.   \nonumber
\end{align}

\end{proof}

To show $ e^{  i  (\frac{4}{3} k^3 + sk  )} $ is bounded in the disk  $\mathcal{U}_k(0)$ (see Figure \ref{scalingzk} (b)), we give the estimates.

\begin{proposition} \label{opow1}
Let   $(x,t)\in \mathcal{P}_{<-1}(x,t)$, then for  large $t$, we have
\begin{align}
  &{\rm Re}\left[i  \left(\frac{4}{3} k^3 + sk \right)  \right]  \leq - \frac{8}{3}u^2v, \ \  k\in \Omega_j', \label{re011}\\
  &{\rm Re}\left[i  \left(\frac{4}{3} k^3 + sk \right)  \right]  \geq \frac{8}{3}u^2v,\ \   k \in  \overline{\Omega}_j', \label{re022}
\end{align}
where  $k=k_j+u+iv$ is the scaled variable.

\end{proposition}

\begin{proof}  We only show  (\ref{re011}) and the proof  of  (\ref{re022}) is  similar.
The estimate  (\ref{re011}) is equivalent to  the following  estimate
\begin{align}
  &{\rm Re}\left[2i  \left(t(z+1)^3 + (x+2t) (z+1) \right)  \right]  \leq  -4tu^2v,   \ z\in \Omega_{j},   \label{re01}
\end{align}
where  $z=\xi_j+u+iv, \ j=1,2$.    For $j=1$, $u>0$ and $v>0$, while for $j=2$, $u<0$ and $v>0$.
 Since $\arg (z- \xi_1) <\frac{\pi}{4}$ and  $  \pi -\pi/4 <\arg (z- \xi_2) < \pi$, we have $v<|u|$ and    direct calculations show that
\begin{align*}
  {\rm Re}\left[2i  \left(t(z+1)^3 + (x+2t) (z+1) \right)  \right] \leq    -4tu^2v + 2tv f(\xi),
\end{align*}
where
\begin{align}
f(\xi) = -3(\xi_j+1)^2-2(\xi+1)-6(1+(-1)^{j+1})(\xi_j+1)u.\label{f0}
\end{align}
From (\ref{2theta}) and (\ref{eq219}),  we obtain
 \begin{align}
&\eta_1^2 -\xi \eta_1-2=0  \Longrightarrow  \eta_1+2 = \frac{(\xi+1)\eta_1}{\eta_1-1},  \label{eta1}\\
&\xi_j^2 +1 =\eta_1 \xi_j  \Longrightarrow    ( \xi_j +1 )^2=(\eta_1+2)\xi_j,\quad j=1,2. \label{por4}
\end{align}
 Substituting (\ref{eta1}) and (\ref{por4}) into (\ref{f0}) yields
\begin{align}\label{f2}
f(\xi) = \frac{\xi+1}{1-\eta_1}  \left(3 \eta_1 \xi_j +2 \eta_1 - 2\right)-6(1+(-1)^{j+1})(\xi_j+1)u,
\end{align}
in which
\begin{align}\label{eta12}
3 \eta_1 \xi_j +2 \eta_1 - 2 =\frac{1}{\xi_j}  (\xi_j+1)\big[ (\xi_j-1/2)^2+4/7\big].
\end{align}
Further substituting (\ref{eta12}) into (\ref{f2}) yields
\begin{align*}
f(\xi) = (\xi_j+1)g(\xi),
\end{align*}
where
\begin{align}
g(\xi)= \frac{\xi+1}{(1-\eta_1)\xi_j}\big[ (\xi_j-1/2)^2+4/7\big] -6(1+(-1)^{j+1})u.\label{wow}
\end{align}

For  $j=1$,  as $t \to +\infty$, $g(\xi) \to -6 u <0$. Thus there exists  a sufficiently large time $T$ such that for  $t>T$,  $g(\xi)<0$, which
 together with $\xi_1+1>0$, we get  $f(\xi)<0$.

For  $j=2$,    $\xi +1< 0, \  \xi_2<0, \ 1-\eta_1>0,$  and then (\ref{wow}) yields
$$g(\xi)= \frac{\xi+1}{(1-\eta_1)\xi_2}  \big[ (\xi_j-1/2)^2+4/7\big] >0,$$
which together with  $ \xi_2+1<0$ gives  $f(\xi)<0$.
Finally, we obtain the result  (\ref{re01}), which yields  (\ref{re011}) by scaling $u\to u \tau^{-1/3}, v\to v \tau^{-1/3}$.

\end{proof}

    The jump matrix decays to the identity matrix outside $\mathcal{U}_{z}(-1)$ exponentially and uniformly fast as $t \to  +\infty$,
which  enlightens us to construct the solution of $M^{rhp}(z)$ as follows:
 \begin{align} \label{trans6}
 M^{rhp}(z) = \begin{cases}
  E(z),\quad z \in \mathbb{C} \setminus \mathcal{U}_{z}(-1),\\
  E(z) M^{loc}(z), \quad z \in \mathcal{U}_{z}(-1),
 \end{cases}
 \end{align}
where $E(z)$ is an error function which will be determined later, and
 $M^{loc}(z)$  is a solution to the following RH problem.
\begin{prob}
    Find  $ M^{loc}(z)=M^{loc}(z;x,t)$ with properties
	  \begin{itemize}
	  	\item  $M^{loc}(z)$ is analytical in $\mathcal{U}_{z}(-1) \backslash  {\Sigma}^{loc}$, where   $\Sigma^{loc}= \Sigma^{(4)} \cap \mathcal{U}_{z}(-1) $.
     See Figure \ref{scalingzk} (a).
	  	\item  $M^{loc}(z)$  satisfies the jump condition
\begin{equation*}
	  		M^{loc}_+(z)=M^{loc}_-(z) {V}^{loc}(z), \ z\in  {\Sigma}^{loc},
\end{equation*}
\end{itemize}
where
\begin{align*}
        V^{loc}(z)= \begin{cases}
       e^{-it\theta(z)\widehat\sigma_3  }   \left( \begin{array}{cc}
       		1& 0\\
       		 r(\xi_j) T(\xi_j)^2   & 1
       	\end{array}\right),\  z\in \Sigma_{j}, \ j=1,2, \\
  e^{-it\theta(z)\widehat\sigma_3  }  \left(	\begin{array}{cc}
       			1& -\overline{r(\xi_j)} T(\xi_j)^{-2}\\
       			0 & 1
       		\end{array}\right) ,\ z \in \overline{\Sigma}_{j}, \ j=1,2,\\
 T(z)^{-\sigma_3} V(z)T(z)^{\sigma_3},  \   z\in (\xi_2,  \xi_1).
       	\end{cases}
       \end{align*}
	  \begin{itemize}
         \item   $M^{loc}(z)(M^{\infty}((z+1)\tau^{1/3}))^{-1} \to I, \ t\to +\infty$, uniformly for $z \in \partial \mathcal{U}_{z}(-1).$
	  \end{itemize}
\end{prob}

Denote
\begin{align*}
R(z): =r(z)T^2(z),
\end{align*}
then $R(-1) =r(-1)T^{2}(-1)$.
We show that  in the $\mathcal{U}_{k}(0)$,   $M^{loc}(z)$ can be approximated by the solution $M^\infty(k)$ of
 the model RH problem \ref{2minfty} with $r_0= R(-1)$ based on  the following estimates.

 \begin{proposition} \label{ppe}  Let $r\in H^1(\mathbb{R})$, $(x,t)\in \mathcal{P}_{<-1}(x,t)$,  then
\begin{align}
 & \Big| R\left(z\right)   e^{2it\theta \left(z\right)}- r_0    e^{8ik^3/3+2isk }  \Big|   \lesssim t^{-1/6}, \ \ \ k   \in \left(k_2,  k_1\right), \label{e1} \\
 &  \Big| R \left(\xi_j \right)   e^{2it\theta \left(z\right)}- r_0 e^{8ik^3/3+2isk }  \Big|   \lesssim t^{-1/6},    \ k \in  \Sigma_{j}'\cup \overline{\Sigma}_{j}', \ j=1,2.\label{e2}
\end{align}

\end{proposition}
\begin{proof}
For $k \in \left(k_2, k_1\right)$,  then  $z\in (\xi_2, \xi_1)$ and $z$ is real,
$$\left|  e^{2i t   \theta (z) } \right|=1, \ \ \left| e^{  i(\frac{8}{3} k^3+ 2s k)  } \right|=1.$$
  so  we have
\begin{align}
&\left| R (z)    e^{2i t   \theta \left(z\right) } - r_0  e^{  i(\frac{8}{3} k^3+ 2s k)  }    \right| \nonumber\\
&\leq \left|  R (z) - R\left(-1\right) \right|
+ \left|  R\left(-1\right)   \right|       \left|   e^{iS(t;k)}-   1  \right|. \label{poe1}
 \end{align}

Since  there is not discrete spectral in the disk $\mathcal{U}_{z}(-1)$,  $ T(z)$  defined by (\ref{tfcs}) is analytical  in the disk $\mathcal{U}_z(-1)$,  it can be shown  that
  \begin{align}
& \|R (z)\|_{H^1(\xi_2, \xi_1) }  =  \| r(z)\|_{H^1(\xi_2, \xi_1) }. \label{poe232}
 \end{align}
 Noticing that $|k|< \sqrt{2C}$,  with the H\"{o}lder  inequality  and (\ref{poe232}),
  \begin{align}
& \left| R (z) - R (-1)  \right|=     \left|\int_{-1}^{z} R'(s) ds
    \right| \leq \|R' \|_{L^2(\xi_2, \xi_1) } |z+1|^{1/2} \nonumber\\
    &  \leq  \|r \|_{H^1(\xi_2, \xi_1) } |k|^{1/2}     t^{-1/6} \leq c t^{-1/6}. \label{poe2}
 \end{align}
By Proposition \ref{order1},
  \begin{align}
&  \left|   e^{iS(t;k)}-   1  \right| \leq e^{|S(t;k)|}- 1    \leq c t^{-1/3}. \label{poe3}
 \end{align}
 Substituting (\ref{poe2}) and (\ref{poe3}) into (\ref{poe1}) gives the estimate  (\ref{e1}).

For $k \in  {\Sigma}_{1}$,  denote $k = k_1 +u+iv$. By (\ref{re011}), $ \left| e^{i \left( \frac{8}{3} k^3 + 2 s k \right)}\right|$ is bounded.
Similarly to the case on the real axis, we can obtain the estimate   (\ref{e2}). The estimate   on the other jump contours can be given in the same way.
\end{proof}

\begin{corollary}  \label{woee}
 Let $(x,t)\in \mathcal{P}_{<-1}(x,t)$, then for  large $t$, we have
\begin{align}
&V^{loc} \left(z\right) =  V^{\infty}(k)+\mathcal{O}(t^{-1/6}), \ \  k\in\mathcal{U}_{k}(0),\nonumber
\end{align}
\end{corollary}

From the symmetry of $T(z)$ in Proposition \ref{prop1}, we have $T(-1)=1$, and
$ \arg (r(-1) T^2(-1) ) = \arg (r(-1))$.
Further based on above  Corollary \ref{woee},  we can  show
  the following result.

\begin{proposition}\label{locpain}  Let  $(x,t)\in \mathcal{P}_{<-1}(x,t)$, then for  large $t$, we have
\begin{align}
& M^{loc}\left(z\right) = M^{\infty}(k)+\mathcal{O}(t^{-1/6}), \ \  k\in \mathcal{U}_{k}(0),\label{trans7}
\end{align}
where $M^{\infty}(k)$ is given by (\ref{eegs6})  with the  argument
\begin{align}
& \varphi_0=    \arg (r(-1)).\label{trans87}
\end{align}
\end{proposition}
\begin{remark}
In the generic case, we have $ r(-1)=1$ \cite{CJ}  and then $\varphi_0=0$.
\end{remark}

\subsubsection{Small norm RH problem}
\hspace*{\parindent}

We  now consider the error function  $E(z)$ defined by  (\ref{trans6}) and
  have the following  RH problem.

\begin{prob}\label{iew}
  Find      $E(z)$ with the properties
          \begin{itemize}
          	\item   $E(z)$ is analytical in $\mathbb{C}\backslash  \Sigma^{E}$,
          where $\Sigma^E = \partial \mathcal{U}_{z}(-1) \cup \left(\Sigma^{(4)} \backslash \mathcal{U}_{z}(-1)\right)$.
          	\item  $E(z)$  satisfies the jump condition
          	\begin{align*}
          		E_{+}(z)=E_-(z)V^{E}(z), \quad z\in \Sigma^{E},
          	\end{align*}
          	where the jump matrix is given by
           \begin{align}
                V^{E}(z)= \begin{cases}
                   V^{(4)}(z),\quad z \in \Sigma^{(4)} \backslash \mathcal{U}_{z}(-1),\\
                  M^{loc}(z) , \quad z \in \partial \mathcal{U}_{z}(-1). \label{ioei}
                \end{cases}
            \end{align}
See Figure \ref{signdbiiiar}.
          	\item       $E(z)=I+\mathcal{O}(z^{-1}),	\quad  z\to  \infty.$	

          \end{itemize}
\end{prob}

   \begin{figure}
        \begin{center}
  \begin{tikzpicture}[scale=0.7]

                \draw[dotted,-latex](-7,0)--(4,0)node[black,right]{Re$z$};

                \node[below] at (0.3,0) {\footnotesize $0$};

                     \node[below] at (2.8,0) {\footnotesize $ 1$};
                                 \draw [ red] (0.3, 0)--(-1,0.9);
                                 \draw[red](-2.3,0)--(-1,0.9);
                               \draw [ red] (0.3,0 )--(3,-1.87);
                             \draw [-latex,red] (-2.3,0)--(-1.65,0.45);
                                 \draw [ red] (-4.3, 0)--(-7.479, 1.87);
                           \draw [-latex, red] (-7.479, 1.87)--(-5.8895, 0.935);
   \draw [-latex,red] (0.3,0)--(1.65,-0.935 );
   \draw[-latex,red](-1,0.9 )--(-0.35,0.45);

  \draw[red](-1,0)--(-1,0.9);
   \draw[red,-latex](-1,0)--(-1,0.45);
   \draw[blue](-1,-0.9)--(-1,0);
   \draw[blue,-latex](-1,0)--(-1,-0.45);

                             \draw [ blue] (-2.3,0)--(-1, -0.9);
                                \draw [-latex,blue] (-2.3,0)--(-1.65,-0.45);

                                \draw [blue] (-4.3, 0)--(-7.479, -1.87);
                                 \draw [] (-4.3, 0)--(-2.3, 0);
                                  \draw [-latex](-4.3, 0)--(-3.4,0);

                                \draw [blue] (0.3,0 )--(3,1.87);
                                   \draw [blue] (0.3, 0)--(-1,-0.9);
                                 \draw [-latex,blue]  (-7.479, -1.87)--(-5.8895, -0.935);

                                   \draw [-latex,blue] (0.3,0)--(1.65,0.935 );
                                   \draw[-latex,blue](-1,-0.9 )--(-0.35,-0.45);


                \node[shape=circle,fill=blue, scale=0.15]  at (-3,0){0} ;
                \node[shape=circle,fill=blue, scale=0.15]  at (2.8,0){0} ;
            \node[shape=circle,fill=red,scale=0.15] at (-2.3,0) {0};
             \node[shape=circle,fill=red,scale=0.15] at (-4.3,0) {0};
                \node[shape=circle,fill=black,scale=0.1] at (0.3,0) {0};

                                                 \node[below] at (-4.1,0) {\tiny $\xi_2$};
                                         \node[below] at (-2.3,0) {\tiny $\xi_1$};

                                         \node  at (-6.6,1.8) {\tiny  $\Sigma_{2}$};
                                         \node  at (-6.6,-1.8) {\tiny $\overline{\Sigma}_{2}$};

                                            \node  at (2.3,1.8) {\tiny $\Sigma_{3}$};
                                         \node  at (2.3,-1.8) {\tiny $\overline{\Sigma}_{3}$};

                                           \node  at (-1.75,0.8) {\tiny $\Sigma_{1}$};
                                             \node  at (-0.15,0.8) {\tiny $\Sigma_{0}$};

                                             \node  at (-1.75,-0.8) {\tiny  $\overline{\Sigma}$};
                                             \node  at (-0.15,-0.8) {\tiny $\overline{\Sigma}_{0}$};
                                              \node  at (-0.7,0.3) {\tiny $L$};
                                                \node  at (-1.3,-0.3) {\tiny $\overline{L}$};

\node[]  at (-3.3,1.7) {\tiny $\partial \mathcal{U}_{z}(-1)$};
\draw[white!20, fill=white] (-3.3,0) circle (1.3);
          \node[shape=circle,fill=black, scale=0.15]  at (-3.3,0){0} ;
                \node[below] at (-3,0) {\footnotesize $-1$};
                   \draw [ teal](-3.3,0) circle (1.3);
                 \draw [teal,  ->  ]  (-2, 0)  to  [out=90,  in=0] (-3.3,  1.3);
  \end{tikzpicture}
            \caption{\footnotesize{The  jump contour $V^E(z)$  for $E(z)$ }}
      \label{signdbiiiar}
        \end{center}
    \end{figure}
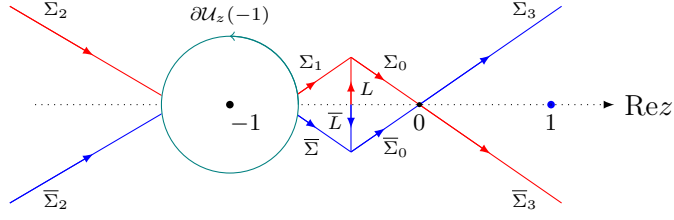

We can  estimate the jump matrix $V^{E}(z)-I$.
\begin{proposition}  Let  $r \in  H^1(\mathbb{R})$. Then
   \begin{equation}\label{vee}
                |V^{E}(z)-I|=\begin{cases}
                    \mathcal{O}(e^{-c t}),\quad z\in \Sigma^{E} \backslash \mathcal{U}_{z}(-1), \\
                    \mathcal{O}(t^{-1/3}), \quad z\in \partial \mathcal{U}_{z}(-1).
                \end{cases}
            \end{equation}
\end{proposition}
\begin{proof}
For $z\in \Sigma^{E} /\mathcal{U}_{z}(-1)$, by (\ref{ioei}) and Proposition \ref{reprop1},
\begin{equation}
|V^{E}(z)-I| = |V^{(4)}(z)-I| \lesssim   \mathcal{O}(e^{-c t}).
\end{equation}
For $z \in \partial\mathcal{U}_{z}(-1)$, by (\ref{ioei}),
\begin{equation}
|V^{E}(z)-I| = |M^{loc}(z)-I| \lesssim   \mathcal{O}(t^{-1/3}).
\end{equation}

\end{proof}

     Define the Cauchy integral operator
        \begin{align*}
           C_{w^{E}}f =C_-\left( f \left( V^{E}(z)-I \right) \right),
        \end{align*}
         where $w^{E} =V^{E}(z)-I$ and  $C_-$ is the Cauchy projection operator on $\Sigma^{E}$.
        By (\ref{vee}), a simple calculation shows that
       $$||C_{w^{E}}||_{L^2(\Sigma^{E})}\lesssim  ||C_-||_{L^2(\Sigma^{E})}||V^{E}(z)-I||_{L^\infty(\Sigma^{E})} \lesssim \mathcal{O}(t^{-1/3}).$$
 According to the theorem of Beals-Coifman,   the solution of the  RH problem \ref{iew} can be expressed by
         \begin{equation}
            E(z)=I +\frac{1}{2\pi i} \int_{\Sigma^{E}} \frac{\mu_E(\zeta) \left( V^{E}(\zeta)-I\right)}{\zeta-z}\, \mathrm{d}\zeta, \nonumber
         \end{equation}
        where $\mu_E \in L^2\left( \Sigma^{E}\right)$   satisfies $\left(I-C_{w^{E}}\right)\mu_E=I$.
   Further, from (\ref{vee}), we have the estimates
       \begin{equation}
        ||V^{E}(z)-I||_{L^2}=\mathcal{O}(t^{-1/3}),\quad
          ||\mu_E-I||_{L^2}=\mathcal{O}(t^{-1/3}), \label{owej}
       \end{equation}
which imply that the  RH problem \ref{iew} exists an unique solution.
 We  make the expansion of  $E(z)$ at $z=\infty$
        \begin{equation}
            E(z)=I +\frac{E_1}{z} +\mathcal{O}\left(z^{-1}\right), \label{trans8}
        \end{equation}
      where
      $$ E_1=-\frac{1}{2\pi i} \int_{\Sigma^{E}} \mu_E(\zeta)\left( V^{E}(\zeta)-I \right)\, \mathrm{d}\zeta.$$

   \begin{proposition} \label{error1}
   $E_1$ and $E(0)$ can be estimated as follows
	\begin{align}
		&  E_1=    -\tau^{-1/3} M_1^\infty(s)+  \mathcal{O}(t^{-2/3}), \label{e001}\\[6pt]
		& E(0) = I-\tau^{-1/3} M_1^\infty(s)+ \mathcal{O}(t^{-2/3}),\label{e00}
	\end{align}
where $M_1^\infty(s)$ is given by (\ref{eegs}) with  the argument  (\ref{trans87}).
\end{proposition}

\begin{proof}
By   (\ref{ioei}) and  (\ref{owej}), we obtain that
	\begin{align*}
		&E_1  = -\frac{1}{2\pi i} \oint_{\partial \mathcal{U}_{z}(-1)} \left( V^{E}(\zeta)-I \right)\, \mathrm{d}\zeta -\frac{1}{2\pi i} \int_{\Sigma^{E} \backslash \partial \mathcal{U}_{z}(-1)} \left( V^{E}(\zeta)-I \right)\, \mathrm{d}\zeta\\
		&- \frac{1}{2\pi i} \int_{\Sigma^{E}}\left( \mu_E(s)-I\right)\left( V^{E}(\zeta)-I \right)\, \mathrm{d}\zeta\\
		& = -\frac{1}{2\pi i} \oint_{\partial \mathcal{U}_{z}(-1)} \left( V^{E}(\zeta)-I \right)\, \mathrm{d}\zeta + \mathcal{O}(t^{-2/3})\\
		& = -\tau^{-1/3}M^{loc}_1(s)+ \mathcal{O}(t^{-2/3}),
	\end{align*}
which gives (\ref{e001}) by  the estimate (\ref{trans7}).

	In a similar way, we have
	\begin{align*}
		&E(0) = I + \frac{1}{2\pi i} \oint_{\partial \mathcal{U}_{z}(-1)} \frac{ V^{E}(\zeta)-I}{\zeta} \, \mathrm{d}\zeta + \mathcal{O}(t^{-2/3})\\
		&= I -\tau^{-1/3}M^{loc}_1(s) + \mathcal{O}(t^{-2/3}),
	\end{align*}
which yields   (\ref{e00}) by  the estimate (\ref{trans7}).
\end{proof}

\subsection{Long-time analysis on a pure $\bar{\partial}$-problem}\label{modi4}
\hspace*{\parindent}

Here we   consider the long-time   asymptotic behavior for  the pure $\bar{\partial}$-problem $M^{(5)}(z)$.
   Define the function
            \begin{equation}\label{transd}
                M^{(5)}(z)=M^{(4)}(z)\left(M^{rhp}(z)\right)^{-1},
            \end{equation}
            which  satisfies the following $\bar{\partial}$-problem.

\begin{prob3}\label{trad}
 Find  $M^{(5)}(z)$ which satisfies
            \begin{itemize}
                \item   $M^{(5)}(z)$ is continuous and has sectionally continuous first partial derivatives in $\mathbb{C}  \backslash \left(\mathbb{R} \cup \Sigma^{(4)}\right)$.
                \item    $M^{(5)}(z)=I+\mathcal{O}(z^{-1}), \ z\to  \infty$.	

                \item For $z\in \mathbb{C}$,  $M^{(5)}(z)$ satisfies the $ \bar{\partial}$-equation
                \begin{equation}
                    \bar{\partial} M^{(5)}(z) = M^{(5)}(z) W^{(5)}(z), \nonumber
                \end{equation}
                where
	\begin{align}
	W^{(5)}(z):=M^{rhp}(z)  \bar{\partial} R^{(3)}(z) \left(M^{rhp}(z)\right)^{-1}, \label{633s}
	\end{align}
 and  $ \bar{\partial}  R^{(3)}(z)$ has been given in (\ref{R3}).
            \end{itemize}

\end{prob3}

The solution of $ \bar{\partial}$-RH problem   \ref{trad} can be given by
\begin{equation} \label{Im3}
	M^{(5)}(z)=I-\frac{1}{\pi}  \iint_\mathbb{C} \frac{ M^{(5)}(\zeta) W^{(5)}(\zeta)}{\zeta-z} \, \mathrm{d}\zeta\wedge \mathrm{d} \bar \zeta,
\end{equation}
which can be  written as an operator equation
\begin{equation}
	(I-S) M^{(5)}(z)=I, \label{Sm3}
\end{equation}
where
\begin{equation}
	Sf(z)=\frac{1}{\pi} \iint_\mathbb{C} \frac{f(\zeta)W^{(5)}(\zeta)}{\zeta-z} \, \mathrm{d}\zeta\wedge \mathrm{d} \bar \zeta. \label{hfu}
\end{equation}

\begin{proposition}\label{pss} Let  $q_0 \in \tanh (x)+H^{4,4}(\mathbb{R})$. Then
	the    operator $S$  admits the estimate
	\begin{align}
&	\parallel S\parallel_{L^\infty\to L^\infty}\lesssim  t^{-1/6}, \label{esuf}
	\end{align}
	which implies the existence of $(I-S)^{-1}$ for large $t$.
\end{proposition}
\begin{proof}
	The estimate  of the operator $S$  on  $\Omega_{3}  $ and  $\overline{\Omega}_{3}  $
 used (\ref{438}), which was given by Cuccagna and Jenkins \cite{CJ}.   We   estimate   the operator $S$  on  $\Omega_{2}$  and other cases  are   similar. In fact, by  (\ref{437}),  (\ref{R3}),  (\ref{633s}) and (\ref{hfu}), we have
	\begin{align*}
	\|Sf\|_{L^\infty\to L^\infty}  \leq c (I_1+I_2+I_3+I_4),
	\end{align*}
where
	\begin{align*}
		&I_1  =    \iint_{\Omega_{2} \cap \{|z| \le 2\} } F(\zeta,z) \mathrm{d}\zeta\wedge \mathrm{d} \bar \zeta, \ \  I_2 =  \iint_{\Omega_{2} \cap \{|z| > 2\}}  F(\zeta,z) \mathrm{d}\zeta\wedge \mathrm{d} \bar \zeta,\\
&I_3  =    \iint_{\Omega_{2} \cap \{|z| \le 2\} } G(\zeta,z)\mathrm{d}\zeta\wedge \mathrm{d} \bar \zeta, \  \ \ I_4=\iint_{\Omega_{2} \cap \{|z| > 2\}} G(\zeta,z)	\mathrm{d}\zeta\wedge \mathrm{d} \bar \zeta.
	\end{align*}
with
$$  F(\zeta,z)= \frac{1}{|\zeta-z|}\left|r'(\re(\zeta)) \right|e^{\re(2it\theta)}, \  G(\zeta,z)= \frac{1}{|\zeta-z|}\left| \zeta-\xi_2 \right|^{-1/2}e^{\re(2it\theta)}.   $$

Let $z=x+iy$ and $\zeta =\xi_2 + u+iv = |\zeta|e^{iw}$.
   Using Proposition \ref{reprop1} and  the Cauchy-Schwartz's inequality, we have
	\begin{align*}
		 I_1  &= \int_{0}^{2 \sin w} \int_{-\xi_2-2\cos w}^{-v} F(\zeta,z) \mathrm{d}u \mathrm{d}v \\
		&	\lesssim \int_{0}^{2 \sin w} ||r'||_{L^2}   |v-y|^{-1/2} e^{-\frac{1}{8|\xi_2|^3}tv^3} \mathrm{d}v   \lesssim t^{-1/6},\\
         I_2  &= \int_{2 \sin w}^\infty \int_{-\infty}^{-\xi_2-2\cos w} F(\zeta,z) \mathrm{d}u \mathrm{d}v \\
		&	\lesssim \int_{2 \sin w}^\infty ||r'||_{L^2}   |v-y|^{-1/2} e^{-2\sqrt{2}tv} \mathrm{d}v   \lesssim t^{-1/2}.
	\end{align*}

In a similar way, using Proposition \ref{reprop1} and   the H\"{o}lder's inequality with $p>2$ and $1/p+1/q=1$, we obtain
	\begin{align*}
		 I_3  &\lesssim  \int_{0}^{2 \sin w} v^{1/p-1/2}|v-y|^{1/q-1}e^{-\frac{1}{8|\xi_2|^3}tv^3} \mathrm{d}v \lesssim t^{-1/6},\\
       I_4  &\lesssim  \int_{2 \sin w}^\infty v^{1/p-1/2}|v-y|^{1/q-1}e^{-2\sqrt{2}tv} \mathrm{d}v \lesssim t^{-1/2}.
	\end{align*}
\end{proof}

This Proposition \ref{pss}  implies that   the operator  equation (\ref{Sm3})  exists an unique solution, which
can be expanded in the form
\begin{equation}
	M^{(5)}(z)=I + \frac{M^{(5)}_1(x,t) }{z} +\mathcal{O}(z^{-2}), \quad z\to \infty, \label{trans9}
\end{equation}
where
\begin{align}\label{expanm51}
M^{(5)}_1(x,t)=\frac{1}{\pi} \iint_{\mathbb{C}} M^{(5)}(\zeta)W^{(5)}(\zeta)\, \mathrm{d}\zeta\wedge \mathrm{d} \bar \zeta.
\end{align}

Take $z=0$ in (\ref{Im3}), then
\begin{equation} \label{m50}
	M^{(5)}(0)=I-\frac{1}{\pi}  \iint_\mathbb{C} \frac{ M^{(5)}(\zeta) W^{(5)}(\zeta)}{\zeta}  \mathrm{d}\zeta\wedge \mathrm{d} \bar \zeta.
\end{equation}
\begin{proposition}  Let  $q_0 \in \tanh (x)+H^{4,4}(\mathbb{R})$. We have the  following estimates
	\begin{equation}\label{m51infty}
		|M^{(5)}_1(x,t)| \lesssim t^{-1/2}, \ \ \ |M^{(5)}(0)-I| \lesssim t^{-1/2}.
	\end{equation}
\end{proposition}
\begin{proof}
	Similarly to the proof of Proposition \ref{pss}, we take $z \in \Omega_{2}$ as an example and divide the integration (\ref{expanm51}) on $\Omega_{2}$ into four parts.
Firstly, we consider the estimate of  $M^{(5)}_1(x,t)$.
By (\ref{transd}) and the boundedness of $M^{(4)}(z)$ and $M^{rhp}(z)$ on $\Omega_{2}$, we have
	\begin{align}\label{m51esti}
	|M^{(5)}_1(x,t)| \lesssim I_1+I_2+I_3+I_4,
	\end{align}
	where
	\begin{align*}
		&I_1  =    \iint_{\Omega_{2} \cap \{|z| \le 2\} } F(\zeta,z) \mathrm{d}\zeta\wedge \mathrm{d} \bar \zeta, \ \  I_2 =  \iint_{\Omega_{2} \cap \{|z| > 2\}}  F(\zeta,z) \mathrm{d}\zeta\wedge \mathrm{d} \bar \zeta,\\
&I_3  =    \iint_{\Omega_{2} \cap \{|z| \le 2\} } G(\zeta,z)\mathrm{d}\zeta\wedge \mathrm{d} \bar \zeta, \  \ \ I_4=\iint_{\Omega_{2} \cap \{|z| > 2\}} G(\zeta,z)	\mathrm{d}\zeta\wedge \mathrm{d} \bar \zeta.
	\end{align*}
with
$$  F(\zeta,z)= \left|r'(\re(\zeta)) \right|e^{\re(2it\theta)}, \  G(\zeta,z)= \left| \zeta-\xi_2 \right|^{-1/2}e^{\re(2it\theta)}.   $$
 Let  $\zeta =\xi_2 + u+iv = |\zeta|e^{iw}$.  By Cauchy-Schwartz's inequality and Proposition \ref{reprop1}, we have
	\begin{align*}
		I_1 & \lesssim \int_{0}^{2 \sin w} \int_{-\xi_2-2\cos w}^{-v} \left|r'(\re(\zeta)) \right|e^{-\frac{1}{8|\xi_2|^3}tv^3} \mathrm{d}u \mathrm{d}v  \lesssim t^{-1/2},\\
I_2 & \lesssim \int_{2 \sin w}^\infty \int_{-\infty}^{-\xi_2-2\cos w} \left|r'(\re(\zeta)) \right| e^{-2\sqrt{2}tv} \mathrm{d}u \mathrm{d}v  \lesssim t^{-3/2}.
	\end{align*}
	By H\"{o}lder's inequality with $p>2$ and $1/p+1/q=1$ and Proposition \ref{reprop1}, we have
	\begin{align*}
		I_3 &\lesssim \int_{0}^{2 \sin w} \int_{-\xi_2-2\cos w}^{-v}   \left| u + i v\right|^{-1/2} e^{-\frac{1}{8|\xi_2|^3}tv^3} \mathrm{d}u  \mathrm{d}v \lesssim t^{-1/2},\\
I_4 &\lesssim \int_{2 \sin w}^\infty \int_{-\infty}^{-\xi_2-2\cos w}  \left| u + i v\right|^{-1/2}e^{-2\sqrt{2}tv} \mathrm{d}u  \mathrm{d}v \lesssim t^{-3/2}.
	\end{align*}

Then we estimate $M^{(5)}(0)-I$. Notice that if $z \in \Omega_{2}$, then $|z|^{-1} \le |\xi_2|^{-1}$. By (\ref{m50}) and (\ref{m51esti}), we get $|M^{(5)}(0)-I| \lesssim t^{-1/2}$.

\end{proof}

 To recover the potential from the reconstruction formula (\ref{sol}), we need the following estimate.
\begin{proposition}$M^{(3)}(0)$ satisfies the   estimate
	\begin{align}\label{m30}
		M^{(3)}(0)= E(0) + \mathcal{O}(t^{-1/2}),
	\end{align}
where $E(0)$ is given by (\ref{e00}).
\end{proposition}
\begin{proof}
	Reviewing the series of transformations  (\ref{trans4}), (\ref{trans5}), and (\ref{trans6}),  for $z$ large and satisfying   $R^{(3)}(z)=I$,  the solution of $M^{(3)}(z)$ is given by
	\begin{align*}
		M^{(3)}(z) = M^{(5)}(z)E(z) .
	\end{align*}
By (\ref{trans9}) and (\ref{m51infty}), we further obtain
\begin{align*}
	M^{(3)}(z)= E(z)  + \mathcal{O} (t^{-1/2}),
\end{align*}
which yields (\ref{m30}) by taking $z=0$.
\end{proof}

\subsection{Proof of Theorem \ref{th}---Case I } \label{thm-result1}
\hspace*{\parindent}


\begin{proof}
	
	Inverting  the sequence  of transformations (\ref{trans1}), (\ref{trans2}), (\ref{trans3}), (\ref{trans4}), (\ref{trans5}), (\ref{trans6}),
	and especially taking $z\to \infty$ vertically such that  $R^{(3)}(z)= G(z)=I$, then  the  solution of RH problem \ref{RHP0} is given by
	\begin{align}
		& 	M(z) = T(\infty)^{\sigma_3} (  I + z^{-1}\sigma_1 M^{(3)}(0)^{-1} )  M^{(5)}(z) E(z)T(z)^{-\sigma_3} + \mathcal{O}(e^{-ct}). \nonumber
	\end{align}
	Noticing that (\ref{Texpan1}) and  (\ref{m30}), further substituting asymptotic expansions  (\ref{Texpan}),  (\ref{trans8}),   (\ref{trans9}) into  the above formula,
	the reconstruction formula (\ref{sol}) yields
	\begin{align*}
		&q(x,t) 
		=T(\infty )^2 \left(  1-   \frac{1}{2} i \tau^{-\frac{1}{3}}  \left(u(s) e^{i\varphi_0}+\int_s^\infty u^2(\zeta)  \mathrm{d}\zeta   \right) \right) + \mathcal{O}\left( t^{-\frac{1}{2}}\right),
	\end{align*}
which leads to the result (\ref{q1}) in Theorem \ref{th}.

\end{proof}

\section{Painlev\'e asymptotics in  transition   region  $\mathcal{P}_{+1}(x,t)$} \label{sec4}

In a way similar to Section \ref{sec3},     we study  the Painlev\'e asymptotics    in the region
 $(x,t)\in \mathcal{P}_{+1}(x,t) $ and consider the region   $ 0< (\xi-1) t^{2/3}<C$, which corresponds to  Figure \ref{figuref}. For brevity, we denote
$$   \mathcal{P}_{> +1}(x,t)=\{ (x,t): 0< (\xi-1) t^{2/3}<C\}.  $$
  In this case,
the two stationary points $\xi_1, \xi_2$ defined by (\ref{xi2}) are real and close to $z= 1$ at least the speed of $t^{-1/3}$ as $t\to +\infty$.

\subsection{Modifications to the basic RH problem}
\hspace*{\parindent}


Opening    the contour  $(0, \infty)$   needs the  second  matrix   decomposition in  (\ref{v}),
    so  we introduce the function
         \begin{equation}
        T(z)=  \exp \left(  - i \int_{0}^{\infty} \nu(\zeta) \left(\frac{1}{\zeta-z}- \frac{1}{2\zeta} \right) \, \mathrm{d}\zeta \right),\label{resed}
        \end{equation}
     and obtain   the following proposition.

        \begin{lemma} \label{lemma21} \cite{WF}
            The function $T(z)$ has the following properties
            \begin{itemize} \label{prop2}
                \item $T(z)$ is  analytical  in $\mathbb{C} \backslash [0,\infty)$.
                \item $T(z)$ has the symmetry $\overline{T(\bar{z})}=T(z)^{-1}=T(z^{-1})$.
                \item $T(z)$ satisfies the jump condition \begin{equation*}
                    T_+(z)=T_-(z)(1-|r(z)|^2), \quad z\in (0,\infty).
                \end{equation*}
                \item Let \begin{equation}\label{2Tinfty}
                    T(\infty) := \lim_{z \to \infty} T(z)=  \exp \left(i \int_{0}^{\infty} \frac{\nu(\zeta)}{2\zeta} \, \mathrm{d}s \right).
                \end{equation}
                Then,  $|T(\infty)|=1$ and the asymptotic expansion at infinity is
                \begin{equation}\label{2Tzinfty}
                   T(z)=T(\infty)\Bigg[ 1+\frac{i}{z}\Bigg(\int_{0}^{\infty}\nu(\zeta)d\zeta\Bigg)+\mathcal{O}(z^{-2})\Bigg].
                \end{equation}
            \item The ratio $\frac{s_{11}(z)}{T(z)}$ here  has the same boundedness with that ratio in Lemma  \ref{prop1}.


            \end{itemize}
        \end{lemma}

    To  remove poles  on the unitary  circle $|z|=1$ by converting   their  residues   into   jumps,
    we fix a constant $\rho$  defined by (\ref{definerho}) and introduce  the  interpolation function
       \begin{equation*}
        G(z) = \begin{cases}
            \left(\begin{array}{cc} 1&0\\ -\displaystyle {\frac{ c_j e^{2it\theta(z_j)} }{ z-z_j }} &1\end{array}  \right), \;  |z-z_j|<\rho,  \\[4pt]
      \left(\begin{array}{cc} 1& -\displaystyle { \frac{ \bar{c}_j e^{-2it\theta(\bar{z}_j) }}{ z-\bar{z}_j}} \\ 0&1\end{array}  \right), \;   |z-\bar{z}_j|<\rho,\\
         \left(\begin{array}{cc} 1&0 \\ 0&1\end{array}  \right), \;  \; \; \text{otherwise},
        \end{cases}
    \end{equation*}
  where $z_j \in \mathcal{Z}^+$ and $\bar{z}_j \in \mathcal{Z}^-$.

 Define a contour
\begin{equation}
	\Sigma^{(1)}=\mathbb{R} \cup\left[\bigcup_{j=0 }^{N-1}  \{ z\in \mathbb{C}: |z-z_j|=\rho,  or \   |z-\bar z_j|=\rho\} \right], \nonumber
\end{equation}
and make the following transformation
       \begin{equation}\label{2trans1}
        M^{(1)}(z)=T(\infty)^{-\sigma_3} M(z) G(z)T(z)^{\sigma_3},
       \end{equation}
 then $M^{(1)}(z)$ satisfies the RH problem as follows.

  \begin{prob}\label{m12}
 Find  $M^{(1)}(z)=M^{(1)}(z;x,t)$  satisfying
       \begin{itemize}
        \item $M^{(1)}(z)$ is  analytical  in $\mathbb{C} \backslash \Sigma^{(1)}$.
        \item $M^{(1)}(z)=\sigma_1 \overline{M^{(1)}(\bar{z})}\sigma_1 =z^{-1}M^{(1)}(z^{-1})\sigma_1$.
        \item $M^{(1)}(z)$ satisfies the jump condition
        \begin{equation*}
            M^{(1)}_+(z)=M^{(1)}_-(z)V^{(1)}(z),
        \end{equation*}
        where
\begin{equation}\nonumber
	V^{(1)}(z)=\left\{\begin{array}{ll}
 B_-^{-1}B_+,  \quad  z\in (0,  \xi_1)\cup (\xi_2,\infty),\\[8pt]
 T(z)^{-\sigma_3} V(z)T(z)^{\sigma_3},   \quad  z\in (\xi_1,  \xi_2),\\[8pt]
 b_-^{-1}b_+,  \quad z\in (-\infty,  0),\\[6pt]
\left(\begin{array}{cc}
		1 & 0\\
		-\frac{c_j} {  z-z_j} T^2(z) e^{ 2it\theta( z_j)} & 1
		\end{array}\right),   \quad	|z-z_j|=\rho, \\[12pt]
		\left(\begin{array}{cc}
			1 & \frac{\bar c_j} {  z-\bar z_j} T^{-2} (z) e^{-2it\theta(\bar z_j)}\\
			0 & 1
		\end{array}\right), \quad |z-\bar z_j|=\rho,
\end{array}\right.
\end{equation}
where $b_\pm,$ and $  B_\pm$   are given by (\ref{opep1}) and (\ref{opep2}) respectively.

        \item $M^{(1)}(z)$ has the asymptotic behaviors
        \begin{align*}
                &M^{(1)}(z)=I+\mathcal{O}(z^{-1}),	\quad  z \to  \infty,\\
                &zM^{(1)}(z)=\sigma_1+\mathcal{O}(z), \quad z \to 0.
        \end{align*}

    \end{itemize}
\end{prob}
Noticing that $ V^{(1)}(z) \to I, \ t\to +\infty $ on the circles  $|z-  z_j|=\rho$ and  $|z-\bar z_j|=\rho$,
 so we  get rid of the exponential infinitesimal  term in the jump matrices.
  The    RH problem  \ref{m12}  is asymptotically equivalent to the RH problem \ref{2rhp2m2}
  with the exponential   error $\mathcal{O}(e^{- ct})$.
  \begin{align}\label{2trans2}
  	M^{(1)}(z) = M^{(2)}(z)(I+\mathcal{O}(e^{- ct})),
  \end{align}
where $M^{(2)}(z)$ is the solution of the following RH problem.

  \begin{prob}\label{2rhp2m2}
 Find $M^{(2)}(z)=M^{(2)}(z;x,t)$  with properties
       \begin{itemize}
        \item   $M^{(2)}(z)$ is  analytical  in $\mathbb{C}\backslash \mathbb{R}$.
        \item   $M^{(2)}(z)=\sigma_1 \overline{M^{(2)}(\bar{z})}\sigma_1 =z^{-1}M^{(2)}(z^{-1})\sigma_1$.
        \item $M^{(2)}(z)$ satisfies the  jump condition
        \begin{equation*}
            M^{(2)}_+(z)=M^{(2)}_-(z)V^{(2)}(z),
        \end{equation*}
        where
\begin{equation}\label{2v2}
	V^{(2)}(z)=\left\{\begin{array}{ll}
 B_-^{-1}B_+,  \quad  z\in (0,  \xi_1)\cup (\xi_2,\infty),\\[8pt]
T(z)^{-\sigma_3} V(z)T(z)^{\sigma_3},   \quad  z\in (\xi_1,  \xi_2),\\[8pt]
 b_-^{-1}b_+,  \quad z\in (-\infty,  0).
\end{array}\right.
\end{equation}
        \item $M^{(2)}(z)$ admits the asymptotic behaviors
        \begin{align*}
                &M^{(2)}(z)=I+\mathcal{O}(z^{-1}),	\quad  z \to  \infty,\\
                &zM^{(2)}(z)=\sigma_1+\mathcal{O}(z), \quad z \to 0.
        \end{align*}

    \end{itemize}
\end{prob}

Similarly to the case in Section \ref{remove2}, we remove the singularity by the following transformation
\begin{align}\label{2trans3}
M^{(2)}(z)=\left( I+ \frac{1}{z} \sigma_1 M^{(3)}(0)^{-1} \right )M^{(3)}(z),
\end{align}
where $M^{(3)}(z)$ satisfies the following  RH problem.
\begin{prob} \label{mse3}
 Find  $M^{(3)}(z)=M^{(3)}(z;x,t)$ which satisfies
       \begin{itemize}
        \item $M^{(3)}(z)$ is  analytical  in $\mathbb{C}\backslash \mathbb{R}$.
        \item $M^{(3)}(z)=\sigma_1 \overline{M^{(3)}(\bar{z})}\sigma_1 =\sigma_1 M^{(3)}(0)^{-1}M^{(3)}(z^{-1})\sigma_1$.
        \item $M^{(3)}(z)$ satisfies the  jump condition
        \begin{equation*}
            M^{(3)}_+(z)=M^{(3)}_-(z)V^{(2)}(z),
        \end{equation*}
        where
        $V^{(2)}(z)$ has been given in  (\ref{2v2}).
        \item
        $  M^{(3)}(z)=I+\mathcal{O}(z^{-1}),	\quad  z \to  \infty.$

    \end{itemize}
\end{prob}

\subsection{Transformation to a hybrid $\bar{\partial}$-RH problem}
\hspace*{\parindent}

We  define the contours and sectors. See Figure \ref{sign21}. We especially point out the properties of the real part of the phase function.

\begin{proposition}\label{reprop2}
  Let $\xi>1$ and $(x,t) \in \mathcal{P}_{>+1}(x,t)$. The real part of the phase function $\re \left(2i\theta(z)\right)$ has the following estimations.
 \begin{itemize}
   \item[{\rm   Case I.}](corresponding to $z=0$)
    \begin{align}
       & \re(2i\theta(z))\geq |\sin 2\varphi  | |v|,\quad z\in \Omega_{0}\cup \overline{\Omega}_{3},\\
       & \re(2i\theta(z))\leq - |\sin 2\varphi | |v|,\quad z\in  \overline{\Omega}_{0} \cup \Omega_{3},
    \end{align}
    where $z= u+iv$ and $\varphi =\arg z$.
      \item[{\rm   Case II.}](corresponding to $z=\xi_1$)
    \begin{align}
       & \re(2i\theta(z))\geq \frac{4}{|\xi_1|}u^2 |v|,\quad z\in \Omega_{1},\\
       & \re(2i\theta(z))\leq -\frac{4}{|\xi_1|}u^2 |v|,\quad z\in  \overline{\Omega}_{1},
    \end{align}
    where $z= \xi_1+u+iv$.
    \item[{\rm   Case III.}] (corresponding to $z=\xi_2$)
      \begin{align}
        & \re(2i\theta(z))\geq \begin{cases}
        \frac{1}{8|\xi_2|^3}u^2 |v|,\quad z\in \Omega_{2} \cap\{ |z| \leq 2\},\\[4pt]
       	  2 \sqrt{2}|v|,\quad z\in \Omega_{2} \cap  \{  |z|>2 \},	
       	\end{cases}\\
       & \re(2i\theta(z))\leq
       \begin{cases}
        -\frac{1}{8|\xi_2|^3}u^2 |v|,\quad z\in \overline{\Omega}_{2} \cap\{ |z| \leq 2\},\\[4pt]
       	  -2 \sqrt{2}|v|,\quad z\in \overline{\Omega}_{2} \cap  \{  |z|>2 \},	
       	\end{cases}
       \end{align}
    where $z=\xi_2+ u+iv$.
 \end{itemize}

  \end{proposition}

\begin{proof}
The proof here is similar to the proof on the region $\mathcal{P}_{< -1}(x,t)$, see more details in Proposition   \ref{reprop1}.
\end{proof}

  To open the dashed part  of the real axis $\mathbb{R}$ in Figure \ref{sign21},
we introduce the following extension  functions $R_j(z), \ j=0, 1, 2, 3$.

   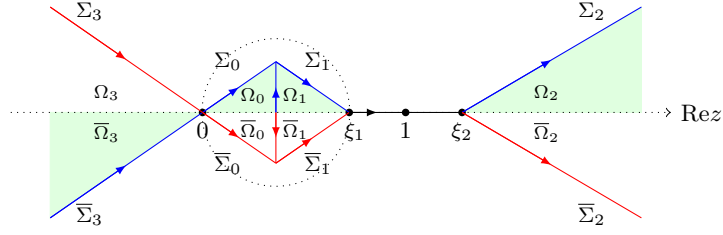
\begin{figure}
	\begin{center}
		\begin{tikzpicture}[scale=0.75]
     \draw[green!12, fill=green!12] (4.3, 0)--(7.479,  1.87)--(7.479,0)--(4.3, 0);
   \draw[green!12, fill=green!12] (-0.3, 0)--(-3,-1.87)--(-3,0)--(-0.3, 0);
    \draw[green!12, fill=green!12] (-0.3, 0)--(2.3,0)--(1, 0.9)--(-0.3, 0);

			\draw[dotted, ->](-3.7,0)--(8,0)node[black,right]{\footnotesize Re$z$};
			\node[shape=circle,fill=black, scale=0.15]  at (3.3,0){0} ;
			\node[shape=circle,fill=black,scale=0.15] at (2.3,0) {0};
			\node[shape=circle,fill=black,scale=0.15] at (4.3,0) {0};
			\node[shape=circle,fill=black,scale=0.15] at (-0.3,0) {0};
			\node[below] at (-0.3,0) {\footnotesize $0$};
			\node[below] at (3.3,0) {\footnotesize $1$};
			\draw [blue] (-0.3, 0)--(1,0.9);
			\draw[blue](2.3,0)--(1,0.9);
			\draw [blue] (-0.3,0 )--(-3,-1.87);
			\draw [-latex,blue] (1,0.9)--(1.65,0.45);
			
			\draw [blue] (4.3, 0)--(7.479, 1.87);
			\draw [-latex, blue] (4.3, 0)--(5.8895, 0.935);
			
			\draw [-latex,blue] (-3,-1.87)--(-1.65,-0.935 );
			\draw[-latex,blue](-0.3,0 )--(0.35,0.45);

			\draw[blue](1, 0.9)--(1,0);
            \draw[red](1,-0.9)--(1,0);
			\draw[blue,-latex](1,0)--(1,0.45);
			\draw[red,-latex](1,0)--(1,-0.45);
			
			\draw [red] (2.3,0)--(1, -0.9);
			\draw [-latex,red] (1, -0.9)--(1.65,-0.45);
			
			\draw [red] (4.3, 0)--(7.479, -1.87);
			\draw [] (4.3, 0)--(2.3, 0);
			\draw [-latex](2.3, 0)--(2.8,0);
			
			\draw [red] (-0.3,0 )--(-3,1.87);
			\draw [red] (-0.3, 0)--(1,-0.9);
			\draw [-latex,red]  (4.3, 0)--(5.8895, -0.935);
			
			\draw [-latex,red] (-3,1.87)--(-1.65,0.935 );
			\draw[-latex,red](-0.3, 0)--(0.35,-0.45);
			
	                           \node  at (6.6,1.8) {\scriptsize  $\Sigma_{2}$};
                                         \node  at (6.6,-1.8) {\scriptsize $\overline{\Sigma}_{2}$};

                                            \node  at (-2.3,1.8) {\scriptsize $\Sigma_{3}$};
                                         \node  at (-2.3,-1.8) {\scriptsize $\overline{\Sigma}_{3}$};

                                           \node  at (1.75,0.9) {\scriptsize $\Sigma_{1}$};
                                             \node  at (0.15,0.9) {\scriptsize $\Sigma_{0}$};

                                             \node  at (1.75,-0.9) {\scriptsize  $\overline{\Sigma}_{1}$};
                                             \node  at (0.15,-0.9) {\scriptsize $\overline{\Sigma}_{0}$};
			
			\node[below] at (4.3,0) {\scriptsize $\xi_2$};
			\node[below] at (2.4,0) {\scriptsize $\xi_1$};
			\node  at (5.8,0.35) {\tiny $\Omega_{2}$};
			\node  at (5.8,-0.35) {\tiny $\overline{\Omega}_{2}$};
			\node  at (-2,0.35) {\tiny $\Omega_{3}$};
			\node  at (-2,-0.35) {\tiny $\overline{\Omega}_{3}$};
			\node  at (1.35,0.3) {\tiny $\Omega_{1}$};
			\node  at (0.6,0.3) {\tiny $\Omega_{0}$};
			\node  at (1.35,-0.3) {\tiny $\overline{\Omega}_{1}$};
			\node  at (0.6,-0.3) {\tiny $\overline{\Omega}_{0}$};
\draw [dotted ](1,0) circle (1.31);

		\end{tikzpicture}
	 \caption{\footnotesize{Open the jump contour $\mathbb{R}\setminus (\xi_1, \xi_2)$ along  red rays  and blue rays.
           The green regions are  continuous extension sectors  with
             $\re \left(2i\theta(z)\right)> 0$, while the white regions
             are continuous extension sectors with $\re \left(2i\theta(z) \right)< 0$. }}
		\label{sign21}
	\end{center}
\end{figure}

 \begin{proposition}
     Let  $q_0 \in \tanh (x)+H^{4,4}(\mathbb{R})$.  Define  functions $R_{j}(z) (j=0,1,2,3) $ with the  boundary values
       \begin{align*}
       &R_{j}(z)= \begin{cases}
          \frac{\overline{r  (z)} }{1-|r(z)|^2} T_+^{-2}(z), \quad z \in (0,\xi_1) \cup (\xi_2, \infty),\\[2pt]
             R(\xi_j),\quad z \in \Sigma_{j},\ j=0,1,2,\\
           \end{cases}\\
             &\overline{R}_{j}(z)= \begin{cases}
               \frac{ {r  (z)} }{1-|r(z)|^2} T_+^{-2}(z), \quad z \in (0,\xi_1) \cup (\xi_2, \infty),\\[2pt]
                \overline{R(\xi_j)}, \quad z \in \overline{\Sigma}_{j},\ j=0,1,2,\\
               \end{cases}\\
        &R_{3}(z)= \begin{cases}
             r(z)T^2(z), \quad  z \in (-\infty,0),\\
             0,\quad z \in \Sigma_{3},
            \end{cases}\\
        &\overline{R}_{3}(z)= \begin{cases}
         \overline{r(z)}T^{-2}(z), \quad z  \in (-\infty,0),\\
                  0,\quad z \in \overline{\Sigma}_{3},
               \end{cases}
       \end{align*}
       where
 \begin{equation}
R(\xi_j)= \frac{ \overline{S_{21}(\xi_j)} }{S_{11}(\xi_j)} \left(\frac{ s_{11}(\xi_j) }{T_+(\xi_j)}\right)^2,
\end{equation}
where $S_{21}(z)$ and $S_{11}(z)$ are given by (\ref{oeeedk2}).
 Then  for $j=0,1,2,3$, there exists a constant $c=c(q_0)$, which only depends on $q_0$, such that
       \begin{equation*}
	|\bar{\partial}R_{j}| \le
c\left( |\varphi(\re(z))|+  | r'\left( \re\left( z \right)\right)|+    |z-\xi_j|^{-1/2} \right), \;\text{for} \;z\in \Omega_{j} \cup \overline{\Omega}_{j},\nonumber
\end{equation*}
where $\varphi \in C_0^\infty \left( \mathbb{R}, \left[0,1\right] \right)$ with small support near $z=1$.
        \end{proposition}

 \begin{proof}
 The extension functions $R_j(z), j=0,1,2$ are defined in  the same way  in (\ref{430}).
 The function $R_3(z)$ can be treated by  (\ref{429}).
  The   proof  is  similar to that in  Proposition \ref{prop3}.

\end{proof}

Define the function
    \begin{align}\label{2R3}
    	R^{(3)}(z)=\begin{cases}
      		\left(\begin{array}{cc}
    	1&  R_{j}e^{-2it\theta(z)}\\0&1
    \end{array} \right), \quad z \in \Omega_{j},\quad  j = 0,1,2,\\
          		\left(\begin{array}{cc}
    	1&  0\\\overline{R}_{j}e^{2it\theta(z)}&1
    \end{array} \right), \quad z \in \overline{\Omega}_{j},\quad  j = 0,1,2,\\
    \left( \begin{array}{cc}
    	1& 0\\-R_{3}e^{2it\theta(z)}&1
    \end{array} \right), z \in \Omega_{3},\\
    \left(\begin{array}{cc}
    	1&  -\overline{R}_{3}e^{-2it\theta(z)}\\0&1
    \end{array} \right), \quad z \in \overline{\Omega}_{3},\\
    I, \quad {\rm other \ regions},
    	\end{cases}
    \end{align}
    and the contour
$$\Sigma^{(4)}=\Gamma\cup \overline{\Gamma}\cup [\xi_1, \xi_2], \ \Gamma= \cup_{j=0}^3 \Sigma_j \cup L.$$
 Then we obtain the following function
    \begin{equation}\label{2trans4}
        M^{(4)}(z)=M^{(3)}(z)R^{(3)}(z),
    \end{equation}
satisfies the following mixed $\bar{\partial}$-RH problem.

\begin{prob2}
    Find   $M^{(4)}(z)=M^{(4)}(z;x,t)$   satisfying
    \begin{itemize}
        \item $M^{(4)}(z)$ is continuous in $\mathbb{C}\setminus  \Sigma^{(4)} $. See Figure \ref{sign21}.
        \item $M^{(4)}(z)$ satisfies the  jump condition
        \begin{equation*}
            M^{(4)}_+(z)=M^{(4)}_-(z)V^{(4)}(z),
        \end{equation*}
        where
       \begin{align}\label{2v4}
       	V^{(4)}(z)=\begin{cases}
           \left( \begin{array}{cc}
    	1& -R(\xi_j) e^{-2it\theta(z)}\\0&1
    \end{array} \right), \quad z \in \Sigma_{j},\\
    \left(\begin{array}{cc}
    	1& 0 \\ \overline{R(\xi_j)} e^{2it\theta(z)}&1
    \end{array} \right), \quad z \in \overline{\Sigma}_{j},\\
     	\left(	\begin{array}{cc}
     	1& -R(\xi_j) e^{-2it\theta(z)}\\
     	0& 1
     \end{array}\right),\quad z\in L,\\
            \left(		\begin{array}{cc}
       	1& 0\\
       	\overline{R(\xi_j)} e^{2it\theta(z)} & 1
       \end{array}\right),\quad z\in \overline{L},\\
      T(z)^{-\sigma_3} V(z)T(z)^{\sigma_3},\quad z \in (\xi_1, \xi_2).
       	\end{cases}
       \end{align}

        \item $M^{(4)}(z)$ admits the asymptotic behavior $$M^{(4)}(z)=I+\mathcal{O}(z^{-1}),	\quad  z \to  \infty.$$

        \item For $z\in \mathbb{C}\setminus \Sigma^{(4)} $, $\bar{\partial}M^{(4)}(z)= M^{(4)}(z) \bar{\partial}R^{(3)}(z),$
        where
        \begin{equation}\label{2r3}
            \bar{\partial}R^{(3)}(z)= \begin{cases}
            		\left(\begin{array}{cc}
    	1&  \bar{\partial}R_{j}e^{-2it\theta(z)}\\0&1
    \end{array} \right), \quad z \in \Omega_{j},\quad  j = 0,1,2,\\
          		\left(\begin{array}{cc}
    	1&  0\\\bar{\partial} \overline{R}_{j}e^{2it\theta(z)}&1
    \end{array} \right), \quad z \in \overline{\Omega}_{j},\quad  j = 0,1,2,\\
    \left( \begin{array}{cc}
    	1& 0\\-\bar{\partial}R_{3}e^{2it\theta(z)}&1
    \end{array} \right), z \in \Omega_{3},\\
    \left(\begin{array}{cc}
    	1&  -\bar{\partial}\overline{R}_{3}e^{-2it\theta(z)}\\0&1
    \end{array} \right), \quad z \in \overline{\Omega}_{3}.
            	\end{cases}
           \end{equation}

\end{itemize}
    \end{prob2}

Then  we continue to decompose $M^{(4)}(z)$ into two parts:
\begin{align} \label{2trans5}
M^{(4)}(z) = \begin{cases}
M^{rhp}(z), \; \bar{\partial} R^{(3)}(z)=0,\\
M^{(5)}(z), \; \bar{\partial} R^{(3)} (z)\neq 0.
\end{cases}
\end{align}

\subsection{Asymptotic analysis on  a pure RH problem}
\hspace*{\parindent}

The pure RH problem is described as follows.

\begin{prob}\label{2mrhp}
    Find   $M^{rhp}(z)=M^{rhp}(z;x,t)$ which satisfies
	  \begin{itemize}
	  	\item   $M^{rhp}(z)$ is analytical in $\mathbb{C}\backslash \Sigma^{(4)}$.
	  	\item $M^{rhp}(z)$  satisfies the jump condition
\begin{equation*}
	  		M^{rhp}_+(z)=M^{rhp}_-(z)V^{(4)}(z),
	  	\end{equation*}
	  	where $V^{(4)}(z)$ is given by (\ref{2v4}).
	  	\item $M^{rhp}(z)$  has the same asymptotic behavior with  $M^{(4)}(z)$.
	
	  \end{itemize}
\end{prob}

\subsubsection{Local paramatrix near $z=1$}
\hspace*{\parindent}

In the region $\mathcal{P}_{>+1}(x,t)$, the phase points $\xi_1$ and $\xi_2$ merge to $1$ as $t \to +\infty$,
and we have for small $k t^{-1/3}$
\begin{align}
	t \theta(z)&= -t \left( \textcolor{red} {  (z-1)^3 }+\frac{1}{2}\sum_{n=4}^\infty (-1)^{n+1}(n-1)(z-1)^n  \right) \nonumber \\
	& +(x-2t) \left( \textcolor{red} { (z-1) }+ \frac{1}{2}\sum_{n=2}^\infty  (-1)^{n+1} (z-1)^n\right) \nonumber\\
	&:= \frac{4}{3} k^3 + s k + S(t;k ),\label{asymn2}
\end{align}
where
\begin{align}
&k  = -\tau^{\frac{1}{3}} (z-1),\ \ \
	s  = -\frac{8}{3}(\xi-1)\tau^{\frac{2}{3}}, \label{case2ks} \\
 & S(t;k )  = \frac{2}{3} \sum_{n=4}^\infty (n-1) \tau^{\frac{3-n}{3}} k^{n}-\frac{4}{3}(\xi-1)\sum_{n=2}^\infty \tau^{\frac{3-n}{3}} k^{n},\label{dssd}
\end{align}
with $ \tau = \frac{3}{4}t$.
The scaled phase points have the following properties.
\begin{proposition} \label{2k1k2}
	In the transition region $ \mathcal{P}_{> +1}(x,t)$,  two phase points $k_j= -\tau^{1/3} (  \xi_j- 1  ), \ j=1,2$
	always lie  in a  fixed interval
	\begin{align}
	k_j \in	[-(3/4)^{1/3} \sqrt{2C}, (3/4)^{1/3} \sqrt{2C} ].
	\end{align}
	
\end{proposition}

 For a fixed constant $\varepsilon \leq \sqrt{2C}$,  define a  neighbourhood
$$\mathcal{U}_z(1)= \{z \in \mathbb{C}: |z-1|< \varepsilon \tau^{-1/3} \},$$
such that $\xi_1, \xi_2 \in \mathcal{U}_z(1)$.
The map $z \mapsto k$ maps $\mathcal{U}_{z}(1)$ onto the disk $\mathcal{U}_k(0) = \{k \in \mathbb{C}: |k|< \varepsilon \}$ in the $k$-plane, where $k_1, k_2 \in\mathcal{U}_k(0)$.

\begin{proposition}\label{order2}
 In the   region $ \mathcal{P}_{> + 1}(x,t)$,
  the series $S(t;k ) $ defined by (\ref{dssd})  converges uniformly   in  $ \mathcal{U}_z(1)$  and decay  concerning   $t$.
\end{proposition}

From Proposition \ref{order2}, the first two terms  in (\ref{asymn2}) are crucial  to
obtain the   Painlev\'e  asymptotics.  Additionally, we have the following proposition  in the $k$-plane.

\begin{proposition} \label{opow2}
In the transition region $ \mathcal{P}_{> +1}(x,t)$, when $t$ is large enough,
\begin{align}
  &{\rm Re}\left[i  \left(\frac{4}{3} k^3 + sk \right)  \right]  \geq \frac{8}{3}u^2v, \ \ for \ k \in \Omega'_{j}, \label{re211}\\
  &{\rm Re}\left[i  \left(\frac{4}{3} k^3 + sk \right)  \right]  \leq -\frac{8}{3}u^2v,\ \ for \ k \in \overline{\Omega}'_{j}, \label{re222}
\end{align}
where   $\Omega_{j}'$ is the scaled contour of $\Omega_{j}$ and  $k=k_j+u+iv$, $j=1,2$.
\end{proposition}

\begin{proof}
The proofs for Proposition \ref{2k1k2}, \ref{order2}, and \ref{opow2} are similar to the proofs on the region $\mathcal{P}_{< -1}(x,t)$. See more details in Proposition \ref{opow},  \ref{order1}, and \ref{opow1}.
\end{proof}
As $t \to +\infty$, the jump matrix decays to the unit matrix exponentially fast and uniformly outside a small neighborhood of $z=1$.  In this case, we  construct the solution of $M^{rhp}(z)$ as follows:
\begin{align}\label{2trans6}
	M^{rhp}(z) = \begin{cases}
		E(z),\quad z \in \mathbb{C} \setminus \mathcal{U}_{z}(1),\\
		E(z) M^{loc}(z), \quad z \in \mathcal{U}_{z}(1).
	\end{cases}
\end{align}
 where $M^{loc}(z)$ satisfies the following RH problem.

  	\begin{prob}
  		Find  $ M^{loc}(z)=M^{loc}(z;x,t)$ which satisfies
  		\begin{itemize}
  			\item Analytical in $\mathcal{U}_{z}(1)\backslash  {\Sigma}^{loc}$, where $\Sigma^{loc}= \Sigma^{(4)} \cap \mathcal{U}_{z}(1)  $. See Figure \ref{sign21}.
  			\item Jump condition:\begin{equation*}
  				M^{loc}_+(z)=M^{loc}_-(z) {V}^{loc}(z), \ z\in  {\Sigma}^{loc},
  			\end{equation*}
  			where
  			\begin{align*}
  				V^{loc}(z)= \begin{cases}
  					\left( \begin{array}{cc}
  						1& -  R(\xi_j)  e^{-2it\theta(z)}\\
  					0    & 1
  					\end{array}\right),\quad z\in \Sigma_{j}, \ j=1,2,\\
  					\left(	\begin{array}{cc}
  						1&  0 \\
  							\overline{R(\xi_j)} e^{2it\theta(z)} & 1
  					\end{array}\right),\quad z \in \overline{\Sigma}_{j},\ j=1,2,\\
  T(z)^{-\sigma_3} V(z)T(z)^{\sigma_3},\quad z \in (\xi_1, \xi_2).
  				\end{cases}
  			\end{align*}

  			\item   $M^{loc}(z)(M^{\infty}( (z-1)\tau^{1/3}))^{-1} \to I, \ t\to +\infty$ uniformly for $z \in \partial \mathcal{U}_{z}(1).$
  		\end{itemize}
  	\end{prob}
  	
Since    $\xi_1, \xi_2 \to 1$     as $t\to +\infty$,
    it follows   from   (\ref{weie})  that    $ r(\xi_j)\to   -1$   as $t\to \infty$ for the generic case,
which further causes    singularity of the  $ \frac{  r(\xi_j) } {1-|r(\xi_j)|^2} $ as $\xi_j\to 1$.
However $R(\xi_j)$  has no singularity as $\xi_j\to 1$ since this  singularity can be
neutralized by the factor $T_-(\xi_j)^{-2} $.
 We  define a  cutoff function $\chi(z)\in C_0^\infty(\mathbb{R},[0,1])$
satisfying
\begin{equation}
\chi (z)=1, \ z\in  \mathbb{R}\cap \mathcal{U}_z(1),
\end{equation}
and write the function $R(z)$   as
\begin{equation}
 R(z)=(1-\chi (z)) \frac{\overline{r(z)} }{1-|r(z)|^2}  T ( z)^{-2}+\chi (z)  F(z)  G(z)^2,
\label{5oedk}
\end{equation}
where
$$F(z) := \frac{ \overline{S_{21} (z)} }{S_{11} (z)},  \ \ \  G(z):=\frac{ s_{11}(z) }{T_+(z)}. $$
  Using Lemmas  \ref{lemma3.1} and  \ref{lemma21},   each factor in the r.h.s. of (\ref{oedk}) is   analytical in $ \Omega_1\cup \Omega_2$,
 with well defined nonzero limit on $\partial \Omega_1\cup \partial\Omega_2$.

 \begin{lemma} \label{eowe}
  Let $q_0   \in \tanh (x)+ L^{1,2}(\mathbb{R}), q'_0 \in W^{1,\infty }(\mathbb{R})$, then   $F(z)\in H^1( \mathbb{R})$.
 \end{lemma}

 \begin{proof}
In a similar to the proof  of  $r(z)\in H^1(\mathbb{R})$ in  Lemma \ref{lemm3.2},
fix a small $\delta>0$, such that   $|z\pm 1|<\delta$ and $ |z|<\delta$ have empty intersection.
 By using Lemma \ref{lemma3.1}, we have
  \begin{align}
& F(z)\in W^{1,\infty}(I_\delta)\cap H^1(I_\delta), \label{5p90}\\
&|\partial_z^j|\leq c \langle z\rangle^{-1}, \ \ j=0, 1. \label{5p91}
 \end{align}
where
$$ I_\delta=\mathbb{R}  \setminus ((-\delta, \delta)\cup (1-\delta,1+\delta)\cup(-1-\delta, -1+\delta)). $$

Next to show
  \begin{align}
& F(z)\in  H^1( \mathbb{R }\setminus I_\delta). \label{5p92}
 \end{align}
For $|z-1|<1$,   we write $ f(z)$ as
 \begin{align}
 F(z) =\frac{ -\overline{s}_+  +\int_1^z \partial_y \overline{S_{21}(y)}dy}{ s_+ +\int_1^z \partial_z S_{11}(y )dy}, \label{fre}
\end{align}
by which,  it follows  that $F'(z)$ is defined and bounded around $z=1$.
The same discussion holds at $z=-1$.   $F'(z)$ is defined and bounded around  $z=0$  by using  symmetry $F(z^{-1})=\bar F(z)$
and (\ref{5p91}).

Finally,  from (\ref{5p90}) and (\ref{5p92}),  we conclude that $F(z)\in H^1(\mathbb{R})$.

 \end{proof}

 Then   we show  the following proposition.
\begin{proposition}   Let $q_0   \in \tanh (x)+ L^{1,2}(\mathbb{R}), q'_0 \in W^{1,\infty }(\mathbb{R})$.
  Then for $k \in \left(k_2,  k_1\right)$,
\begin{align}
  \Big| R(z) e^{ 2it\theta \left(z\right)}- R(1) e^{8ik^3/3+2isk }  \Big|   \lesssim t^{-1/6}. \label{5e1}
\end{align}
And for $k \in  {\Sigma}_{j}'$, $j=1,2$,
\begin{align}
  \Big| R(\xi_j) e^{ 2it\theta \left(z\right)}-  R(1) e^{8ik^3/3+2isk }  \Big|   \lesssim t^{-1/6}, \label{eee2}
\end{align}
where ${\Sigma}_{j}'$ is the scaled contour of ${\Sigma}_{j}$.
\end{proposition}
\begin{proof}
For $k \in \left(k_2, k_1\right)$,  then  $z\in (\xi_2, \xi_1)$ and $z$ is real,
$$\left|  e^{2i t   \theta (z) } \right|=1, \ \ \left| e^{  i(\frac{8}{3} k^3+ 2s k)  } \right|=1.$$
  so  we have
\begin{align}
&\left| R (z)    e^{2i t   \theta  (z ) } - R(1)e^{  i(\frac{8}{3} k^3+ 2s k)  }    \right| \leq |G(z)|^2  | F (z) - F( 1 )  | \nonumber\\
&+ |F(1)| |G(z)+G(1)| |G(z)-G(1)|
+ \left|  R ( 1 )   \right|       \left|   e^{iS(t;k)}-   1  \right|. \label{5poe1}
 \end{align}

By Lemma  \ref{eowe} and H\"older inequality,
  \begin{align}
& \left| F (z) - F ( 1)  \right|=     \left|\int_{-1}^{z} F'(s) ds
    \right| \leq \|F \|_{H^1 } |z-1|^{1/2} \leq c t^{-1/6}. \label{5poe2}
 \end{align}

From  (\ref{trance}) and (\ref{resed}), we have
\begin{align}
G(z)=	  \prod_{j=1}^{N} \frac{z-z_j}{z-\bar{z}_j} \exp\left(-i \int_{-\infty}^0 \frac{\nu(\zeta)}{\zeta-z} \, \mathrm{d}\zeta\right)
\exp \left(  - i \int_{0}^{\infty}   \frac{\nu(\zeta) }{2\zeta}  \, \mathrm{d}\zeta \right),\nonumber
\end{align}
with which,
\begin{align}
&G(z) -G(1) = e^{   -i \int_{-\infty}^0 \frac{\nu(\zeta)}{\zeta-z} \, \mathrm{d}\zeta}
e^{  - i \int_{0}^{\infty}   \frac{\nu(\zeta) }{2\zeta}  \, \mathrm{d}\zeta }
\left(   \prod_{j=1}^{N} \frac{z-z_j}{z-\bar{z}_j} -\prod_{j=1}^{N} \frac{1-z_j}{1-\bar{z}_j}\right) \nonumber\\
& + \prod_{j=1}^{N} \frac{1-z_j}{1-\bar{z}_j}  e^{   - i \int_{0}^{\infty}   \frac{\nu(\zeta) }{2\zeta}  \, \mathrm{d}\zeta }  \left(  e^{   -i \int_{-\infty}^0 \frac{\nu(\zeta)}{\zeta-z} \, \mathrm{d}\zeta} - e^{   -i \int_{-\infty}^0 \frac{\nu(\zeta)}{\zeta-1} \, \mathrm{d}\zeta}    \right). \nonumber
\end{align}
In the r.h.s. all factor before two brackets have the absolute value $\leq 1$ for $z\in \mathbb{R}$,
while the  terms in two  brackets  are controlled by $|z-1|$, therefore,
\begin{align}
&| G(z) -G(1) | \leq c |z-1|. \label{addw}
\end{align}

 Substituting (\ref{5poe2}) and (\ref{5poe3}) into (\ref{5poe1}) gives
  \begin{align}
& \left| R (z)    e^{2i t   \theta  (z ) } - R (1)   e^{  i(\frac{8}{3} k^3+ 2s k)  }    \right|   \leq c t^{-1/6}. \label{5poe3}
 \end{align}

For $k \in  {\Sigma}_{1}$,  denote $k = k_1 +u+iv$. By (\ref{re211}), $ \left| e^{i \left( \frac{8}{3} k^3 + 2 s k \right)}\right|$ is bounded.
Similarly to the case on the real axis, we can obtain the estimate   (\ref{eee2}). The estimate   on the other jump contours can be given in the same way.
\end{proof}

Therefore, as $t \to +\infty$,  the solution of the RH problem  $ {M}^{loc}(z)$
can be approximated by the solution of a  limit model
\begin{align}
M^{loc}(z) = M^{\infty}( k) + \mathcal{O} \left(t^{-1/6} \right),
\end{align}
where $M^{\infty}(k)$ is given by (\ref{eegs6})  with the  argument
\begin{align}
& \varphi_0  =\arg \frac{\overline{S_{21} (1)}}{S_{11}(1)}+2\sum_{j=0}^{N-1} \arg (1-z_j)
-  \int_{0}^{\infty} \frac{ \nu(\zeta)}{ 2\zeta } \mathrm{d}\zeta -  \int_{-\infty}^0  \frac{ \nu(\zeta)}{  \zeta-1 } \mathrm{d}\zeta, \label{trans857}
\end{align}
 where the functions  $S_{21}(z)$ and $S_{11}(z)$ are defined by  (\ref{oeeedk2}),  and $\nu(z)$ is given by (\ref{nu}).

\subsubsection{Small norm RH problem}
\hspace*{\parindent}

We  consider the error function  $E(z)$ defined by  (\ref{2trans6}).
 Denote $\Sigma^E = \partial \mathcal{U}_{z}(1) \cup \left(\Sigma^{(4)} \backslash \mathcal{U}_{z}(1)\right).$

\begin{prob}\label{2iew}
  Find   $E(z)$ with the properties as follows
          \begin{itemize}
          	\item  $E(z)$ is analytical in $\mathbb{C}\backslash  \Sigma^{E}$.
          	\item $E(z)$ satisfies the jump condition
          	\begin{align*}
          		E_{+}(z)=E_-(z)V^{E}(z), \quad z\in \Sigma^{E},
          	\end{align*}
          	where the jump matrix is given by
   \begin{align*}
                V^{E}(z)= \begin{cases}
                   V^{(4)}(z),\quad z \in \Sigma^{(4)} / \mathcal{U}_{z}(1),\\
                  M^{loc}(z) , \quad z \in \partial\mathcal{U}_{z}(1).
                \end{cases}
            \end{align*}

          	\item As $z\to \infty$, $E(z)=I+\mathcal{O}(z^{-1}).$	

          \end{itemize}
\end{prob}

 Similarly to the analysis in the previous sector,  we obtain
             \begin{equation*}
 	|V^{E}(z)-I|=\begin{cases}
 		\mathcal{O}(e^{-c t}),\quad z\in \Sigma^{E} / \mathcal{U}_{z}(1), \\
 		\mathcal{O}(t^{-1/3}), \quad z\in \partial \mathcal{U}_{z}(1),
 	\end{cases}
 \end{equation*}
 and  the RH problem \ref{2iew} exists an unique solution
         \begin{equation}\label{2trans7}
            E(z)=I +\frac{1}{2\pi i} \int_{\Sigma^{E}} \frac{\mu_E(\zeta) \left( V^{E}(\zeta) -I\right)}{\zeta-z}\, \mathrm{d}\zeta,
         \end{equation}
        where $\mu_E \in L^2\left( \Sigma^{E}\right)$   satisfies $\left(I-C_{w^{E}}\right)\mu_E=I$.
As $z \to \infty$,
        \begin{equation*}
            E(z)=I +\frac{E_1}{z} +\mathcal{O}\left(z^{-1}\right),
        \end{equation*}
      where $ E_1=-\frac{1}{2\pi i} \int_{\Sigma^{E}} \mu_E(\zeta)\left( V^{E}(\zeta)-I \right)\, \mathrm{d}\zeta.$
       In particular, $E_1$ and $E(0)$ have the following estimations.
        \begin{proposition}
            \begin{align}
          &  E_1=   \tau^{-1/3} M_1^\infty(s)+    \mathcal{O}(t^{-2/3}),\\[4pt]
       & E(0) = I+\tau^{-1/3} M_1^\infty(s) + \mathcal{O}(t^{-2/3}).\label{2e0}
 \end{align}
        \end{proposition}

\subsection{Asymptotic analysis on  a pure $\bar{\partial}$-problem}
\hspace*{\parindent}

In the subsection, we consider the pure $\bar{\partial}$-problem. By (\ref{2trans5}), we have
            \begin{equation}\label{2m3m2m2rhp}
                M^{(5)}(z)=M^{(4)}(z)\left(M^{rhp}(z)\right)^{-1},
            \end{equation}
            which  satisfies the following $\bar{\partial}$-problem.

\begin{prob3}\label{2m5}
 Find a matrix-valued function  $M^{(5)}(z)$ which satisfies
            \begin{itemize}
                \item  $M^{(5)}(z)$ is continuous and has sectionally continuous first partial derivatives in $\mathbb{C}  \backslash \left(\mathbb{R} \cup \Sigma^{(4)}\right)$.
                \item As $z\to  \infty$, $M^{(5)}(z)=I+\mathcal{O}(z^{-1})$.

                \item For $z\in \mathbb{C}$,  we have
                \begin{equation*}
                    \bar{\partial} M^{(5)}(z) = M^{(5)}(z) W^{(5)}(z),
                \end{equation*}
                where $W^{(5)}(z):=M^{rhp}(z)  \bar{\partial} R^{(3)}(z) \left(M^{rhp}(z)\right)^{-1}$ with $\bar{\partial} R^{(3)}(z)$ in (\ref{2r3})  and $M^{rhp}(z)$ in (\ref{2trans6}) .
            \end{itemize}

\end{prob3}

We can prove that the solution of the $\bar{\partial}$-RH problem \ref{2m5} exists and that the asymptotic unfolding at infinity satisfies
     \begin{equation}\label{2m5infty}
            M^{(5)}(z)=I + \frac{M^{(5)}_1(x,t) }{z} +\mathcal{O}(z^{-2}), \quad z\to \infty,
        \end{equation}
where
$$M^{(5)}_1(x,t)=\frac{1}{\pi} \iint_{\mathbb{C}} M^{(5)}(\zeta)W^{(5)}(\zeta)\, \mathrm{d}\zeta\wedge \mathrm{d} \bar \zeta.$$
with
\begin{equation}\label{2m51infty}
	|M^{(5)}_1(x,t)| \lesssim t^{-1/2}, \ \ |M^{(5)}(0)-I| \lesssim t^{-1/2}.
\end{equation}

\begin{proposition}
  $M^{(3)}(0)$ admits the estimate
	\begin{align}\label{2m30}
		M^{(3)}(0)
		=E(0) + \mathcal{O}(t^{-1/2}),
	\end{align}
	where $E(0)$ is given by (\ref{2e0}).
\end{proposition}

\subsection{Proof  of Theorem  \ref{th}---Case II  }\label{thm-result2}
\hspace*{\parindent}

\begin{proof}
		Inverting  the sequence  of transformations (\ref{2trans1}), (\ref{2trans2}), (\ref{2trans3}) , (\ref{2trans4}), (\ref{2trans5}), (\ref{2trans6}),
	and especially taking $z\to \infty$ vertically   such that  $R^{(3)}(z)= G(z)=I$, the    solution of RH problem \ref{RHP0} is given by
	\begin{align}
		& 	M(z) = T(\infty)^{\sigma_3} (  I + z^{-1}\sigma_1 M^{(3)}(0)^{-1} )  M^{(5)}(z) E(z)T(z)^{-\sigma_3} + \mathcal{O}(e^{-ct}). \nonumber
	\end{align}
	Noticing that (\ref{2Tinfty}) and  (\ref{2m30}), further substituting asymptotic expansions  (\ref{2Tzinfty}),  (\ref{2trans7}) and (\ref{2m5infty}) into the above formula,
	then the reconstruction formula (\ref{sol}) yields
	\begin{align*}
		&q(x,t) 
		=T(\infty )^2 \left(  1+ \frac{1}{2} i \tau^{-\frac{1}{3}} \left( u(s)e^{i\varphi_0} + \int_s^\infty u^2(\zeta)  \mathrm{d}\zeta \right)  \right) + \mathcal{O}\left( t^{-\frac{1}{2}}\right),
	\end{align*}
	which gives the result (\ref{q2}) in Theorem \ref{th}.

\end{proof}

\appendix

\section{A model  for $\mathcal{P}_{-1}(x,t)$ and $\mathcal{P}_{+1}(x,t)$} \label{appx}
\hspace*{\parindent}

 The long-time asymptotics in the region $\mathcal{P}_{\mp 1}(x,t)$ are related to the solution $M^{\infty}( k)$ of the following RH problem.

Let $k\in \mathbb{C}$ be a complex variable,  $k_1, k_2\in \mathbb{R}$ and $r_0 \in \mathbb{C}$  be   fixed constants,  $s$ be a parameter,
and  $\Sigma_j, \ j=1,\cdots,6$,  be rays  passing through points $k_1, k_2$ at a fixed angle  $\varphi$  with the real axis, see   Figure \ref{FIG635}.
Denote $\Sigma^{\infty}=\cup_{j=1}^6\Sigma_j$.
We consider   the  following   model RH problem.
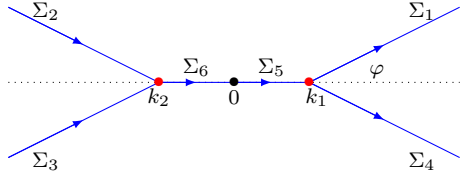
\begin{figure}
	\begin{center}
		\begin{tikzpicture}[scale=1]
		\draw [dotted ](-6.5,0)--(-0.5,0);
			\draw [blue ](-4.5,0)--(-2.5,0);
			\draw [blue, -latex](-4.5,0)--(-4,0);
			\draw [blue, -latex](-3.5,0)--(-3,0);
			\draw [ blue, ](-4.5,0)--(-6.5,1);
			\draw [blue,-latex ] (-6.5,1)--(-11/2,1/2);
			\draw [ blue, ](-4.5,0)--(-6.5,-1);
			\draw [blue,-latex ] (-6.5,-1)--(-11/2,-1/2);
			\draw [blue,  ](-2.5,0)--(-0.5, 1 );
			\draw [blue,-latex ] (-2.5,0)--(-1.5, 1/2);
			\draw[ blue, ](-2.5,0)--(-0.5,-1);
			\draw [blue,-latex ] (-2.5,0)--(-1.5,-1/2);
			\node  [below]  at (-1,1.2) {\scriptsize $\Sigma_1$};
			\node  [below]  at (-6,1.2) {\scriptsize $\Sigma_2$};
			\node  [below]  at (-6,-0.8) {\scriptsize $\Sigma_3$};
			\node  [below]  at (-1,-0.8) {\scriptsize $\Sigma_4$};
			\node[red]    at (-4.5,0)  {\scriptsize $\bullet$};
			\node[red]    at (-2.5,0)  {\scriptsize $\bullet$};
			\node    at (-3.5,0)  {\scriptsize $\bullet$};
			\node    at (-4.5,-0.2)  {\scriptsize $k_2$};
			\node    at (-2.4,-0.2)  {\scriptsize $k_1$};
				\node    at (-3.5,-0.2)  {\scriptsize $0$};
			\node    at (-3, 0.2)  {\scriptsize $\Sigma_5$};
             \node    at (-4, 0.2)  {\scriptsize $\Sigma_6$};
		\node    at (- 1.6, 0.15)  {\scriptsize $\varphi$};
		\end{tikzpicture}
	\end{center}
	\caption{\footnotesize The jump contours  $\Sigma^{\infty}$.}
	\label{FIG635}
\end{figure}

\begin{prob}\label{2minfty}
	Find    $M^{\infty}( k)=M^{\infty}(k;s )$   with properties
	\begin{itemize}
		\item  $M^{\infty}(k)$ is analytical in $ \mathbb{C} \setminus \Sigma^{\infty}$.
		\item  $M^{\infty}(k)$  satisfies the jump condition\begin{equation*}
			M^{\infty}_+( k)=M^{\infty}_-(k)V^{\infty}(k),
		\end{equation*}
		where
		\begin{align}\nonumber
			V^{\infty}(k)= \begin{cases}
				 \left( \begin{array}{cc}
					1& 0\\
					r_0   e^{ 2i \left(\frac{4}{3}  k^3 +  sk \right) }   & 1
				\end{array}\right):=B_+,\quad k\in  {\Sigma}_{1}\cup  {\Sigma}_{2},\\
				 \left(	\begin{array}{cc}
					1&   -\overline{r_0} e^{- 2i \left(\frac{4}{3}  k^3 +  sk \right) }  \\
					0 & 1
				\end{array}\right) :=B_-^{-1},\quad k \in  {\Sigma}_{3} \cup  {\Sigma}_{4},\\
				B_-^{-1} B_+,\quad k \in {\Sigma}_{5}\cup {\Sigma}_{6}.
			\end{cases}
		\end{align}
		\item   $M^{\infty}( k)=I+\mathcal{O}(k ^{-1}),	\quad k \to  \infty.$

	\end{itemize}
\end{prob}

The above RH problem can be transformed into a standard Painlev\'{e} \uppercase\expandafter{\romannumeral2} model
via an  appropriate equivalent  deformation.
For this purpose, we add four new auxiliary paths $L_j, j=1,2,3,4$,  passing through the point $k=0$
at  the angle  $\pi/3$  with real axis,
which company with the original paths $\Sigma^\infty$
 divide the complex plane into eight regions $\Omega_j, j=1,\cdots,8$. See Figure \ref{desc45}.

We further define
\begin{align}
&P(k) =\left\{\begin{matrix}
 B_+  ^{-1},\  \ & k\in \Omega_2\cup\Omega_4,\cr
 B_-^{-1},\  \ & k\in \Omega_6\cup\Omega_8,\cr
I,\  \ & k \in \Omega_1\cup\Omega_3\cup\Omega_5\cup\Omega_7,
\end{matrix}\right. \nonumber
\end{align}
and  make a transformation
\begin{align}
&\widehat M^{P} (k)=M^{\infty} (k)P(k),\label{pre1}
\end{align}
then we obtain a  Painlev\'e model.

\begin{prob}\label{mp2}
    Find   $ \widehat  M^{P} ( k)=\widehat M^{P} (k;s )$ with properties
    \begin{itemize}
        \item  $\widehat M^{P}_+( k)$ is analytical in $\mathbb{C}\setminus  \widehat \Sigma^{P} $, where $\widehat \Sigma^{P}=\cup_{j=1}^4 \left\{L_j  = e^{i (j-1)\pi/3 } \mathbb{R}_+ \right\}. $ See  Figure  \ref{desc45}.
        \item  $\widehat M^{P}_+( k)$ satisfies the  jump condition \begin{equation*}
           \widehat M^{P}_+( k)=\widehat M^{P}_-(k) \widehat V^{P}(k), \ \ k\in \widehat \Sigma^{P},
        \end{equation*}
        where
       \begin{align}\label{vp}
       \widehat V^{P}(k)= \begin{cases}
        \left( \begin{array}{cc}
       		1& 0\\
       	r_0 e^{2i \left(\frac{4}{3}  k^3 +  sk \right) }	& 1
       	\end{array}\right),\quad k\in L_{1}\cup L_{2},\\
         \left(	\begin{array}{cc}
       			1&  -\overline{r_0} e^{-2i \left(\frac{4}{3}  k^3 +  sk \right)  }   \\
       			0 & 1
       		\end{array}\right),\quad k \in L_{3} \cup L_{4}.
       	\end{cases}
       \end{align}
        \item   $\widehat M^{P}( k)=I+\mathcal{O}(k ^{-1}),	\quad k \to  \infty. $

\end{itemize}
    \end{prob}

\begin{figure}
\begin{center}
\begin{tikzpicture}\label{Fig3}
		\draw [dotted,-latex ](-6.8,0)--(-0.3,0);
			\draw [blue ](-4.5,0)--(-2.5,0);
			\draw [blue, -latex](-4.5,0)--(-4,0);
			\draw [blue, -latex](-3.5,0)--(-3,0);
			\draw [ blue, ](-4.5,0)--(-6.5,1);
			\draw [blue,-latex ] (-6.5,1)--(-11/2,1/2);
			\draw [ blue, ](-4.5,0)--(-6.5,-1);
			\draw [blue,-latex ] (-6.5,-1)--(-11/2,-1/2);
			\draw [blue,  ](-2.5,0)--(-0.5, 1 );
			\draw [blue,-latex ] (-2.5,0)--(-1.5, 1/2);
			\draw[ blue, ](-2.5,0)--(-0.5,-1);
			\draw [blue,-latex ] (-2.5,0)--(-1.5,-1/2);
			\node[red]    at (-4.5,0)  {\scriptsize $\bullet$};
			\node[red]    at (-2.5,0)  {\scriptsize $\bullet$};
			\node    at (-3.5,0)  {\scriptsize $\bullet$};
			\node    at (-4.5,-0.2)  {\scriptsize $k_2$};
			\node    at (-2.4,-0.2)  {\scriptsize $k_1$};
				\node    at (-3.5,-0.2)  {\scriptsize $0$};
			\node    at (-3, 0.2)  {\scriptsize $\pi/3$};
		\node    at (- 1.6, 0.15)  {\scriptsize $\varphi$};

 \draw[dotted,thick](-4.5,-1.5)--(-2.5, 1.5 );
 \draw[dotted,thick, -latex  ](-4.5,-1.5)--(-4, -0.75 );
 \draw[dotted,thick, -latex  ](-3.5,0)--(-3, 0.75 );
\draw[dotted,thick](-4.5,1.5)--(-2.5,-1.5 );
  \draw[dotted,thick, -latex  ](-4.5, 1.5)--(-4, 0.75 );
   \draw[dotted,thick, -latex   ](-3.5,0)--(-3, -0.75  );

  \node    at (-1, -0.2)  {\scriptsize $\Omega_1$};
    \node    at (-2.1, 0.7)  {\scriptsize $\Omega_2$};
 \node    at (-3.5,1)  {\scriptsize $\Omega_3$};
 \node    at (-5, 0.7)  {\scriptsize $\Omega_4$};
 \node    at (-6, -0.2)  {\scriptsize $\Omega_5$};
  \node    at (-5,-0.7)  {\scriptsize $\Omega_6$};
  \node    at (-3.5, -1)  {\scriptsize $\Omega_7$};
    \node    at (-2.1, -0.7)  {\scriptsize $\Omega_8$};

\end{tikzpicture}
\end{center}
\caption {\footnotesize  Add four new auxiliary lines on  the jump contour  for  $M^{\infty} (k)$, which then can be deformed
into the  Painlev\'e model $\widehat  M^{P}(k)$  with  the jump contour  in   four dotted rays.}
\label{desc45}
\end{figure}
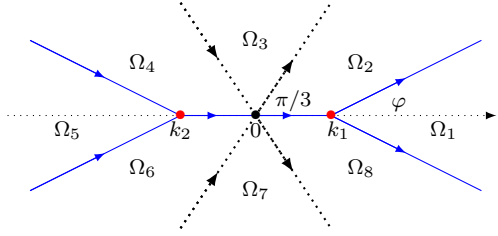

If  the parameter $r_0$  is  not real-valued, the solution to the RH problem \ref{mp2} is related to the Painlev\'{e} \uppercase\expandafter{\romannumeral34} equation. However, looking closely at the jump matrix, we find that the RH problem   \ref{mp2} can be reduced into that usually associated with  the Painlev\'{e} \uppercase\expandafter{\romannumeral2} equation by a gauge transformation.
Let $\varphi_0=\arg r_0$, so  $r_0 = |r_0 | e^{i\varphi_0}$. Following the idea  \cite{Fokas1}, we make the following transformation
\begin{equation}
    {M}^{P}(k) =e^{i \left( \frac{ \varphi_0}{2} -\frac{\pi}{4}\right)\widehat\sigma_3} \widehat {M}^{P}( k),\label{pre2}
\end{equation}
then ${M}^{P}(k)$ satisfies the RH problem.
\begin{prob}\label{mpain2}
    Find   $    M^{P} ( k)=  M^{P} (k;s )$ with properties
    \begin{itemize}
        \item  $M^{P} ( k)$ is analytical  in $\mathbb{C}\setminus   \Sigma^{P} $, where $\Sigma^P = \bigcup_{j=1}^4 L_j$, see Figure \ref{Sixrays}.
        \item  $M^{P} ( k)$ satisfies the  jump condition \begin{equation*}
             M^{P}_+( k)=  M^{P}_-(k)  V^{P}(k), \ \ k\in   \Sigma^{P},
        \end{equation*}
        where
            \begin{align}\label{vp2}
       	 {V}^{P}( k)= \begin{cases}
        e^{ -i \left(\frac{4}{3}  k^3 + sk\right)\widehat{\sigma}_3  } \left( \begin{array}{cc}
       		1& 0\\
       		i |r_0| & 1
       	\end{array}\right),\quad k \in L_{1},\\[8pt]
       e^{ -i \left(\frac{4}{3}  k^3 + sk\right)\widehat{\sigma}_3  } \left( \begin{array}{cc}
       		1& 0\\
       		-i |r_0| & 1
       	\end{array}\right),\quad k \in   L_{2},\\[8pt]
       	 e^{ -i \left(\frac{4}{3}  k^3 + sk\right)\widehat{\sigma}_3  } \left(	\begin{array}{cc}
       			1&  i |\overline{ r_0} | \\
       			0 & 1
       		\end{array}\right) ,\quad k \in L_{3},\\
       e^{ -i \left(\frac{4}{3}  k^3 + sk\right)\widehat{\sigma}_3  } \left(	\begin{array}{cc}
       			1&  -i |\overline{ r_0} | \\
       			0 & 1
       		\end{array}\right) ,\quad k \in  L_{4}.	
       	\end{cases}
       \end{align}
        \item   $ M^{P}( k)=I+\mathcal{O}(k ^{-1}),	\quad k \to  \infty. $

\end{itemize}
\end{prob}
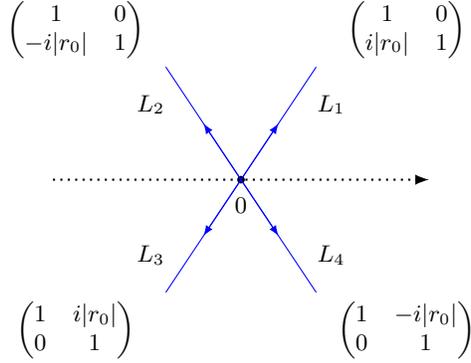
\begin{figure}
	\begin{center}
		\begin{tikzpicture}[scale=1]

			\node[shape=circle,fill=black,scale=0.15] at (0,0) {0};
			\node[below] at (0,-0.1) {\footnotesize $0$};
			\draw[dotted,thick,-latex] (-2.5,0 )--(2.5,0);
			\draw [blue] (-1,-1.5 )--(1,1.5);
			\draw [blue,-latex] (0,0)--(0.5,0.75);
			\draw [blue,-latex] (0,0)--(0.5,-0.75);
			\draw [blue] (-1,1.5 )--(1,-1.5);
			\draw [blue,-latex] (0,0)--(-0.5,0.75);
			\draw [blue, -latex] (0,0)--(-0.5,-0.75);

			\node at (1.2,1) {\footnotesize$L_1$};
			\node at (-1.2,-1) {\footnotesize$L_3$};
			\node at (-1.2,1) {\footnotesize$L_2$};
			\node at (1.2,-1) {\footnotesize$L_4$};
			
			\node  at (2.2,2) {\footnotesize $\begin{pmatrix} 1 & 0 \\ i|r_0| & 1 \end{pmatrix}$};
			\node  at (-2.2,2) {\footnotesize $\begin{pmatrix} 1 & 0\\ -i|r_0|& 1 \end{pmatrix}$};
			\node  at (2.2,-2) {\footnotesize $\begin{pmatrix} 1 &  -i|r_0| \\ 0& 1 \end{pmatrix}$};
			\node  at (-2.2,-2) {\footnotesize $\begin{pmatrix} 1 &  i|r_0| \\ 0 & 1 \end{pmatrix}$};
			
		\end{tikzpicture}
		\caption{ \footnotesize { The jump contours of $M^P(k)$.}}
		\label{Sixrays}
	\end{center}
\end{figure}
This  RH problem model   \ref{mpain2} is exactly a special
case of the Painlev\'{e} \uppercase\expandafter{\romannumeral2} model.
 Therefore the solution of the RH problem model   \ref{mpain2}
is given by
\begin{equation}
   {M}^{P}( k) = I + \frac{  {M}_1^{P}(s)}{k}+\mathcal{O}\left(k^{-2}\right), \label{eue2}
\end{equation}
where $ {M}_1^{P}(s)$ is given by
\begin{align}\label{posee}
M_1^P(s) = \frac{1}{2} \begin{pmatrix} -i\int_s^\infty u(\zeta)^2\mathrm{d}\zeta & u(s) \\ u(s) & i\int_s^\infty u(\zeta)^2\mathrm{d}\zeta \end{pmatrix},
\end{align}
and for each $C_1 > 0$,
\begin{align}\label{mPbounded}
	\sup_{z \in \mathbb{C}\setminus \Sigma^P} \sup_{s \geq -C_1} |M^P(k,s)|  < \infty.
\end{align}
And  $u(s)$ is a solution of the Painlev\'{e} \uppercase\expandafter{\romannumeral2} equation
\begin{align}
	u_{ss} = 2u^3 +su, \quad s \in \mathbb{R},
\end{align}
With  transformations  (\ref{pre1}) and (\ref{pre2}),  expanding $M^\infty (k)$ along the  region $\Omega_3$ or $\Omega_7$ yields
\begin{align}
M^{\infty}(k) = I + \frac{ M_1^{\infty} ( s)}{k}+ \mathcal{O} \left(k^{-2 }\right),\label{eegs6}
\end{align}
where
 \begin{align}
M_1^\infty(s)
= - \frac{i}{2} \begin{pmatrix}
 \int_s^\infty u(\zeta)^2d\zeta & -e^{-i\varphi_0 }u(s) \\ e^{i\varphi_0 }u(s) & -\int_s^\infty u(\zeta)^2d\zeta
 \end{pmatrix}. \label{eegs}
\end{align}

\vspace{2mm}

\noindent\textbf{Acknowledgements}

This work is supported by  the National Science
Foundation of China (Grant No. 12271104,  51879045).\vspace{2mm}

    \noindent\textbf{Data Availability Statements}

    The data which supports the findings of this study is available within the article.\vspace{2mm}

    \noindent{\bf Conflict of Interest}

    The authors have no conflicts to disclose.

\hspace*{\parindent}
\\

\end{document}